\documentclass[acmsmall,screen,10pt,authorversion]{acmart}

\usepackage{amsthm}                    %
\usepackage{mathtools}
\usepackage{microtype}
\usepackage[utf8]{inputenc}            %
\usepackage[T1]{fontenc}               %
\usepackage{mathpartir}                %
\usepackage{mathwidth}                 %
\usepackage{stmaryrd}                  %
\usepackage{xspace}                    %
\usepackage{hyperref}                  %
\usepackage{xcolor}
\usepackage{xifthen}
\usepackage[capitalise]{cleveref}
\usepackage{quiver}
\usetikzlibrary{fit,calc,trees,positioning,arrows,chains,shapes.geometric,%
    decorations.pathreplacing,decorations.pathmorphing,shapes,%
    matrix,shapes.symbols,intersections}
\usepackage[customcolors,shade]{hf-tikz}
\hfsetfillcolor{gray!50}
\hfsetbordercolor{gray!0}
\usepackage{enumitem}
\SetLabelAlign{spaceleft}{\quad #1\hfil}
\SetLabelAlign{bigspaceleft}{\qquad #1\hfil}
\setlist[description]{leftmargin=4.5em, style=nextline}
\usepackage{thmtools,thm-restate}      %
\usepackage{nccmath}                   %
\usepackage{pifont}                    %
\usepackage{trimclip}                  %
\usepackage[skins]{tcolorbox}

\usepackage[scaled=0.85]{beramono}   %
\usepackage{listings}
\makeatletter
\newcommand*{\halfblacktriangleright}{}%
\DeclareRobustCommand*{\halfblacktriangleright}{%
  \mathrel{%
    \mathpalette\@halfblacktriangleright{}%
  }%
}
\newcommand*{\@halfblacktriangleright}[2]{%
  \rlap{%
    \clipbox{0 0 {.7\width} 0}{$#1\blacktriangleright\m@th$}%
  }%
  \rhd
}
\newcommand*{\halfwhitetriangleright}{}%
\DeclareRobustCommand*{\halfwhitetriangleright}{%
  \mathbin{%
    \mathpalette\@halfwhitetriangleright{}%
  }%
}
\newcommand*{\@halfwhitetriangleright}[2]{%
  \rlap{%
    \clipbox{0 0 {.5\width} 0}{$#1\rhd\m@th$}%
  }%
  \blacktriangleright
}
\makeatother

\makeatletter
\@ifpackageloaded{stix}{%
}{%
  \DeclareFontEncoding{LS2}{}{\noaccents@}
  \DeclareFontSubstitution{LS2}{stix}{m}{n}
  \DeclareSymbolFont{stix@largesymbols}{LS2}{stixex}{m}{n}
  \SetSymbolFont{stix@largesymbols}{bold}{LS2}{stixex}{b}{n}
  \DeclareMathDelimiter{\lBrace}{\mathopen} {stix@largesymbols}{"E8}%
                                            {stix@largesymbols}{"0E}
  \DeclareMathDelimiter{\rBrace}{\mathclose}{stix@largesymbols}{"E9}%
                                            {stix@largesymbols}{"0F}
}
\makeatother

\newcommand{\LS}{\ensuremath{\mathcal{L}}}

\newcommand{\True}{\ensuremath{\mathsf{true}}}
\newcommand{\False}{\ensuremath{\mathsf{false}}}

\newcommand{\reducesto}{\ensuremath{\leadsto}}
\newcommand{\reducestocl}{\ensuremath{\rightsquigarrow}}
\newcommand{\breducesto}{\ensuremath{\reducesto_\beta}}
\newcommand{\ereducesto}{\ensuremath{\reducesto_\eta}}
\newcommand{\breducestocl}{\ensuremath{\reducestocl_\beta}}

\newcommand{\ureducesto}{\ensuremath{\reducesto_\Upcast}}
\newcommand{\ureducestocl}{\ensuremath{\reducestocl_\Upcast}}
\newcommand{\bureducesto}{\ensuremath{\reducesto_{\beta\Upcast}}}
\newcommand{\bureducestocl}{\ensuremath{\reducestocl_{\beta\Upcast}}}
\newcommand{\unewreducesto}{\ensuremath{\reducesto_\Upcastnew}}
\newcommand{\unewreducestocl}{\ensuremath{\reducestocl_\Upcastnew}}
\newcommand{\uunewreducesto}{\ensuremath{\reducesto_{\Upcast\Upcastnew}}}
\newcommand{\uunewreducestocl}{\ensuremath{\reducestocl_{\Upcast\Upcastnew}}}

\newcommand{\rd}[1]{\ensuremath{\reducesto_{#1}}}
\newcommand{\rdcl}[1]{\ensuremath{\reducestocl_{#1}}}

\newcommand{\vreducesto}{\rd{\nu}}
\newcommand{\vreducestocl}{\rdcl{\nu}}
\newcommand{\treducesto}{\rd{\tau}}
\newcommand{\treducestocl}{\rdcl{\tau}}

\newcommand{\slab}[1]{}
\newcommand{\rulelabel}[2]{\ensuremath{\mathsf{#1\textrm{-}#2}}}
\newcommand{\klab}[1]{\rulelabel{K}{#1}}

\newcommand{\tylab}[1]{\rulelabel{T}{#1}}

\newcommand{\ssublab}[1]{\rulelabel{S}{#1}}
\newcommand{\fsublab}[1]{\rulelabel{FS}{#1}}

\newcommand{\betalab}[1]{\rulelabel{\beta}{#1}}
\newcommand{\etalab}[1]{\rulelabel{\eta}{#1}}
\newcommand{\upcastlab}[1]{\rulelabel{\Upcast\!\!}{#1}}
\newcommand{\upcastnewlab}[1]{\rulelabel{\Upcastnew\!\!}{#1}}

\newcommand{\tbetalab}[1]{\rulelabel{\tau}{#1}}

\newcommand{\tlab}[1]{\rulelabel{\tau}{#1}}
\newcommand{\vlab}[1]{\rulelabel{\nu}{#1}}

\newcommand{\CatName}[1]{\mathsf{#1}}
\newcommand{\Terms}{\CatName{Term}}
\newcommand{\Types}{\CatName{Type}}
\newcommand{\TyScheme}{\CatName{TypeScheme}}

\newcommand{\Envs}{\CatName{Env}}
\newcommand{\TyEnvs}{\CatName{TyEnv}}

\newcommand{\Labels}{\CatName{Label}}
\newcommand{\Kinds}{\CatName{Kind}}
\newcommand{\Rows}{\CatName{Row}}
\newcommand{\Presences}{\CatName{Presence}}
\newcommand{\RowVars}{\CatName{RowVar}}
\newcommand{\Deriv}{\CatName{Derivation}}
\newcommand{\Belong}[2]{#2 \ni #1}

\newcommand{\ba}{\begin{array}}
\newcommand{\ea}{\end{array}}

\newenvironment{syntax}{\[\ba{@{}l@{}r@{~}c@{~}l@{}}}{\ea\]\ignorespacesafterend}
\newenvironment{reductions}{\[\ba{@{}l@{\qquad}@{}r@{~~}c@{~~}l@{}}}{\ea\]\ignorespacesafterend}

\newenvironment{equations}{\[\ba{@{}r@{~}c@{~}l@{}}}{\ea\]\ignorespacesafterend}

\newcommand{\transl}[1]{\ensuremath{\llbracket #1 \rrbracket}}
\newcommand{\transpre}[2]{\ensuremath{\llparenthesis #1, #2 \rrparenthesis}}
\newcommand{\transla}[1]{\ensuremath{\llbracket #1 \rrbracket}}
\newcommand{\translb}[1]{\ensuremath{\llbracket #1 \rrbracket^\ast}}
\newcommand{\transrowa}[2]{\ensuremath{\llparenthesis #1, #2 \rrparenthesis}}
\newcommand{\transrowb}[2]{\ensuremath{\llparenthesis #1, #2 \rrparenthesis^\ast}}
\newcommand{\transrowsub}[1]{\ensuremath{\llparenthesis #1 \rrparenthesis}}

\newcommand{\keyw}[1]{\mathbf{#1}}
\newcommand{\typ}[2]{#1 \vdash #2}

\newcommand{\dec}[1]{\ensuremath{\mathsf{#1}}}
\newcommand{\ftv}[1]{\mathsf{ftv}(#1)}

\newcommand{\var}[1]{\mathtt{#1}}

\newcommand{\Int}{\mathsf{Int}}

\newcommand{\String}{\mathsf{String}}

\newcommand{\Case}{\keyw{case}}

\newcommand{\Let}{\keyw{let}}
\newcommand{\In}{\keyw{in}}

\newcommand{\Variant}[1]{\ensuremath{[ #1 ]}}
\newcommand{\Record}[1]{\ensuremath{\langle #1 \rangle}}
\newcommand{\Variantonly}{\text{\scalebox{1}[.8]{$[]$}}}
\newcommand{\Recordonly}{\text{\scalebox{1}[.8]{$\langle\rangle$}}}

\newcommand{\Upcast}{\ensuremath{\vartriangleright}}
\newcommand{\Upcastnew}{\ensuremath{\blacktriangleright}}

\newcommand{\subtype}{\ensuremath{\leqslant}}
\newcommand{\subtyperow}{\ensuremath{\preccurlyeq}}
\newcommand{\equivsub}{\ensuremath{\sqsubseteq}}

\newcommand{\Present}{\ensuremath{\bullet}}
\newcommand{\Absent}{\ensuremath{\circ}}

\newcommand{\Presence}{\ensuremath{\mathsf{Pre}}}
\newcommand{\Type}{\ensuremath{\mathsf{Type}}}
\newcommand{\Row}{\ensuremath{\mathsf{Row}}}

\newcommand{\labn}[1]{\dec{#1}}
\newcommand{\Name}{\labn{Name}}
\newcommand{\Age}{\labn{Age}}
\newcommand{\Year}{\labn{Year}}
\newcommand{\Child}{\labn{Child}}
\newcommand{\Raw}{\labn{Raw}}

\newcommand{\makestring}[1]{\text{\textquotedbl #1\textquotedbl}}
\newcommand{\xmark}{\text{\ding{55}}}

\newenvironment{examples}{
  \fleqn
  \begin{equations}
  \ba{rll}
}{
  \ea
  \end{equations}
  \endfleqn
}

\newcommand{\localtermarr}[2]{
  \begin{tikzcd}[cramped, ampersand replacement=\&, column sep=small]
    #1
    \arrow[draw={rgb,255:red,92;green,92;blue,214}, r, two heads]
    \& #2
  \end{tikzcd}
}
\newcommand{\localtypearr}[2]{
  \begin{tikzcd}[cramped, ampersand replacement=\&, column sep=small]
    #1
    \arrow[r]
    \& #2
  \end{tikzcd}
}
\newcommand{\globaltypearr}[2]{
  \begin{tikzcd}[cramped, ampersand replacement=\&, column sep=small]
    #1
    \arrow[draw={rgb,255:red,214;green,92;blue,214}, r, maps to]
    \& #2
  \end{tikzcd}
}
\newcommand{\nonexistarr}[2]{
  \begin{tikzcd}[cramped, ampersand replacement=\&, column sep=scriptsize]
    #1
    \arrow[draw={rgb,255:red,214;green,92;blue,92}, r, squiggly]
    \& #2
  \end{tikzcd}
}

\makeatletter
    \newcommand{\colorboxed}[3][white]{\fcolorbox{#2}{#1}{\m@th$\displaystyle#3$}}
\makeatother

\newcommand{\smod}{\mathrm}

\newcommand{\stlc}[2]{\ensuremath{\lambda_{#1}^{#2}}\xspace}
\newcommand{\STLC}{\stlc{}{}}
\newcommand{\STLCVar}{\stlc{\Variantonly}{}}
\newcommand{\STLCRec}{\stlc{\Recordonly}{}}
\newcommand{\STLCVarSub}{\stlc{\Variantonly}{\subtype}}
\newcommand{\STLCVarSubFull}{\stlc{\Variantonly}{\subtype \smod{full}}}
\newcommand{\STLCVarSubFullCond}[1]{\stlc{\Variantonly #1}{\subtype \smod{full}}}
\newcommand{\STLCVarSubCo}{\stlc{\Variantonly}{\subtype \smod{co}}}
\newcommand{\STLCVarRow}{\stlc{\Variantonly}{\rho}}
\newcommand{\STLCVarPre}{\stlc{\Variantonly}{\theta}}
\newcommand{\STLCVarRowPre}{\stlc{\Variantonly}{\rho\theta}}
\newcommand{\STLCRecSub}{\stlc{\Recordonly}{\subtype}}
\newcommand{\STLCRecSubFull}{\stlc{\Recordonly}{\subtype \smod{full}}}
\newcommand{\STLCRecSubFullA}{\stlc{\Recordonly}{\subtype \smod{afull}}}
\newcommand{\STLCRecSubFullCond}[1]{\stlc{\Recordonly #1}{\subtype \smod{full}}}
\newcommand{\STLCRecSubFullACond}[1]{\stlc{\Recordonly #1}{\subtype \smod{afull}}}
\newcommand{\STLCRecSubCo}{\stlc{\Recordonly}{\subtype \smod{co}}}
\newcommand{\STLCRecPre}{\stlc{\Recordonly}{\theta}}
\newcommand{\STLCRecRow}{\stlc{\Recordonly}{\rho}}
\newcommand{\STLCRecRowPre}{\stlc{\Recordonly}{\rho\theta}}
\newcommand{\STLCVarRecSubFull}{\stlc{\Variantonly\Recordonly}{\subtype \smod{full}}}
\newcommand{\STLCVarRecSubCo}{\stlc{\Variantonly\Recordonly}{\subtype \smod{co}}}
\newcommand{\STLCVarRec}{\stlc{\Variantonly\Recordonly}{}}
\newcommand{\STLCRecPrenex}{\stlc{\Recordonly}{\rho 1}}
\newcommand{\STLCRecPrePrenex}{\stlc{\Recordonly}{\theta 1}}
\newcommand{\STLCVarPrenex}{\stlc{\Variantonly}{\rho 1}}
\newcommand{\STLCVarPrePrenex}{\stlc{\Variantonly}{\theta 1}}

\newcommand{\ol}[1]{\overline{#1}}  %
\newcommand{\hk}[1]{@\,#1} %
\newcommand{\que}{?}
\newcommand{\dom}[1]{\mathsf{dom}(#1)}
\newcommand{\recrank}[2]{\mho^{#1}(#2)}
\newcommand{\varrank}[2]{\Omega^{#1}(#2)}
\newcommand{\nocontratwice}[1]{\mho^2(#1)}
\newcommand{\nocontraonce}[1]{\mho^1(#1)}

\newcommand{\erase}[1]{\mathsf{erase}(#1)}

\colorlet{varcolor}{lightgray!50}
\colorlet{reccolor}{violet}
\newcommand{\hlvar}[1]{\colorbox{varcolor}{$#1$}_{\Variant{}}}
\newcommand{\hlrec}[1]{\colorbox{varcolor}{$#1$}_{\Record{}}}
\newcommand{\hlvarrec}[1]{\colorbox{varcolor}{$#1$}_{\Variant{}\Record{}}}
\newcommand{\hlvsub}[1]{\colorbox{varcolor}{$#1$}_{\Variant{}}}
\newcommand{\hlrsub}[1]{\colorbox{varcolor}{$#1$}_{\Record{}}}
\newcommand{\hlvrsub}[1]{\colorbox{varcolor}{$#1$}_{\Variant{}\Record{}}}

\newcommand{\hlmod}[1]{\colorbox{varcolor}{$#1$}}

\newcommand{\refa}[1]{{\color{blue}   (1){#1}}}
\newcommand{\refb}[1]{{\color{red}    (2){#1}}}
\newcommand{\refc}[1]{{\color{orange} (3){#1}}}
\newcommand{\refd}[1]{{\color{purple} (4){#1}}}

\AtBeginDocument{%
  \providecommand\BibTeX{{%
      \normalfont B\kern-0.5em{\scshape i\kern-0.25em b}\kern-0.8em\TeX}}}

\setcopyright{rightsretained}
\acmPrice{}
\acmDOI{10.1145/3622836}
\acmYear{2023}
\copyrightyear{2023}
\acmSubmissionID{oopslab23main-p306-p}
\acmJournal{PACMPL}
\acmVolume{7}
\acmNumber{OOPSLA2}
\acmArticle{260}
\acmMonth{10}
\received{2023-04-14}
\received[accepted]{2023-08-27}

\begin{document}

\title{Structural Subtyping as Parametric Polymorphism}

\author{Wenhao Tang}
\orcid{0009-0000-6589-3821}
\affiliation{%
  \institution{The University of Edinburgh}
  \country{UK}}
\email{wenhao.tang@ed.ac.uk}

\author{Daniel Hillerstr{\"o}m}
\orcid{0000-0003-4730-9315}
\affiliation{%
  \institution{Huawei Zurich Research Center}
  \country{Switzerland}}
\email{daniel.hillerstrom@ed.ac.uk}

\author{James McKinna}
\orcid{0000-0001-6745-2560}
\affiliation{%
  \institution{Heriot-Watt University}
  \country{UK}}
\email{j.mckinna@hw.ac.uk}

\author{Michel Steuwer}
\orcid{0000-0001-5048-0741}
\affiliation{%
  \institution{Technische Universität Berlin}
  \country{Germany}}
\affiliation{%
  \institution{The University of Edinburgh}
  \country{UK}}
\email{michel.steuwer@tu-berlin.de}

\author{Ornela Dardha}
\orcid{0000-0001-9927-7875}
\affiliation{%
  \institution{University of Glasgow}
  \country{UK}}
\email{ornela.dardha@glasgow.ac.uk}

\author{Rongxiao Fu}
\orcid{0009-0005-6966-4037}
\affiliation{%
  \institution{The University of Edinburgh}
  \country{UK}}
\email{s1742701@sms.ed.ac.uk}

\author{Sam Lindley}
\orcid{0000-0002-1360-4714}
\affiliation{%
  \institution{The University of Edinburgh}
  \country{UK}}
\email{sam.lindley@ed.ac.uk}

\renewcommand{\shortauthors}{W. Tang, D. Hillerstr{\"o}m, J. McKinna, M. Steuwer, O. Dardha, R. Fu, and S. Lindley}

\begin{abstract}
Structural subtyping and parametric polymorphism provide similar
flexibility and reusability to programmers. For example, both features enable
the programmer to provide a wider record as an argument to a function
that expects a narrower one.
However, the means by which they do so differs substantially, and the
precise details of the relationship between them exists, at best, as
folklore in literature.

In this paper, we systematically study the relative expressive power
of structural subtyping and parametric polymorphism.
We focus our investigation on establishing the extent to which
parametric polymorphism, in the form of row and presence polymorphism,
can encode structural subtyping for variant and record types.
We base our study on various Church-style $\lambda$-calculi extended
with records and variants, different forms of structural subtyping,
and row and presence polymorphism.

We characterise expressiveness by exhibiting compositional
translations between calculi. For each translation we prove a type
preservation and operational correspondence result.
We also prove a number of non-existence results.
By imposing restrictions on both source and target types, we reveal
further subtleties in the expressiveness landscape, the restrictions
enabling otherwise impossible translations to be defined.
More specifically, we prove that full subtyping cannot be encoded via
polymorphism, but we show that several restricted forms of subtyping
can be encoded via particular forms of polymorphism.

\end{abstract}



\maketitle

\section{Introduction}
\label{sec:introduction}

Subtyping and parametric polymorphism offer two distinct means for
writing modular and reusable code. Subtyping
allows one value to be substituted for another provided that the type
of the former is a subtype of that of the
latter~\citep{Reynolds80,Cardelli88}. Parametric polymorphism allows
functions to be defined generically over arbitrary
types~\citep{Girard72,Reynolds74}.

There are two main approaches to \emph{syntactic} subtyping: nominal
subtyping~\citep{BirtwistleDMN79} and structural
subtyping~\citep{CardelliW85,Cardelli88,Cardelli84}. The former
defines a subtyping relation as a collection of explicit constraints
between named types. The latter defines a subtyping relation
inductively over the structure of types.
This paper is concerned with the latter.
For programming languages with variant types (constructor-labelled
sums) and record types (field-labelled products) it is natural
to define a notion of structural subtyping.
We may always treat a variant with a collection of constructors as a
variant with an \emph{extended} collection of constructors (i.e.,
variant subtyping is covariant).
Dually, we may treat a record with a collection of fields as a record
with a \emph{restricted} collection of those fields (i.e., record
subtyping is contravariant).

We can implement similar functionality to record and variant subtyping
using \emph{row polymorphism}~\citep{Wand87, Remy94}.
A \emph{row} is a mapping from labels to types and is thus a common
ingredient for defining both variants and records.
Row polymorphism is a form of parametric polymorphism that allows us
to abstract over the extension of a row.
Intuitively, by abstracting over the possible extension of a variant
or record we can simulate the act of substitution realised by
structural subtyping.
Such intuitions are folklore, but pinning them down turns out to be
surprisingly subtle.
In this paper we make them precise by way of translations between a
series of different core calculi enjoying type preservation and
operational correspondence results as well as non-existence results.
We show that though folklore intuitions are to some extent correct,
exactly how they manifest in practice is remarkably dependent on what
assumptions we make, and much more nuanced than we anticipated.
We believe that our results are not just of theoretical interest.
It is important to carefully analyse and characterise the relative
expressive power of different but related features to understand the
extent to which they overlap.

To be clear, there is plenty of other work that hinges on inducing a
subtyping relation based on generalisation (i.e., polymorphism) --- and
indeed this is the basis for principal types in Hindley-Milner type
inference --- but this paper is about something quite different,
namely encoding prior notions of structural subtyping using
polymorphism. In short, principal types concern polymorphism as
subtyping, whereas this paper concerns subtyping as polymorphism.

In order to distil the features we are interested in down to their
essence and eliminate the interference on the expressive power of
other language features (such as higher-order store), we take plain
\citeauthor{Church40}-style call-by-name simply-typed
$\lambda$-calculus (\STLC) as our starting point and consider the
relative expressive power of minimal extensions in turn.
We begin by insisting on writing explicit upcasts, type abstractions,
and type applications in order to expose structural subtyping and
parametric polymorphism at the term level.
Later we also consider ML-style calculi, enabling new expressiveness
results by exploiting the type inference for rank-1 polymorphism.
For the dynamic semantics, we focus on the reduction theory generated
from the $\beta$-rules, adding further $\beta$-rules for each term
constructor and upcast rules for witnessing subtyping.

First we extend the simply-typed $\lambda$-calculus with variants
(\STLCVar), which we then further augment with \emph{simple subtyping}
(\STLCVarSub) that only considers the subtyping relation shallowly on
variant and record constructors (width subtyping), and (higher-rank)
row polymorphism (\STLCVarRow), respectively.
Dually, we extend the simply-typed $\lambda$-calculus with records
(\STLCRec), which we then further augment with simple subtyping
(\STLCRecSub) and (higher-rank) \emph{presence polymorphism}
(\STLCRecPre), respectively.
Presence polymorphism~\citep{Remy94} is a kind of dual to row
polymorphism that allows us to abstract over which fields are present
or absent from a record independently of their potential types,
supporting a restriction of a collection of record fields, similarly
to record subtyping.
We then consider richer extensions with strictly covariant subtyping
(\STLCVarSubCo, \STLCRecSubCo), which propagates the subtyping relation
through strictly covariant positions, and full subtyping
(\STLCVarSubFull, \STLCRecSubFull), which propagates the subtyping
relation through any positions.
We also consider target languages with both row and presence
polymorphism (\STLCVarRowPre, \STLCRecRowPre).
Our initial investigations make essential use of higher-rank
polymorphism.
Subsequently, we consider ML-like calculi with rank-1 row or presence
polymorphism (\STLCRecPrenex, \STLCRecPrePrenex, \STLCVarPrenex,
\STLCVarPrePrenex), which admit Hindley-Milner type inference
\citep{DamasMilner} without requirements of type annotations or
explicit type abstractions and applications.
The focus on rank-1 polymorphism demands a similar restriction to the
calculi with subtyping (\STLCRecSubFullCond{1},
\STLCRecSubFullCond{2}, \STLCVarSubFullCond{1},
\STLCVarSubFullCond{2}), which constrains the positions where records
and variants can appear in types.

In this paper, we will consider only correspondences expressed as
\emph{compositional translations} inductively defined on language
constructs following \citet{Felleisen91}.
In order to give a refined characterisation of expressiveness and
usability of the type systems of different calculi, we make use of two
orthogonal notions of \emph{local} and \emph{type-only} translations.
\begin{itemize}
  \item A \emph{local} translation restricts which features are
    translated in a non-trivial way. It provides non-trivial
    translations only of constructs of interest (e.g., record types,
    record construction and destruction, when considering record
    subtyping), and is homomorphic on other constructs; a
    \emph{global} translation may allow any construct to have a
    non-trivial translation.
  \item A \emph{type-only} translation restricts which features a
    translation can use in the target language. Every term must
    translate to itself modulo constructs that serve only to
    manipulate types (e.g., type abstraction and application); a
    \emph{term-involved} translation has no such restriction.
\end{itemize}
Local translations capture the intuition that a feature can be
expressed locally as a macro rather than having to be implemented by
globally changing an entire program~\cite{Felleisen91}.
Type-only translations capture the intuition that a feature can be
expressed solely by adding or removing type manipulation operations
(such as upcasts, type abstraction, and type application) in terms,
thereby enabling a more precise comparison between the expressiveness
of different type system features.

This paper gives a \emph{precise account of the relationship between
subtyping and polymorphism for records and variants}.
We present relative expressiveness results by way of a series of
translations between calculi, type preservation proofs,
operational correspondence proofs, and non-existence proofs.
\begin{figure}

\newcommand{\sect}{\S}

\[\begin{tikzcd}
	& \STLCVarPre && {\STLCVarSubFullCond{2}} & \STLCVarPrePrenex & {} &&& {} \\
	\STLCVar & \STLCVarSub & \STLCVarSubCo & \STLCVarSubFull & \STLCVarRowPre & {} &&& {} \\
	& \STLCVarRow && {\STLCVarSubFullCond{1}} & \STLCVarPrenex & {} &&& {} \\[-1.35em]
	\STLC \\[-1.35em]
	& \STLCRecPre && {\STLCRecSubFullCond{1}} & \STLCRecPrePrenex & {} &&& {} \\
	\STLCRec & \STLCRecSub & \STLCRecSubCo & \STLCRecSubFull & \STLCRecRowPre & {} &&& {} \\
	& \STLCRecRow && {\STLCRecSubFullCond{2}} & \STLCRecPrenex & {} &&& {}
	\arrow[dotted, no head, from=2-1, to=2-2]
	\arrow[curve={height=6pt}, dotted, no head, from=2-1, to=3-2]
	\arrow["{\sect\ref{sec:encode-stlcvarsub-stlcvarrow}}"', from=2-2, to=3-2]
	\arrow["{\sect\ref{sec:encode-stlcvarsub-stlcvar}}", shift left=2, draw={rgb,255:red,92;green,92;blue,214}, two heads, from=2-2, to=2-1]
	\arrow[dotted, no head, from=6-1, to=6-2]
	\arrow[curve={height=-6pt}, dotted, no head, from=6-1, to=5-2]
	\arrow["{\sect\ref{sec:encode-stlcrecsub-stlcrec}}"', shift right=2, draw={rgb,255:red,92;green,92;blue,214}, two heads, from=6-2, to=6-1]
	\arrow["{\sect\ref{sec:encode-stlcrecsub-stlcrecpre}}", from=6-2, to=5-2]
	\arrow[curve={height=6pt}, dotted, no head, from=6-1, to=7-2]
	\arrow["{\sect\ref{sec:swapping-row-and-pre}}"', draw={rgb,255:red,214;green,92;blue,92}, squiggly, from=6-2, to=7-2]
	\arrow["{\sect\ref{sec:swapping-row-and-pre}}", draw={rgb,255:red,214;green,92;blue,92}, squiggly, from=2-2, to=1-2]
	\arrow[curve={height=-6pt}, dotted, no head, from=2-1, to=1-2]
	\arrow[shift right=1, curve={height=6pt}, dotted, no head, from=7-2, to=6-5]
	\arrow["{\sect\ref{sec:non-full-subtyping}}", draw={rgb,255:red,214;green,92;blue,92}, squiggly, from=6-4, to=6-5]
	\arrow[shift left=1, curve={height=-6pt}, dotted, no head, from=1-2, to=2-5]
	\arrow[shift right=1, curve={height=6pt}, dotted, no head, from=3-2, to=2-5]
	\arrow["{\sect\ref{sec:encode-co-subtyping}}"', draw={rgb,255:red,214;green,92;blue,214}, maps to, from=6-3, to=5-2]
	\arrow[dotted, no head, from=6-2, to=6-3]
	\arrow[dotted, no head, from=6-3, to=6-4]
	\arrow[dotted, no head, from=2-2, to=2-3]
	\arrow[dotted, no head, from=2-3, to=2-4]
	\arrow[shift left=1, curve={height=-6pt}, dotted, no head, from=5-2, to=6-5]
	\arrow["{\sect\ref{sec:encode-full-subtyping}}"', shift left=1, draw={rgb,255:red,92;green,92;blue,214}, curve={height=18pt}, two heads, from=2-4, to=2-1]
	\arrow["{\sect\ref{sec:encode-full-subtyping}}", shift right=1, draw={rgb,255:red,92;green,92;blue,214}, curve={height=-18pt}, two heads, from=6-4, to=6-1]
	\arrow["{\sect \ref{sec:non-stlcvarsubco-stlcvarrowpre}}"'{pos=0.4}, shift left=1, draw={rgb,255:red,214;green,92;blue,92}, curve={height=12pt}, squiggly, from=2-3, to=2-5]
	\arrow[dotted, no head, from=4-1, to=6-1]
	\arrow[dotted, no head, from=4-1, to=2-1]
	\arrow["{\text{extension}}", dotted, no head, from=1-6, to=1-9]
	\arrow["{\text{local type-only}}", from=3-6, to=3-9]
	\arrow["{\text{local term-involved}}", color={rgb,255:red,92;green,92;blue,214}, two heads, from=5-6, to=5-9]
	\arrow["{\text{global type-only}}", color={rgb,255:red,214;green,92;blue,214}, maps to, from=6-6, to=6-9]
	\arrow["{\text{non-existence of type-only}}", color={rgb,255:red,214;green,92;blue,92}, squiggly, from=7-6, to=7-9]
	\arrow["{\sect\ref{sec:prenex-polymorphism}}"', from=7-4, to=7-5]
	\arrow[draw={rgb,255:red,89;green,89;blue,89}, dashed, no head, from=6-4, to=7-4]
	\arrow["{\sect\ref{sec:prenex-polymorphism}}", from=5-4, to=5-5]
	\arrow[draw={rgb,255:red,89;green,89;blue,89}, dashed, no head, from=6-4, to=5-4]
	\arrow["{\sect\ref{sec:prenex-polymorphism}}", from=1-4, to=1-5]
	\arrow[draw={rgb,255:red,89;green,89;blue,89}, dashed, no head, from=2-4, to=1-4]
	\arrow["{\sect\ref{sec:prenex-polymorphism}}"', from=3-4, to=3-5]
	\arrow[draw={rgb,255:red,89;green,89;blue,89}, dashed, no head, from=2-4, to=3-4]
	\arrow[shift left=2, draw={rgb,255:red,89;green,89;blue,89}, curve={height=-12pt}, dashed, no head, from=5-2, to=5-5]
	\arrow[shift right=2, draw={rgb,255:red,89;green,89;blue,89}, curve={height=12pt}, dashed, no head, from=7-2, to=7-5]
	\arrow[shift left=2, draw={rgb,255:red,89;green,89;blue,89}, curve={height=-12pt}, dashed, no head, from=1-2, to=1-5]
	\arrow[shift right=2, draw={rgb,255:red,89;green,89;blue,89}, curve={height=12pt}, dashed, no head, from=3-2, to=3-5]
	\arrow["{\text{\textcolor{black}{restriction}}}", color={rgb,255:red,89;green,89;blue,89}, dashed, no head, from=2-6, to=2-9]
\end{tikzcd}\]

{\footnotesize
  Extensions and restrictions go from calculi with
  shorter names to those with longer names\\(e.g. $\STLCVar$ extends
  $\STLC$ and $\STLCVarPrePrenex$ restricts $\STLCVarPre$).
}

  \caption{Overview of translations and non-existence results covered
    in the paper.}
  \label{fig:results-summary}
\end{figure}
The main contributions of the paper (summarised in
\Cref{fig:results-summary}) are as follows.

\begin{itemize}
  \item We present a collection of examples in order to convey the
    intuition behind all translations and non-existence results in
    \Cref{fig:results-summary} (\Cref{sec:examples}).
  \item We define a family of Church-style calculi extending
    $\lambda$-calculus with variants and records, simple subtyping,
    and (higher-rank) row or presence polymorphism
    (Section~\ref{sec:calculi}).
  \item
    We prove that simple subtyping can be elaborated away for variants
    and records by way of local term-involved translations
    (Sections~\ref{sec:encode-stlcvarsub-stlcvar}~and~\ref{sec:encode-stlcrecsub-stlcrec}).
  \item
    We prove that simple subtyping can be expressed as row
    polymorphism for variants and presence polymorphism for records by
    way of local type-only translations
    (Sections~\ref{sec:encode-stlcvarsub-stlcvarrow}~and~\ref{sec:encode-stlcrecsub-stlcrecpre}).
  \item
    We prove that there exists no type-only translation of simple
    subtyping into presence polymorphism for variants or row
    polymorphism for records (\Cref{sec:swapping-row-and-pre}).
  \item
    We expand our study to calculi with covariant and full subtyping
    and with both row- and presence-polymorphism, covering further
    translations and non-existence proofs
    (Section~\ref{sec:full-subtyping}). In so doing we reveal a
    fundamental asymmetry between variants and records.
  \item
    We prove that if we suitably restrict types and switch to
    ML-style target calculi with implicit rank-1 polymorphism, then we
    can exploit type inference to encode full subtyping for records
    and variants using either row polymorphism or presence
    polymorphism (\Cref{sec:prenex-polymorphism}).

  \item
    For each translation we prove type preservation and operational
    correspondence results.
\end{itemize}
Sections~\ref{sec:record-extension}~and~\ref{sec:combine-sub-and-poly}
discuss extensions. %
Section~\ref{sec:related-work} discusses related work.
Section~\ref{sec:conclusion} concludes.

\section{Examples}
\label{sec:examples}

To illustrate the relative expressive power of subtyping and
polymorphism for variants and records with a range of extensions, we
give a collection of examples. These cover the intuition behind the
translations and non-existence results summarised in
\Cref{fig:results-summary} and formalised later in the paper.

\subsection{Simple Variant Subtyping as Row Polymorphism}
\label{sec:example-simple-variant-subtyping}
We begin with variant types.
Consider the following function.
\begin{examples}
  &\var{getAge} = \lambda x^{\Variant{\Age:\Int; \Year:\Int}} .~
    \Case\ x\ \{\Age\ y \mapsto y; \Year\ y \mapsto 2023-y\} \\
\end{examples}
The variant type $\Variant{\Age:\Int; \Year:\Int}$ denotes the type of
variants with two constructors $\Age$ and $\Year$ each containing an
$\Int$.
We cannot directly apply $\var{getAge}$ to the following variant
\begin{examples}
  &\var{year} = (\Year\ 1984)^{\Variant{\Year:\Int}} \\
\end{examples}
as $\var{year}$ and $x$ have different types.
With simple variant subtyping (\STLCVarSub) which considers subtyping
shallowly on variants, we can upcast $\var{year} :
\Variant{\Year:\Int}$ to the supertype $\Variant{\Age:\Int;
\Year:\Int}$ which has more labels.
This makes intuitive sense, as it is always safe to treat a variant
with fewer constructors ($\Year$ in this case) as one with more
constructors ($\Age$ and $\Year$ in this case).

\begin{examples}
  &\var{getAge}\ (\var{year} \Upcast \Variant{\Age:\Int; \Year:\Int}) \\
\end{examples}
One advantage of subtyping is reusability: by upcasting we can apply
the same $\var{getAge}$ function to any value whose type is a subtype
of $\Variant{\Age:\Int; \Year:\Int}$.

\begin{examples}
  &\var{age} = (\Age\ 9)^{\Variant{\Age:\Int}} \\
  &\var{getAge}\ (\var{age} \Upcast \Variant{\Age:\Int; \Year:\Int}) \\
\end{examples}

In a language without subtyping (\STLCVar), we can simulate applying
$\var{getAge}$ to $\var{year}$ by first deconstructing the variant
using $\Case$ and then reconstructing it at the appropriate type --- a
kind of generalised $\eta$-expansion on variants.
\begin{examples}
  &\var{getAge}\ (\Case\ \var{year}\
                  \{\Year\ y \mapsto (\Year\ y)^{\Variant{\Age:\Int; \Year:\Int}}\})
\end{examples}
This is the essence of the translation
\localtermarr{\STLCVarSub}{\STLCVar} in
\cref{sec:encode-stlcvarsub-stlcvar}.
The translation is \emph{local} in the sense that it only requires us
to transform the parts of the program that relate to variants (as
opposed to the entire program).
However, it still comes at a cost. The deconstruction and
reconstruction of variants adds extra computation that was not present
in the original program.

Can we achieve the same expressive power of subtyping without
non-trivial term de- and re-construction?
Yes we can! Row polymorphism (\STLCVarRow) allows us to rewrite
$\var{year}$ with a type compatible (via row-variable substitution)
with any variant type containing $\Year:\Int$ and additional cases.
\footnote{We omit the kinds of row variables for
simplicity. They can be easily reconstructed from the contexts.}
\begin{examples}
  &\var{year'} = \Lambda \rho.~(\Year\ 1984)^{\Variant{\Year:\Int;\rho}} \\
\end{examples}
As before, the translation to $\var{year'}$ also adds new term syntax.
However, the only additional syntax required by this translation
involves type abstraction and type application; in other words the
program is unchanged up to type erasure.
Thus we categorise it as a \emph{type-only} translation as opposed to
the previous one which we say is \emph{term-involved}.
We can instantiate $\rho$ with $(\Age:\Int)$ when applying
$\var{getAge}$ to it.
The parameter type of $\var{getAge}$ must also be translated to a
row-polymorphic type, which requires higher-rank
polymorphism. Moreover, we re-abstract over $\var{year'}$ after
instantiation to make it polymorphic again.
\begin{examples}
  &\var{getAge'} = \lambda x^{\forall \rho. \Variant{\Age:\Int; \Year:\Int; \rho}} .~
    \Case\ (x\ \cdot)\ \{\Age\ y \mapsto y; \Year\ y \mapsto 2023-y\} \\
  &\var{getAge'}\ (\Lambda \rho.~\var{year'}\ (\Age:\Int; \rho))
\end{examples}
The type application $x\ \cdot$ instantiates $\rho$ with the empty
closed row type $\cdot$.
The above function application is well-typed because we ignore the
order of labels when comparing rows ($\Age:\Int;\Year:\Int;\rho \equiv
\Year:\Int;\Age:\Int;\rho$) as usual.
This is the essence of the local type-only translation
\localtypearr{\STLCVarSub}{\STLCVarRow} in
\cref{sec:encode-stlcvarsub-stlcvarrow}.

We are relying on higher-rank polymorphism here in order to simulate
upcasting on demand.
For instance, an upcast on the parameter of a function of type
$(\forall \rho.\Variant{\Age:\Int; \Year:\Int; \rho}) \to B$ is
simulated by instantiating $\rho$ appropriately.
We will show in \Cref{sec:example-prenex} that restricting the target
language to rank-1 polymorphism requires certain constraints on the
source language.

\subsection{Simple Record Subtyping as Presence Polymorphism}
\label{sec:example-simple-record-subtyping}
Now, we consider record types, through the following function.
\begin{examples}
  &\var{getName} = \lambda x^{\Record{\Name:\String}}.~(x.\Name)
\end{examples}
The record type $\Record{\Name:\String}$ denotes the type of records
with a single field $\Name$ containing a string. We cannot directly
apply $\var{getName}$ to the following record
\begin{examples}
  &\var{alice} = \Record{\Name = \makestring{Alice}; \Age = 9}
\end{examples}
as the types of $\var{alice}$ and $x$ do not match. With simple record
subtyping (\STLCRecSub), we can upcast $\var{alice} :
\Record{\Name:\String;\Age:\Int}$ to the supertype
$\Record{\Name:\String}$.
It is intuitive to treat a record with more fields ($\Name$ and
$\Age$) as a record with fewer fields (only $\Name$ in this case).
\begin{examples}
  &\var{getName}\ (\var{alice}\Upcast \Record{\Name:\String})
\end{examples}
Similarly to variant subtyping, we can reuse $\var{getName}$ on
records of different subtypes.

\begin{examples}
  &\var{bob} = \Record{\Name = \makestring{Bob}; \Year = 1984} \\
  &\var{getName}\ (\var{bob}\Upcast \Record{\Name:\String})
\end{examples}

In a language without subtyping (\STLCRec), we can first deconstruct
the record by projection and then reconstruct it with only the
required fields, similarly to the generalised $\eta$-expansion of
records.
\begin{examples}
  &\var{getName}\ \Record{\Name = \var{alice}.\Name}
\end{examples}
This is the essence of the local term-involved translation
\localtermarr{\STLCRecSub}{\STLCRec} in
\cref{sec:encode-stlcrecsub-stlcrec}.
Using presence polymorphism (\STLCRecPre), we can simulate
$\var{alice}$ using a type-only translation.
\begin{examples}
  &\var{alice'} = \Lambda \theta_1 \theta_2 .~
                 \Record{\Name = \makestring{Alice}; \Age = 9}^{
                    \Record{\Name^{\theta_1}:\String; \Age^{\theta_2}:\Int}} \\
\end{examples}
The presence variables $\theta_1$ and $\theta_2$ can be substituted
with a marker indicating that the label is either present $\Present$
or absent $\Absent$. We can instantiate $\theta_2$ with absent
$\Absent$ when applying $\var{getName}$ to it, ignoring the $\Age$
label.
This resolves the type mismatch as the equivalence relation on row
types considers only present labels ($\Name^\theta:\String ~\equiv~
\Name^\theta:\String; \Age^\Absent:\Int$).
For a general translation, we must make the parameter type of
$\var{getName}$ presence-polymorphic, and re-abstract over
$\var{alice'}$.
\begin{examples}
  &\var{getName'} = \lambda x^{\forall \theta.\Record{\Name^\theta:\String}}.~((x\ \Present).\Name) \\
  &\var{getName'}\ (\Lambda \theta .~ \var{alice'}\ \theta\ \Absent)
\end{examples}
This is the essence of the local type-only translation
\localtypearr{\STLCRecSub}{\STLCRecPre} in
\cref{sec:encode-stlcrecsub-stlcrecpre}.
The duality between variants and records is reflected by the need for
dual kinds of polymorphism, namely row and presence
polymorphism, which can extend or shrink rows, respectively.

\subsection{Exploiting Contravariance}
\label{sec:example-swapping}

We have now seen how to encode simple variant subtyping as row
polymorphism and simple record subtyping as presence
polymorphism. These encodings embody the intuition that row
polymorphism supports extending rows and presence polymorphism
supports shrinking rows.
However, presence polymorphism is typically treated as an optional
extra for row typing. For instance, \citet{Remy94} uses row
polymorphism for both record and variant types, and introduces
presence polymorphism only to support record extension and default
cases (which fall outside the scope of our current investigation).

This naturally raises the question of whether we can encode simple
record subtyping using row polymorphism alone.
More generally, given the duality between records and variants, can we
swap the forms of polymorphism used by the above translations?

Though row polymorphism enables extending rows and what upcasting does
on record types is to remove labels, we can simulate the same
behaviour by extending record types that appear in contravariant
positions in a type.
The duality between row and presence polymorphism can be reconciled by
way of the duality between covariant and contravariant positions.

Let us revisit our $\var{getName}\ \var{alice}$ example, which we
previously encoded using polymorphism.
With row polymorphism (\STLCRecRow), we can give the function a row
polymorphic type where the row variable appears in the record type of
the function parameter.
\begin{examples}
  &\var{getName}_\xmark = \Lambda \rho. \lambda x^{\Record{\Name:\String;\rho}}.~(x.\Name) \\
\end{examples}
Now in order to apply $\var{getName}_\xmark$ to $\var{alice}$, we
simply instantiate $\rho$ with $(\Age:\Int)$.
\begin{examples}
  &\var{getName}_\xmark\ (\Age:\Int)\ \var{alice}
\end{examples}

Though the above example suggests a translation which only introduces
type abstractions and type applications, the idea does not extend to a
general composable translation.
Intuitively, the main problem is that in general we cannot know which
type should be used for instantiation ($\Age:\Int$ in this case) in
a compositional type-only translation, which is only allowed to use
the type of $\var{getName}$ and $\var{alice}\Upcast
\Record{\Name:\String}$. These tell us nothing about $\Age:\Int$.

In fact, a much stronger result holds. In
\Cref{sec:swapping-row-and-pre}, we prove that there exists no
type-only encoding of simple record subtyping into row polymorphism
(\nonexistarr{\STLCRecSub}{\STLCRecRow}), and dually for variant types
with presence polymorphism (\nonexistarr{\STLCVarSub}{\STLCVarPre}).

\subsection{Full Subtyping as Rank-1 Polymorphism}
\label{sec:example-prenex}

The kind of translation sought in \Cref{sec:example-swapping} cannot
be type-only, as it would require us to know the type used for
instantiation.
A natural question is whether type inference can provide the type.

In order to support decidable, sound, and complete type inference, we
consider a target calculus with rank-1 polymorphism (\STLCRecPrenex)
and Hindley-Milner type inference.
Now the $\var{getName}\ \var{alice}$ example type checks without an
explicit upcast or type application.
\footnote{Actually, the principal type of $\var{getName}$ should be $\forall
\alpha\,\rho. \Record{\Name:\alpha;\rho}\to\alpha$. We ignore value type
variables for simplicity.}
\begin{examples}
  &\var{getName} = \lambda x.~ (x.\Name)
    &: \forall \rho. \Record{\Name:\String;\rho}\to\String \\
  &\var{alice}\phantom{xx} = \Record{\Name = \makestring{Alice}; \Age = 9}
    &: \Record{\Name:\String; \Age:\Int} \\
  &\var{getName}\ \var{alice}
    &: \String
\end{examples}
Type inference automatically infers a polymorphic type for
$\var{getName}$, and instantiates the variable $\rho$ with
$\Age:\Int$.
This observation hints to us that we might encode terms with explicit
record upcasts in \STLCRecPrenex by simply erasing all upcasts (and
type annotations, given that we have type inference).
The global nature of erasure implies that it also works for full
subtyping (\STLCRecSubFull) which lifts the width subtyping of rows to
any type by propagating the subtyping relation to the components of
type constructors. For instance, the following function upcast using
full subtyping is also translated into $\var{getName}\ \var{alice}$,
simply by erasing the upcast.
\begin{examples}
  &(\var{getName}\Upcast (\Record{\Name:\String;\Age:\Int}\to\String))\ \var{alice}\\
\end{examples}

Thus far, the erasure translation appears to work well even for full
subtyping. Does it have any limitations? Yes, we must restrict the
target language to rank-1 polymorphism, which can only generalise
let-bound terms.
The type check would fail if we were to bind $\var{getName}$ via
$\lambda$-abstraction and then use it at different record types. For
instance, consider the following function which concatenates two names
using the $+\!\!+$ operator and is applied to $\var{getName}$.
\begin{examples}
  &(\lambda f^{\Record{\Name:\String}\to\String} .~
    f\ (\var{alice}\Upcast\Record{\Name:\String})
    +\!\!+\
    f\ (\var{bob}\Upcast\Record{\Name:\String}))\ \var{getName}
\end{examples}
The erasure of it is
\begin{examples}
  &(\lambda f.~
    f\ \var{alice}
    +\!\!+\
    f\ \var{bob})\ \var{getName}
\end{examples}
which is not well-typed as $f$ can only have a monomorphic function
type, whose parameter type cannot unify with both
$\Record{\Name:\String;\Age:\Int}$ and
$\Record{\Name:\String;\Year:\Int}$.

In order to avoid such problems, we will define an erasure translation
on a restricted subcalculus of \STLCRecSubFull.
The key idea is to give row-polymorphic types for record manipulation
functions such as $\var{getName}$.
However, the above function
takes a record manipulation function of type
$\Record{\Name:\String}\to\String$ as a parameter, which cannot be
polymorphic as we only have rank-1 polymorphism.
Inspired by the notion of rank-$n$ polymorphism, we say that a type
has \emph{rank-$n$ records}, if
no path from the root of the type (seen as an abstract syntax tree)
to a record type passes to the left of $n$ or more arrows.
We define the translation only on the subcalculus
\STLCRecSubFullCond{2} of \STLCRecSubFull in which all types have
rank-2 records.

Such an erasure translation underlies the local type-only translation
\localtypearr{\STLCRecSubFullCond{2}}{\STLCRecPrenex}.

We obtain a similar result for presence polymorphism.
With presence polymorphism, we can make all records
presence-polymorphic (similar to the translation in
\Cref{sec:example-simple-record-subtyping}), instead of making all
record manipulation functions row-polymorphic.
For instance, we can infer the following types for the
$\var{getName}\ \var{alice}$ example.
\begin{examples}
  &\var{getName} = \lambda x.~ (x.\Name)
    &: \Record{\Name:\String}\to\String \\
  &\var{alice}\phantom{xx} = \Record{\Name = \makestring{Alice}; \Age = 9}
    &: \forall \theta_1\theta_2. \Record{\Name^{\theta_1}:\String; \Age^{\theta_2}:\Int} \\
  &\var{getName}\ \var{alice}
    &: \String
\end{examples}
Consequently, records should appear only in positions that can be
generalised with rank-1 polymorphism, which can be ensured by
restricting \STLCRecSubFull to the subcalculus \STLCRecSubFullCond{1} in
which all types have rank-1 records.
We give a local type-only translation:
\localtypearr{\STLCRecSubFullCond{1}}{\STLCRecPrePrenex}.

For variants, we can also define the notion of \emph{rank-$n$
  variants} similarly.
Dually to records, we can either make all variants be row-polymorphic
(similar to the translation in
\Cref{sec:example-simple-variant-subtyping}) and require types to have
rank-1 variants (\STLCVarSubFullCond{1}), or make all variant
manipulation functions be presence-polymorphic and require types to
have rank-2 variants (\STLCVarSubFullCond{2}).
For instance, we can make the $\var{getAge}$ function
presence-polymorphic.
\begin{examples}
  &\var{getAge} = \lambda x.~
    \Case\ x\ \{\Age\ y \mapsto y; \Year\ y \mapsto 2023-y\}
    : \forall \theta_1\theta_2. {\Variant{\Age^{\theta_1}\!:\Int; \Year^{\theta_2}\!:\Int}}
      \to \Int \\
  &\var{year}\phantom{xx} = \Year\ 1984 : \Variant{\Age:\Int} \\
  &\var{getAge}\ \var{year}
\end{examples}
We give two type-only encodings of full variant subtyping:
\localtypearr{\STLCVarSubFullCond{1}}{\STLCVarPrenex}
and \localtypearr{\STLCVarSubFullCond{2}}{\STLCVarPrePrenex}.
\Cref{sec:prenex-polymorphism} discusses in detail the four erasure
translations from full subtyping to rank-1 polymorphism with type
inference.

\subsection{Strictly Covariant Record Subtyping as Presence Polymorphism}
\label{sec:example-co-record-subtyping}

The encodings of full subtyping discussed in \Cref{sec:example-prenex}
impose restrictions on types in the source language and rely heavily
on type-inference.
We now consider to what extent we can support a richer form of
subtyping than simple subtyping, if we turn our attention to target
calculi with higher-rank polymorphism and no type inference.

One complication of extending simple subtyping to full subtyping is
that if we permit propagation through contravariant positions, then
the subtyping order is reversed.
To avoid this scenario, we first consider \emph{strictly covariant
  subtyping} relation derived by only propagating simple subtyping
through strictly covariant positions (i.e. never to the left of any
arrow).
For example, the upcast $\var{getName}\Upcast
(\Record{\Name:\String;\Age:\Int}\to\String)$ in
\Cref{sec:example-prenex} is ruled out.
We write \STLCRecSubCo for our calculus with strictly covariant record
subtyping.

Consider the function $\var{getChildName}$ returning the name of the
child of a person.
\begin{examples}
  &\var{getChildName} = \lambda x^{\Record{\Child:\Record{\Name:\String}}}.~
                        \var{getName}\ (x.\Child)
\end{examples}
We can apply $\var{getChildName}$ to $\var{carol}$ who has a daughter
$\var{alice}$ with the strictly covariant subtyping relation
$\Record{\Name:\String;\Child:\Record{\Name\!:\!\String;\Age\!:\!\Int}}
\subtype \Record{\Child:\Record{\Name\!:\!\String}}$.
\begin{examples}
  &\var{carol} = \Record{\Name=\makestring{Carol}; \Child=\var{alice}} \\
  &\var{getChildName}\ (\var{carol}\Upcast \Record{\Child:\Record{\Name:\String}})
\end{examples}

If we work in a language without subtyping (\STLCRec), we can still
use $\eta$-expansions instead, by nested deconstruction and
reconstruction.
\begin{examples}
  &\var{getChildName}\ \Record{\Child=\Record{\Name=\var{carol}.\Child.\Name}}
\end{examples}
In general, we can simulate the full subtyping
(not only strictly covariant subtyping)
of both records and
variants using this technique. The nested de- and
re-construction can be reformulated into coercion functions to be more
compositional \citep{BREAZUTANNEN1991172}. In
\Cref{sec:encode-full-subtyping}, we show the standard local term-involved
translation \localtermarr{\STLCVarRecSubFull}{\STLCVarRec} formalising
this idea.

However, for type-only encodings, the idea of making every record
presence-polymorphic in \Cref{sec:example-simple-record-subtyping}
does not work directly. Following that idea, we would translate
$\var{carol}$ to
\begin{examples}
  &\var{carol}_\xmark = \Lambda\theta_1'\theta_2'.~ \Record{
    \dots;
    \Child=\var{alice'}
  }^{\Record{\Name^{\theta_1'}:\String; \Child^{\theta_2'}:
        \forall\theta_1\theta_2. \Record{\Name^{\theta_1}:\String; \Age^{\theta_2}:\Int}}} \\
\end{examples}
Then, as $\theta_1$ and $\theta_2$ are abstracted inside a record, we
cannot directly instantiate $\theta_2$ with $\Absent$ to remove the
$\Age$ label without deconstructing the outer record.
However, we can tweak the translation by moving the quantifiers
$\forall \theta_1\theta_2$ to the top-level through introducing new
type abstraction and type application, which gives rise to a
translation that is type-only but global.
\begin{examples}
  &\var{carol'} = \Lambda \theta_1\theta_2\theta_3\theta_4 .
    \Record{\dots; %
      \Child=\var{alice'}\ \theta_3\ \theta_4}^{
      \Record{\Name^{\theta_1}:\String; \Child^{\theta_2}:
              \Record{\Name^{\theta_3}:\String; \Age^{\theta_4}:\Int}}} \\
\end{examples}
Now we can remove the $\Name$ of $\var{carol'}$ and $\Age$ of $\var{alice'}$
by instantiating $\theta_1$ and $\theta_4$ with $\Absent$. As for simple subtyping,
we make the parameter type of $\var{getChildName}$ polymorphic, and re-abstract
over $\var{carol'}$.
\begin{examples}
  &\var{getChildName'} = \lambda x^{
    \forall\theta_1\theta_2 . \Record{\Child^{\theta_1}:\Record{\Name^{\theta_2}:\String}}}.~
                        \var{getName}\ ((x\ \Present\ \Present).\Child) \\
  &\var{getChildName'}\ (\Lambda\theta_1\theta_2.~ \var{carol'}\ \Absent\ \theta_1\ \theta_2\ \Absent)
\end{examples}
This is the essence of the global type-only translation
\globaltypearr{\STLCRecSubCo}{\STLCRecPre} in
\Cref{sec:encode-co-subtyping}.

\subsection{No Type-Only Encoding of Strictly Covariant Variant Subtyping as Polymorphism}
\label{sec:example-co-variant-subtyping}

We now consider whether we could exploit hoisting of quantifiers in
order to encode strictly covariant subtyping for variants
(\STLCVarSubCo) using row polymorphism.
Interestingly, we will see that this cannot work, thus breaking the
symmetry between the results for records and variants we have seen so
far.
To understand why, consider the following example involving nested
variants.
\begin{examples}
  &\var{data} = (\Raw\ \var{year})^{\Variant{\Raw:\Variant{\Year:\Int}}} \\
  &\var{data}\Upcast {\Variant{\Raw:\Variant{\Year:\Int;\Age:\Int}}}
\end{examples}
Following the idea of moving quantifiers, we can translate
$\var{data}$ to use a polymorphic variant, and the upcast can then be
simulated by instantiation and re-abstraction.
\begin{examples}
  &\var{data}_\xmark = \Lambda \rho_1\rho_2.~
    (\Raw\ (\var{year'}\ \rho_2))^{\Variant{\Raw:\Variant{\Year:\Int;\rho_2};\rho_1}} \\
  &\Lambda \rho_1\rho_2.~ \var{data}_\xmark\ \rho_1\ (\Age:\Int;\rho_2)
\end{examples}

So far, the translation appears to have worked.
However, it breaks down when we consider the case split on a nested
variant.
For instance, consider the following function.
\begin{examples}
  &\var{parseAge} = \lambda x^{\Variant{\Raw:\Variant{\Year:\Int}}}.~
    \Case\ x\ \{\Raw\ y \mapsto \var{getAge}\ (y\Upcast\Variant{\Age:\Int;\Year:\Int})\} \\
  &\var{parseAge}\ \var{data}
\end{examples}
Using an upcast and $\var{getAge}$ from
\Cref{sec:example-simple-variant-subtyping} in the case clause, 
it accepts the nested variant $\var{data}$.

The difficulty with encoding $\var{parseAge}$ with row polymorphism is
that the abstraction of the row variable for the inner record of
$\var{data_\xmark}$ is hoisted up to the top-level, but case split
requires a monomorphic value. Thus, we must instantiate $\rho_2$ with
$\Age:\Int$ \emph{before} performing the case split.
\begin{examples}
  &\var{parseAge}_\xmark = \lambda x^{\forall\rho_1\rho_2.
    \Variant{\Raw:\Variant{\Year:\Int;\rho_2};\rho_1}}.~
    \Case\ (x\ \cdot\ (\Age:\Int))\
    \{\Raw\ y \mapsto \var{getAge}\ y\} \\
  &\var{parseAge}_\xmark\ \var{data}_\xmark
\end{examples}
However, this would not yield a compositional type-only translation,
as the translation of the $\Case$ construct only has access to the
types of $x$ and the whole case clause, which provide no information
about $\Age:\Int$.
Moreover, even if the translation could somehow access this type
information, the translation would still fail if there were multiple
incompatible upcasts of $y$ in the case clause.
\begin{examples}
  &\Case\ x\ \{\Raw\ y \mapsto \dots y\Upcast\Variant{\Age:\Int;\Year:\Int} \dots
  y\Upcast\Variant{\Age:\String;\Year:\Int}\}
\end{examples}
The first upcast requires $\rho_2$ to be instantiated with $\Age:\Int$
but the second requires it to be instantiated with the incompatible
$\Age:\String$.
The situation is no better if we add presence polymorphism. In
\Cref{sec:non-stlcvarsubco-stlcvarrowpre}, we prove that there exists
no type-only encoding of strictly covariant variant subtyping into row and
presence polymorphism (\nonexistarr{\STLCVarSubCo}{\STLCVarRowPre}).

\subsection{No Type-Only Encoding of Full Record Subtyping as Polymorphism}
\label{sec:example-full-record-subtyping}

For variants, we have just seen that a type-only encoding of full
subtyping does not exist, even if we restrict propagation of simple
subtyping to strictly covariant positions.
For records, we have seen how to encode strictly covariant subtyping
with presence polymorphism by hoisting quantifiers to the top-level.
We now consider whether we could somehow lift the strictly covariance
restriction and encode full record subtyping with polymorphism.

The idea of hoisting quantifiers does not work arbitrarily, exactly
because we cannot hoist quantifiers through contravariant positions.
Moreover, presence polymorphism alone cannot extend rows. Consider the
full subtyping example $\var{getName}\Upcast
(\Record{\Name:\String;\Age:\Int}\to\String)$ from
\Cref{sec:example-prenex}. The $\var{getName}$ function is translated
to the $\var{getName'}$ function in
\Cref{sec:example-simple-record-subtyping}, which provides no way to
extend the parameter record type with $\Age:\Int$.
\begin{examples}
  &\var{getName'} = \lambda x^{\forall \theta.\Record{\Name^\theta:\String}}.~((x\ \Present).\Name) \\
\end{examples}
A tempting idea is to add row polymorphism:
\begin{examples}
  &\var{getName'}_\xmark =
    \Lambda\rho.\lambda x^{\forall \theta.\Record{\Name^\theta:\String;\rho}}.~((x\ \Present).\Name) \\
\end{examples}
Now we can instantiate $\rho$ with $\Age:\Int$ to simulate the upcast.
However, this still does not work.
One issue is that we have no way to remove the labels introduced by
the row variable $\rho$ in the function body, as $x$ is only
polymorphic in $\theta$.
For instance, consider the following upcast of the function $\var{getUnit}$
which replaces the function body of $\var{getName}$ with an upcast of $x$.
\begin{examples}
  &\var{getUnit} = \lambda x^{\Record{\Name:\String}}. (x\Upcast\Record{}) \\
  &\var{getUnit}\Upcast (\Record{\Name:\String;\Age:\Int}\to\Record{})
\end{examples}
Following the above idea, $\var{getUnit}$ is translated to
\begin{examples}
  &\var{getUnit}_\xmark =
    \Lambda\rho.\lambda x^{\forall\theta.\Record{\Name^\theta:\String;\rho}}. x\ \Absent \\
\end{examples}
Then, in the translation of the upcast of $\var{getUnit}$, the row
variable $\rho$ is expected to be instantiated with a row containing
$\Age:\Int$.
However, we cannot remove $\Age:\Int$ again in the translation of the
function body, meaning that the upcast inside $\var{getUnit}$ cannot
yield an empty record.

\Cref{sec:non-full-subtyping} expands on the discussion here and
proves that there exists no type-only translation of unrestricted full
record subtyping into row and presence polymorphism
(\nonexistarr{\STLCRecSubFull}{\STLCRecRowPre}).

\section{Calculi}
\label{sec:calculi}
The foundation for our exploration of relative expressive power of
subtyping and parametric polymorphism is \citeauthor{Church40}'s simply-typed
$\lambda$-calculus~\cite{Church40}.
We extend
it with variants and records, respectively. We further extend the
variant calculus twice: first with simple structural subtyping and
then with row polymorphism. Similarly, we also extend the record
calculus twice: first with structural subtyping and then with presence
polymorphism. In Section~\ref{sec:full-subtyping}~and~
\ref{sec:prenex-polymorphism}, we explore further extensions with
strictly covariant subtyping, full subtyping and rank-1 polymorphism.

\subsection[A Simply-Typed Base Calculus]{A Simply-Typed Base Calculus \STLC}
\label{sec:stlc}

\begin{figure}[htbp]
  \flushleft
  \textbf{Syntax}

  \begin{minipage}{0.45\textwidth}
    \begin{syntax}
      \slab{Kinds}    &   \Belong{K}{\Kinds}       &::= & \Type
        \mid \hlvarrec{\Row_\LS} \\
      \slab{Types}    &   \Belong{A,B}{\Types}     &::=& \alpha \mid A \to B\\
        & & \mid & \hlvar{\Variant{R}} \mid \hlrec{\Record{R}} \\
      \slab{Type~Environments} & \Belong{\Delta}{\TyEnvs} &::=& \cdot \mid \Delta, \alpha\\
      \slab{Term~Environments} & \Belong{\Gamma}{\Envs} &::=& \cdot \mid \Gamma, x : A\\
      \end{syntax}
  \end{minipage}
  \begin{minipage}{0.6\textwidth}
    \begin{syntax}
      \slab{Rows}     &\Belong{R}{\Rows}     &::=& \hlvarrec{\cdot \mid \ell:A; R} \\
      \slab{Terms}    &\Belong{M,N}{\Terms}  &::=& x \mid \lambda x^A.M \mid M\,N \\
        & &\mid & \hlvar{(\ell\,M)^A \mid \Case~M~\{\ell_i~x_i \mapsto N_i\}_i} \\
        & &\mid & \hlrec{\Record{\ell_i=M_i}_i \mid M.\ell}\\
      \text{$\hlvarrec{\Labels \supseteq \LS \ni \ell}$} \span\span\span \\
    \end{syntax}
  \end{minipage}
  \textbf{Static Semantics} \medskip\\
  \fbox{$\typ{\Delta}{A : K}$}
  \def \MathparLineskip {\lineskip=3pt} %
  \begin{mathpar}
    \inferrule*[Lab=\klab{Base}]
      {~}
      {\typ{\Delta, \alpha}{\alpha : \Type}}

    \inferrule*[Lab=\klab{Arrow}]
      {\typ{\Delta}{A : \Type} \\
      \typ{\Delta}{B : \Type}}
      {\typ{\Delta}{A \to B : \Type}}

    \hlvarrec{
    \inferrule*[Lab=\klab{EmptyRow}]
      {~}
      {\typ{\Delta}{\cdot : \Row_\LS}}
    }

    \hlvarrec{
    \inferrule*[Lab=\klab{ExtendRow}]
      {\typ{\Delta}{A : \Type} \\
       \typ{\Delta}{R : \Row_{\LS \uplus \{\ell\}}}}
      {\typ{\Delta}{\ell : A; R : \Row_{\LS}}}
    }

    \hlvar{
    \inferrule*[Lab=\klab{Variant}]
      {\typ{\Delta}{R : \Row_{\emptyset}}}
      {\typ{\Delta}{\Variant{R} : \Type}}
    }

    \hlrec{
    \inferrule*[Lab=\klab{Record}]
      {\typ{\Delta}{R : \Row_{\emptyset}}}
      {\typ{\Delta}{\Record{R} : \Type}}
    }
  \end{mathpar}

  \fbox{$\typ{\Delta;\Gamma}{M : A}$}
  \begin{mathpar}
    \inferrule*[Lab=\tylab{Var}]
        {~}
        {\typ{\Delta;\Gamma, x : A}{x : A}}

    \inferrule*[Lab=\tylab{Lam}]
        {\typ{\Delta;\Gamma, x : A}{M : B}}
        {\typ{\Delta;\Gamma}{\lambda x^A. M : A \to B}}

    \inferrule*[Lab=\tylab{App}]
        {\typ{\Delta;\Gamma}{M : A \to B} \\\\ \typ{\Delta;\Gamma}{N : A}}
        {\typ{\Delta;\Gamma}{M\,N : B}}

    \hlvar{
    \inferrule*[Lab=\tylab{Inject}]
        {(\ell : A) \in R \\\\ \typ{\Delta;\Gamma}{M : A}}
        {\typ{\Delta;\Gamma}{(\ell\,M)^{\Variant{R}} : \Variant{R}}}
    }

    \hlvar{
    \inferrule*[Lab=\tylab{Case}]
        {\typ{\Delta;\Gamma}{M : \Variant{\ell_i:A_i}_i} \\\\
         [\typ{\Delta;\Gamma,x_i : A_i}{N_i : B}]_i}
        {\typ{\Delta;\Gamma}{\Case~M~\{\ell_i~x_i \mapsto N_i\}_i : B}}
    }

    \hlrec{
    \inferrule*[Lab=\tylab{Record}]
        {[\typ{\Delta;\Gamma}{M_i : A_i}]_i}
        {\typ{\Delta;\Gamma}{\Record{\ell_i=M_i}_i : \Record{\ell_i : A_i}_i}}
    }

    \hlrec{
    \inferrule*[Lab=\tylab{Project}]
        {\typ{\Delta;\Gamma}{M : \Record{R} \\\\ (\ell : A) \in R}}
        {\typ{\Delta;\Gamma}{M.\ell : A}}
    }
  \end{mathpar}
  \textbf{Dynamic Semantics}
  \begin{reductions}
    \colorbox{white}{\betalab{Lam}}  &(\lambda x^A.M)\,N &\breducesto& M[N/x]\\
    \hlvar{\betalab{Case}}  & \Case~(\ell_j\,M)^A~\{\ell_i~x_i \mapsto N_i\}_i &\breducesto& N_j[M/x_j]\\
    \hlrec{\betalab{Project}}  & \Record{(\ell_i=M_i)_i}.\ell_j &\breducesto& M_j\\
  \end{reductions}
  \caption{Syntax, static semantics, and dynamic semantics of \STLC
  (unhighlighted parts), and its extensions with variants \STLCVar
  (highlighted parts with $\Variant{}$ subscript), and records \STLCRec
  (highlighted parts with $\Record{}$ subscript).}
  \label{fig:stlc-var-rec}
\end{figure}
Our base calculus is a Church-style simply typed $\lambda$-calculus,
which we denote $\STLC$. Figure~\ref{fig:stlc-var-rec} shows the
syntax, static semantics, and dynamic semantics of it.
The calculus features one kind ($\Type$) to classify well-formed
types. We will enrich the structure of kinds in the subsequent
sections when we add rows (e.g. Sections~\ref{sec:stlcvar} and
\ref{sec:stlcrec}).
The syntactic category of types includes abstract base types ($\alpha$) and
the function types ($A \to B$), which classify functions with domain
$A$ and codomain $B$.
The terms consist of variables ($x$), $\lambda$-abstraction ($\lambda
x^A.M$) binding variable $x$ of type $A$ in term $M$, and application
($M\,N$) of $M$ to $N$.
We track base types in a type environment ($\Delta$) and the type of
variables in a term environment ($\Gamma$). We treat environments as
unordered mappings.
The static and dynamic semantics are standard.
We implicitly require type annotations in terms to be well-kinded,
e.g., $\typ{\Delta;\Gamma}{\lambda x^A.M : A\to B}$ requires
$\Delta\vdash A$.

\subsection[A Calculus with Variants]{A Calculus with Variants \STLCVar}
\label{sec:stlcvar}

\newcommand{\allLS}{\LS_\omega}

$\STLCVar$ is the extension of $\STLC$ with variants.
\Cref{fig:stlc-var-rec} incorporates the extensions to the syntax,
static semantics, and dynamic semantics.
\emph{Rows} are the basis for variants (and later records).
We assume a countably infinite set of labels $\allLS$.
Given a finite set of labels $\LS$, a row of kind $\Row_{\LS}$ denotes
a partial mapping from the cofinite set $(\allLS \mathbin{\backslash}
{\LS})$ of all labels except those in $\LS$ to types.
We say that a row of kind $\Row_\emptyset$ is \emph{complete}.
A variant type ($\Variant{R}$) is given by a complete row $R$. A row
is written as a sequence of pairs of labels and types. We often omit
the leading $\cdot$, writing e.g. $\ell_1:A_1, \dots, \ell_n:A_n$ or
$(\ell_i:A_i)_i$ when $n$ is clear from context.
We identify rows up to reordering of labels.
Injection $(\ell\,M)^A$ introduces a term of variant type by tagging
the payload $M$ with $\ell$, whose resulting type is $A$. A
case split ($\Case~M~\{\ell_i~x_i \mapsto N_i\}_i$) eliminates an $M$ by
matching against the tags $\ell_i$. A successful match on $\ell_i$
binds the payload of $M$ to $x_i$ in $N_i$.
The kinding rules ensure that rows contain no duplicate labels.
The typing rules for injections and case splits and the $\beta$-rule
for variants are standard.

\subsection[A Calculus with Variants and Structural Subtyping]{A Calculus with Variants and Structural Subtyping $\STLCVarSub$}
\label{sec:stlcvarsub}
\begin{figure}[t]
\flushleft
\textbf{Syntax}
\vspace{-\the\baselineskip}
\begin{syntax}
\slab{Terms}    &   \Belong{M}{\Terms}     &::=& \ldots \mid \hlvrsub{M\Upcast A}\\
\end{syntax}
\textbf{Static Semantics} \medskip\\
\def \MathparLineskip {\lineskip=5pt}
\begin{minipage}[t]{0.55\textwidth}
\fbox{$A \subtype A'$}
\begin{mathpar}
  \hlvsub{
  \inferrule*[Lab=\ssublab{Variant}]
            { \dom{R} \subseteq \dom{R'} \\ R'|_{\dom{R}} = R }
            {\Variant{R} \subtype \Variant{R'}}
  } \\

  \hlrsub{
  \inferrule*[Lab=\ssublab{Record}]
            {\dom{R'}\subseteq \dom{R} \\ R|_{\dom{R'}} = R'}
            {\Record{R} \subtype \Record{R'}}
  }
\end{mathpar}
\end{minipage}
\begin{minipage}[t]{0.4\textwidth}
\fbox{$\typ{\Delta;\Gamma}{M : A}$}
\begin{mathpar}
\hlvrsub{
 \inferrule*[Lab=\tylab{Upcast}]
            {\typ{\Delta;\Gamma}{M:A} \\ A \subtype B}
            {\typ{\Delta;\Gamma}{M \Upcast B : B}}
}
\end{mathpar}
\end{minipage}

\textbf{Dynamic Semantics}
\begin{reductions}
  \hlvsub{\upcastlab{Variant}} & (\ell\,M)^A \Upcast B &\ureducesto& (\ell\,M)^B \\
  \hlrsub{\upcastlab{Record}} & \Record{\ell_i=M_{\ell_i}}_i \Upcast \Record{\ell_j':A_j}_j &\ureducesto& \Record{\ell'_j=M_{\ell'_j}}_j \\
\end{reductions}
\caption{Extensions of \STLCVar with simple subtyping \STLCVarSub
(highlighted parts with $\Variant{}$ subscript), and extensions of
\STLCRec with simple subtyping \STLCRecSub (highlighted parts with
$\Record{}$ subscript).}
\label{fig:stlc-var-rec-sub}
\end{figure}
$\STLCVarSub$ is the extension of $\STLCVar$ with simple structural
subtyping. Figure~\ref{fig:stlc-var-rec-sub} shows the extensions to
syntax, static semantics, and dynamic semantics.

\paragraph{Syntax}
The explicit upcast operator ($M \Upcast A$) coerces $M$ to type $A$.

\paragraph{Static Semantics}
The $\ssublab{Variant}$ rule asserts that variant $\Variant{R}$ is a
subtype of variant $\Variant{R'}$ if row $R'$ contains at least the
same label-type pairs as row $R$. We write $\dom{R}$ for the domain of
row $R$ (i.e. its labels), and $R|_{\mathcal{L}}$ for the restriction
of $R$ to the label set $\mathcal{L}$.
The $\tylab{Upcast}$ rule enables the upcast $M \Upcast B$ if the term
$M$ has type $A$ and $A$ is a subtype of $B$.

\paragraph{Dynamic Semantics}
The $\upcastlab{Variant}$ reduction rule coerces an injection
$(\ell\,M)$ of type $A$ to a larger (variant) type $B$. We distinguish
upcast rules from $\beta$ rules writing instead $\ureducesto$
for the reduction relation. Correspondingly, we write $\ureducestocl$
for the compatible closure of $\ureducesto$.

\subsection[A Calculus with Row Polymorphic Variants]{A Calculus with Row Polymorphic Variants \STLCVarRow}
\label{sec:stlcvarrow}

\begin{figure}
\flushleft
\textbf{Syntax}

\vspace{-\the\baselineskip}
\begin{minipage}{0.6\textwidth}
  \begin{syntax}
    \slab{Types} &\Belong{A}{\Types}     &::=& \ldots \mid \forall \rho^K.A \\
    \slab{Rows}  &\Belong{R}{\Rows} &::=& \ldots \mid \rho\\
  \end{syntax}
\end{minipage}
\begin{minipage}{0.4\textwidth}
  \begin{syntax}
    \slab{Terms} &\Belong{M}{\Terms}     &::=& \ldots \mid \Lambda \rho^K . M \mid M\,R\\
    \slab{Type~Environments}&\Belong{\Delta}{\TyEnvs} &::=& \ldots \mid \Delta, \rho : K \\
  \end{syntax}
\end{minipage}

\textbf{Static Semantics}\medskip\\
\def \MathparLineskip {\lineskip=8pt}
\begin{minipage}[t]{0.45\textwidth}
\fbox{$\typ{\Delta\phantom{;\!}}{A : K}$}
\begin{mathpar}
  \inferrule*[Lab=\klab{RowVar}]
    {~}
    {\typ{\Delta, \rho : \Row_\LS}{\rho : \Row_\LS}} \\

  \inferrule*[Lab=\klab{RowAll}]
    {\typ{\Delta,\rho : \Row_\LS}{A : \Type}}
    {\typ{\Delta}{\forall \rho^{\Row_\LS} .  A : \Type}}
\end{mathpar}
\end{minipage}
\begin{minipage}[t]{0.5\textwidth}
\fbox{$\typ{\Delta;\Gamma}{M : A}$}
\begin{mathpar}
  \inferrule*[Lab=\tylab{RowLam}]
     {\typ{\Delta, \rho : K;\Gamma}{M : A} \\ \rho \notin \ftv{\Gamma}}
     {\typ{\Delta;\Gamma}{\Lambda \rho^K.M : \forall \rho^K.A}} \\

  \inferrule*[Lab=\tylab{RowApp}]
     {\typ{\Delta;\Gamma}{M : \forall \rho^K.B} \\ \typ{\Delta}{A : K}}
     {\typ{\Delta;\Gamma}{M\,A : B[A/\rho]}}
\end{mathpar}
\end{minipage}

\textbf{Dynamic Semantics}
\begin{reductions}
  \tlab{RowLam} & (\Lambda \rho^K.M)\,R &\treducesto& M[R/\rho] \\
\end{reductions}
\caption{Extensions of \STLCVar with row polymorphism \STLCVarRow.}
\label{fig:var-calc-row}
\end{figure}

$\STLCVarRow$ is the extension of $\STLCVar$ with row
polymorphism. Figure~\ref{fig:var-calc-row} shows the extensions to
the syntax, static semantics, and dynamic semantics.

\paragraph{Syntax}
The syntax of types is extended with a quantified type ($\forall
\rho^K.A$) which binds the row variable $\rho$ with kind $K$ in the
type $A$ (the kinding rules restrict $K$ to always be of kind
$\Row_{\mathcal{L}}$ for some $\mathcal{L}$).
The syntax of rows is updated to allow a row to end in a row variable
($\rho$). A row variable enables the tail of a row to be extended with
further labels.
A row with a row variable is said to be \emph{open}; a row without a
row variables is said to be closed.

Terms are extended with type (row) abstraction ($\Lambda \rho^K.M$)
binding the row variable $\rho$ with kind $K$ in $M$ and row
application ($M\,R$) of $M$ to $R$.
Finally, type environments are updated to track the kinds of row
variables.

\paragraph{Static Semantics}
The kinding and typing rules for row polymorphism are the standard
rules for System F specialised to rows.

\paragraph{Dynamic Semantics}
The new rule \tlab{RowLam} is the standard $\beta$ rule
for System F, but specialised to rows.
Though it is a $\beta$ rule, we use the notation $\treducesto$ to
distinguish it from other $\beta$ rules as it only influences
types. This distinction helps us to make the meta theory of
translations in \Cref{sec:translations} clearer.
We write $\treducestocl$ for the compatible closure of $\treducesto$.

\subsection[A Calculus with Records]{A Calculus with Records \STLCRec}
\label{sec:stlcrec}

\STLCRec is \STLC extended with records. \Cref{fig:stlc-var-rec}
incorporates the extensions to the syntax, static semantics, and
dynamic semantics.
As with \STLCVar, we use rows as the basis of record types.
The extensions of kinds, rows and labels are the same as
$\STLCVar$.
As with variants a record type ($\Record{R}$) is given by a complete
row $R$.
Records introduction $\Record{\ell_i=M_i}_i$ gives a record in which
field $i$ has label $\ell_i$ and payload $M_i$.
Record projection ($M.\ell$) yields the payload of the field with label
$\ell$ from the record $M$.
The static and dynamic semantics for records are standard.

\subsection[A Calculus with Records and Structural Subtyping]{A Calculus with Records and Structural Subtyping \STLCRecSub}
\label{sec:stlcrecsub}

$\STLCRecSub$ is the extension of $\STLCRec$ with structural
subtyping. Figure~\ref{fig:stlc-var-rec-sub} shows the extensions to
syntax, static semantics, and dynamic semantics.
The only difference from \STLCVarSub is the subtyping rule
\ssublab{Record} and dynamic semantics rule \upcastlab{Record}.
The subtyping relation ($\subtype$) is just like that for
$\STLCVarSub$ except $R$ and $R'$ are swapped. The \ssublab{Record}
rule states that a record type $\Record{R}$ is a subtype of
$\Record{R'}$ if the row $R$ contains at least the same label-type
pairs as $R'$.
The \upcastlab{Record} rule upcasts a record $\Record{\ell_i=M_i}_i$
to type $\Record{R}$ by directly constructing a record with only the
fields required by the supertype $\Record{R}$.
We implicitly assume that the two indexes $j$ range over the same set
of integers.

\subsection[A Calculus with Presence Polymorphic Records]{A Calculus with Presence Polymorphic Records \STLCRecPre}
\label{sec:stlcrecpre}

\begin{figure}
\flushleft
\textbf{Syntax} \medskip\\
\begin{minipage}{0.45\textwidth}
  \begin{syntax}
  \slab{Kinds} & \Belong{K}{\Kinds}   &::=& \ldots %
                                      \mid \Presence \\
  \slab{Types} &  \Belong{A}{\Types}     &::=& \ldots %
                                      \mid \forall \theta.A \\
  \slab{Rows}  &\Belong{R}{\Rows} &::=& \ldots \mid \hlmod{\ell^{P} : A; R} \\
  \end{syntax}
\end{minipage}
\begin{minipage}{0.55\textwidth}
  \begin{syntax}
  \slab{Presence} & \Belong{P}{\Presences} &::=& \Absent \mid \Present \mid \theta \\
  \slab{Terms}& \Belong{M}{\Terms} &::=& \ldots
        \mid \Lambda \theta . M \mid M\,P \mid \hlmod{\Record{\ell_i=M_i}_i^A}\\
  \slab{Type~Environments}& \Belong{\Delta}{\TyEnvs} &::=& \ldots \mid \Delta, \theta \\
  \end{syntax}
\end{minipage}

\textbf{Static Semantics} \medskip\\
\fbox{$\typ{\Delta\phantom{;\!}}{A : K}$}
\def \MathparLineskip {\lineskip=8pt}
\vspace{-2em}
\begin{mathpar}
 \inferrule*[Lab=\klab{Absent}]
           {~}
           {\typ{\Delta}{\Absent : \Presence}}

 \inferrule*[Lab=\klab{Present}]
           {~}
           {\typ{\Delta}{\Present : \Presence}}

 \inferrule*[Lab=\klab{PreVar}]
           {~}
           {\typ{\Delta, \theta}{\theta : \Presence}}

 \inferrule*[Lab=\klab{PreAll}]
           {\typ{\Delta, \theta}{A : \Type}}
           {\typ{\Delta}{\forall \theta.A : \Type}}

  \hlmod{
  \inferrule*[Lab=\klab{ExtendRow}]
           {
            \typ{\Delta}{P : \Presence} \\\\
            \typ{\Delta}{A : \Type} \\\\
            \typ{\Delta}{R : \Row_{\LS \uplus \{\ell\}}}
          }
           {\typ{\Delta}{\ell^P : A; R : \Row_{\LS}}}
  }
\end{mathpar}
\fbox{$\typ{\Delta;\Gamma}{M : A}$}
\begin{mathpar}
  \inferrule*[Lab=\tylab{PreLam}]
           {\typ{\Delta,\theta;\Gamma}{M : A} \\ \theta \notin \ftv{\Gamma}}
           {\typ{\Delta;\Gamma}{\Lambda \theta.M : \forall \theta.A}}

  \inferrule*[Lab=\tylab{PreApp}]
           {\typ{\Delta;\Gamma}{M : \forall \theta.A}\\
            \Delta\vdash P : \Presence}
           {\typ{\Delta;\Gamma}{M\,P : A[P/\theta]}}

  \hlmod{
  \inferrule*[Lab=\tylab{Record}]
           {[\typ{\Delta;\Gamma}{M_i : A_i}]_i}
           {\typ{\Delta;\Gamma}{\Record{\ell_i=M_i}_i^{\Record{\ell_i^{P_i} : A_i}_i} : \Record{\ell_i^{P_i} : A_i}_i}}

  \inferrule*[Lab=\tylab{Project}]
           {\typ{\Delta;\Gamma}{M : \Record{R}} \\
            (\ell^\Present : A) \in R}
           {\typ{\Delta;\Gamma}{M.\ell : A}}
  }
\end{mathpar}

\textbf{Dynamic Semantics}
\begin{reductions}
  \hlmod{\betalab{Project}}  & \Record{(\ell_i=M_i)_i}^A.\ell_j &\breducesto& M_j\\
  \tbetalab{PreLam} & (\Lambda \theta.M)\,P &\treducesto& M[P/\theta]\\
\end{reductions}
\caption{Extensions and modifications to $\STLCRec$ with presence
polymorphism $\STLCRecPre$. Highlighted parts replace the old ones
in \STLCRec, rather than extensions.
}
\label{fig:rec-calc-pre}
\end{figure}

$\STLCRecPre$ is the extension of $\STLCRec$ with presence-polymorphic
records. Figure~\ref{fig:rec-calc-pre} shows the extensions to the
syntax, static semantics, and dynamic semantics.

\paragraph{Syntax}
The syntax of kinds is extended with the kind of presence types
($\Presence$).
The structure of rows is updated with presence annotations on labels
$(\ell_i^{P_i}:A_i)_i$.
Following \citet{Remy94}, a label can be marked as either absent
($\Absent$), present ($\Present$), or polymorphic in its presence
($\theta$). In each case, the label is associated with a
type. Thus, it is perfectly possible to say that some label $\ell$ is
absent with some type $A$.
As for row variables, the syntax of types is extended with a
quantified type ($\forall\theta.A$), and the syntax of terms is
extended with presence abstraction ($\Lambda \theta.M$) and
application ($M~P$).
To have a deterministic static semantics, we need to extend record
constructions with type annotations to indicate the presence types of
labels ($\Record{\ell_i = M_i}^A$).
Finally, the structure of type environments is updated to track
presence variables.
With presence types, we not only ignore the order of labels, but also
ignore absent labels when comparing row types.
Similarly when comparing two typed records in
\STLCRecPre.
For instance, the row $\Record{\ell_1 = M; \ell_2 =
N}^{\Record{\ell_1^\Present : A; \ell_2^\Absent : B}}$ is equivalent
to $\Record{\ell_1 = M}^{\Record{\ell_1^\Present : A}}$.

\paragraph{Static Semantics}
The kinding and typing rules for polymorphism (\klab{PreAll},
\tylab{PreLam}, \tylab{PreApp}) are the standard ones for System F
specialised to presence types.
The first three new kinding rules $\klab{Absent}$, $\klab{Present}$,
and $\klab{PreVar}$ handle presence types directly. They assign kind
$\Presence$ to absent, present, and polymorphic presence annotation
respectively.
The kinding rule $\klab{ExtendRow}$ is extended with a new kinding
judgement to check $P$ is a presence type.
The typing rules for records, $\tylab{Record}$, and projections,
$\tylab{Project}$, are updated to accommodate the presence annotations
on labels.
The typing rule for record introduction, \tylab{Record}, is changed
such that the type of each component coincides with the annotation.
The projection rule, \tylab{Project}, is changed such that the $\ell$
component must be present in the record row.

\paragraph{Dynamic Semantics}
The new rewrite rule \tbetalab{PreLam} is the standard $\beta$ rule
for System F, but specialised to presence types.
As with \STLCVarRow we use the notation $\treducesto$ to distinguish
it from other $\beta$ rules and write $\treducestocl$ for its
compatible closure.
The \betalab{Project^\star} rule is the same as \betalab{Project}, but
with a type annotation on the record.

\section{Simple Subtyping as Polymorphism}
\label{sec:translations}
\label{SEC:TRANS} %

In this section, we consider encodings of simple subtyping.
We present four encodings and two non-existence results as depicted in
\Cref{fig:results-summary}.
Specifically, in addition to the standard term-involved encodings of simple
variant and record subtyping in \Cref{sec:encode-stlcvarsub-stlcvar} and
\Cref{sec:encode-stlcrecsub-stlcrec}, we give type-only encodings of simple
variant subtyping as row polymorphism in
\Cref{sec:encode-stlcvarsub-stlcvarrow}, and simple record subtyping as
presence polymorphism in \Cref{sec:encode-stlcrecsub-stlcrecpre}.
For each translation, we establish its correctness by demonstrating the
preservation of typing derivations and the correspondence between the
operational semantics.
In \Cref{sec:swapping-row-and-pre}, we show the non-existence of type-only
encodings if we swap the row and presence polymorphism of the target languages.

\paragraph{Compositional Translations}

We restrict our attention to compositional translations defined
inductively over the structure of derivations.
For convenience we will often write these as if they are defined on
plain terms, but formally the domain is derivations rather than terms,
whilst the codomain is terms.
In this section translations on derivations will always be defined on
top of corresponding compositional translations on types, kind
environments, and type environments, in such a way that we obtain a
type preservation property for each translation.
In Sections~\ref{sec:full-subtyping}~and~\ref{sec:prenex-polymorphism}
we will allow non-compositional translations on types (as they will
necessarily need to be constructed in a non-compositional global
fashion, e.g., by way of a type inference algorithm).

\subsection{Local Term-Involved Encoding of $\STLCVarSub$ in $\STLCVar$}
\label{sec:encode-stlcvarsub-stlcvar}

We give a local term-involved compositional translation from $\STLCVarSub$ to
$\STLCVar$, formalising the idea of simulating
$\var{age}\Upcast\Variant{\Age:\Int;\Year:\Int}$ with case split and injection
in \Cref{sec:example-simple-variant-subtyping}.
\begin{equations}
  \transl{-} &:& \Deriv \to \Terms\\
  \transl{M^{\Variant{\ell_i:A_i}_i} \Upcast \Variant{R}}
  &=& \Case~\transl{M}~\{\ell_i~x_i \mapsto (\ell_i~x_i)^{\Variant{R}}\}_{i}
\end{equations}

The translation has a similar structure to the $\eta$-expansion of variants:
\begin{reductions}
  \etalab{Case} & M^{\Variant{\ell_i : A_i}_i} &\ereducesto& \Case~M~\{\ell_i~x_i \mapsto (\ell_i\,x_i)^{\Variant{\ell_i : A_i}_i}\}_i
\end{reductions}

The following theorem states that the translation preserves typing
derivations.
Note that compositional translations always translate environments
pointwise.
For type environments, we have $\transl{\Gamma,x:A} =
\transl{\Gamma},x:\transl{A}$.
For kind environments, we have the identity function $\transl{\Delta}
= \Delta$.
\begin{restatable}[Type Preservation]{theorem}{STLCVarSubSTLCVarTYPE}
  Every well-typed \STLCVarSub term $\typ{\Delta;\Gamma}{M : A}$ is translated
  to a well-typed \STLCVar term $\typ{\transl{\Delta};\transl{\Gamma}}{\transl{M}
  : \transl{A}}$.
\end{restatable}

In order to state an operational correspondence result, we first
define $\bureducesto$ as the union of $\breducesto$ and $\ureducesto$,
and $\bureducestocl$ as its compatible closure. There is a one-to-one
correspondence between reduction in $\STLCVarSub$ and reduction in
$\STLCVar$.

\begin{restatable}[Operational Correspondence]{theorem}{STLCVarSubSTLCVar}
  \label{thm:oc-varsub-var}
  For the translation $\transl{-}$ from \STLCVarSub to \STLCVar, we have
  \begin{itemize}[align=spaceleft]
  \item[\textsc{\emph{Simulation}}] If $M \bureducestocl N$, then $\transl M
  \breducestocl \transl N$.
  \item[\textsc{\emph{Reflection}}] If $\transl M \breducestocl \transl N$, then
  $M \bureducestocl N$.
  \end{itemize}
\end{restatable}

Intuitively, every step of $\beta$-reduction in \STLCVarSub is mapped
to itself in \STLCVar.
For every step of upcast reduction of $M^{\Variant{R'}}\Upcast
\Variant{R}$ in \STLCVarSub, the \upcastlab{Variant} rule guarantees
that $M$ must be a variant value.
Thus, it is mapped to one step of $\beta$-reduction which reduces the
$\eta$-expansion of $M$.
The full proofs of type preservation and operational correspondence
can be found in \Cref{app:proof-stlcvarsub-stlcvar}.

\subsection{Local Type-Only Encoding of $\STLCVarSub$ in $\STLCVarRow$}
\label{sec:encode-stlcvarsub-stlcvarrow}

We give a local type-only translation from \STLCVarSub to \STLCVarRow
by making variants row-polymorphic, as demonstrated by $\var{year'}$
and $\var{getAge'}$ in \Cref{sec:example-simple-variant-subtyping}.

\noindent
\hspace{-1em}
\begin{minipage}[t]{0.45\textwidth}
  \begin{equations}
    \transl{-} &:& \Types \to \Types\\
    \transl{\Variant{R}}
      &=& \forall \rho^{\Row_{R}}.\Variant{\transl{R};\rho}\\
    \transl{-} &:& \Rows \to \Rows\\
    \transl{(\ell_i:A_i)_i}
      &=& (\ell_i:\transl{A_i})_i \\
  \end{equations}
\end{minipage}
\hspace{2em}
\begin{minipage}[t]{0.5\textwidth}
  \begin{equations}
    \transl{-} &:& \Deriv \to \Terms\\
    \transl{(\ell\,M)^{\Variant{R}}}
      &=& \Lambda \rho^{\Row_{R}}. (\ell\,\transl{M})^{\Variant{\transl{R}; \rho}}\\
    \transl{\Case~M~\{\ell_i~x_i \mapsto N_i\}_i}
      &=& \Case~(\transl{M}\,{\cdot})~\{\ell_i~x_i \mapsto \transl{N_i}\}_i\\
    \transl{M^{\Variant{R}} \Upcast \Variant{R'}}
      &=& \Lambda \rho^{\Row_{R'}}.\transl{M}\, \hk{(\transl{R' \backslash R}; \rho)}\\
  \end{equations}
\end{minipage}

\bigskip

The $\Row_R$ is short for $\Row_{\dom{R}}$ and $R \backslash R'$ is
defined as row difference:
\begin{align*}
  R \backslash R' &= (\ell:A)_{(\ell:A)\in R \text{ and } (\ell:A)\notin R'}
\end{align*}

The translation preserves typing derivations.
\begin{restatable}[Type Preservation]{theorem}{STLCVarSubSTLCVarRowTYPE}
  Every well-typed \STLCVarSub term $\typ{\Delta;\Gamma}{M : A}$ is translated
  to a well-typed \STLCVarRow term $\typ{\transl{\Delta};\transl{\Gamma}}{\transl{M}
  : \transl{A}}$.
\end{restatable}

In order to state an operational correspondence result, we introduce
two auxiliary reduction relations.
First, we annotate the type application introduced by the translation of upcasts
with the symbol @ to distinguish it from the type application introduced by the
translation of $\Case$.
We write $\vreducesto$ for the associated reduction and
$\vreducestocl$ for its compatible closure.

\begin{reductions}
  \vlab{RowLam} & (\Lambda \rho^K.M)\,\hk{A} &\vreducesto& M[A/\rho]
\end{reductions}

Then, we add another intuitive reduction rule for upcast in
\STLCVarSub, which allows nested upcasts to reduce to a single
upcast.

\begin{reductions}
  \upcastnewlab{Nested} & M\Upcast A \Upcast B &\unewreducesto& M\Upcast B
\end{reductions}

We write $\uunewreducesto$ for the union of $\ureducesto$ and
$\unewreducesto$, and $\uunewreducestocl$ for its compatible closure.
There are one-to-one correspondences between $\beta$-reductions
(modulo $\treducestocl$), and between upcast and $\vreducestocl$.

\begin{restatable}[Operational Correspondence]{theorem}{STLCVarSubSTLCVarRow}
  \label{thm:oc-varsub-varrow}
  For the translation $\transl{-}$ from \STLCVarSub to \STLCVarRow, we have
  \begin{itemize}[align=spaceleft]
  \item[\textsc{\emph{Simulation}}]
  If $M \breducestocl N$, then
  $\transl M \treducestocl^\que \breducestocl \transl N$;
  if $M \ureducestocl N$, then
  $\transl M \vreducestocl \transl N$.
  \item[\textsc{\emph{Reflection}}]
  If $\transl M \treducestocl^\que \breducestocl \transl N$, then
  $M \breducestocl N$;
  if $\transl M \vreducestocl \transl N$, then
  $M \uunewreducestocl N$.
  \end{itemize}
\end{restatable}

We write $\treducestocl^\que$ to represent zero or one step of
$\treducestocl$.
For the $\beta$-reduction of a case-split in \STLCVarSub, in order to
reduce further in \STLCVarRow, the translation of it must first reduce
the empty row type application $\transl{M}\,\cdot$ by $\treducestocl$.
One step of upcast reduction in \STLCVarSub is simply mapped to the
corresponding type application in \STLCVarRow.
The other direction (reflection) is slightly more involved as one step
of $\vreducestocl$ in \STLCVarRow may correspond to a nested upcast;
hence the need for $\uunewreducestocl$ instead of $\ureducestocl$.
The proofs of type preservation and operational correspondence can be
found in \Cref{app:proof-stlcvarsub-stlcvarrow}.

\subsection{Local Term-Involved Encoding of $\STLCRecSub$ in $\STLCRec$}
\label{sec:encode-stlcrecsub-stlcrec}

We give a local term-involved translation from \STLCRecSub to \STLCRec,
formalising the idea of simulating $\var{alice}\Upcast\Record{\Name:\String}$
with projection and record construction in
\Cref{sec:example-simple-variant-subtyping}.

\begin{equations}
  \transl{-} &:& \Deriv \to \Terms\\
      \transl{M \Upcast \Record{\ell_i:A_i}_i}
  &=& \Record{\ell_i = \transl{M}.\ell_i}_i
\end{equations}

The translation has a similar structure to the $\eta$-expanding of
records, which is
\begin{reductions}
  \etalab{Project} & M^{\Record{\ell_i:A_i}_i} &\ereducesto& \Record{\ell_i = M.\ell_i}_i
\end{reductions}

The translation preserves typing derivations.

\begin{restatable}[Type Preservation]{theorem}{STLCRecSubSTLCRecTYPE}
  Every well-typed \STLCRecSub term $\typ{\Delta;\Gamma}{M : A}$ is translated
  to a well-typed \STLCRec term $\typ{\transl{\Delta};\transl{\Gamma}}{\transl{M}
  : \transl{A}}$.
\end{restatable}

One upcast or $\beta$-reduction in \STLCRecSub corresponds to a
sequence of $\beta$-reductions in \STLCRec.

\begin{restatable}[Operational Correspondence]{theorem}{STLCRecSubSTLCRec}
  For the translation $\transl{-}$ from \STLCRecSub to \STLCRec, we have
  \begin{itemize}[align=spaceleft]
  \item[\textsc{\emph{Simulation}}] If $M \bureducestocl N$, then $\transl M
  \breducestocl^\ast \transl N$.
  \item[\textsc{\emph{Reflection}}] If $\transl M \breducestocl N'$, then there
  exists $N$ such that $N' \breducestocl^\ast \transl N$ and $M \bureducestocl %
  N$.
  \end{itemize}
\end{restatable}

We write $\breducestocl^\ast$ to represent multiple (including zero)
steps of $\breducestocl$.
Unlike \Cref{thm:oc-varsub-var}, one step of reduction in
\STLCRecSub might be mapped to multiple steps of reduction in \STLCRec
because the translation of upcast possibly introduces multiple copies
of the same term.
For instance, $\transl{M \Upcast \Record{\ell_1:A;\ell_2:B}} =
\Record{\ell_1 = \transl{M}.\ell_1; \ell_2 = \transl{M}.\ell_2}$.
One step of $\beta$-reduction in $M$ in \STLCRecSub is mapped to at
least two steps of $\beta$-reduction in the two copies of $\transl{M}$
in \STLCRec.
Reflection is basically the reverse of simulation but requires at
least one step of reduction in \STLCRec.
The proofs of type preservation and operational correspondence can be
found in \Cref{app:proof-stlcrecsub-stlcrec}.

\subsection{Local Type-Only Encoding of $\STLCRecSub$ in $\STLCRecPre$}
\label{sec:encode-stlcrecsub-stlcrecpre}

Before presenting the translation, let us focus on order of labels in types.
Though generally we treat row types as unordered collections, in this
section we assume, without loss of generality, that there is a
canonical order on labels, and the labels of any rows (including
records) conform to this order.
This assumption is crucial in preserving the correspondence between
labels and presence variables bound by abstraction. For example,
consider the type $A = \Record{\ell_1 : A_1; \dots; \ell_n : A_n}$ in
$\STLCRecSub$.
Following the idea of making records presence polymorphic as
exemplified by $\var{getName'}$ and $\var{alice'}$ in
\Cref{sec:example-simple-record-subtyping}, this record is translated
as $\transl{A} = \forall
\theta_1\dots\theta_n.\Record{\ell_1^{\theta_1} : \transl{A_1}; \dots;
  \ell_n^{\theta_n} : \transl{A_n}}$.
With the canonical order, we can guarantee that $\ell_i$ always appears at the
$i$-th position in the record and possesses the presence variable bound at the
$i$-th position.
The full translation is as follows.

\noindent
\hspace{-1em}
\begin{minipage}[t]{0.45\textwidth}
  \begin{equations}
    \transl{-} &:& \Types \to \Types\\
    \transl{\Record{\ell_i : A_i}_i}
    &=& (\forall \theta_i)_i.\Record{\ell_i^{\theta_i} : \transl{A_i}}_i\\
  \end{equations}
\end{minipage}
\hspace{2em}
\begin{minipage}[t]{0.5\textwidth}
  \begin{equations}
    \transl{-} &:& \Deriv \to \Terms\\
    \transl{\Record{\ell_i = M_i}_i^{\Record{\ell_i : A_i}_i}}
    &=& (\Lambda \theta_i)_i.\Record{\ell_i =
        \transl{M_i}}_i^{\Record{\ell_i^{\theta_i} : \transl{A_i}}_i}\\

    \transl{M^{\Record{\ell_i : A_i}_i}.\ell_j}
    &=& (\transl{M}\,({P_i})_i).\ell_j \\
    \multicolumn{3}{r}{\text{where }
      \ba[t]{@{~}l@{~}l@{\quad}l}
        P_i = \Absent &, i \neq j &
        P_j = \Present
      \ea
    }\\

    \transl{M^{\Record{\ell_i : A_i}_i} \Upcast \Record{\ell'_j : A'_j}_j}
    &=& (\Lambda \theta_j)_j.\transl{M}\,(\hk{P_i})_i \\
    \multicolumn{3}{r}{\text{where }
      \ba[t]{@{~}l@{~}l@{\quad}l@{~}l}
        P_i = \Absent &, \ell_i \notin (\ell'_j)_j &
        P_i = \theta_j &, \ell_i = \ell'_j
      \ea
    }\\
  \end{equations}
\end{minipage}

The translation preserves typing derivations.
\begin{restatable}[Type Preservation]{theorem}{STLCRecSubSTLCRecPreTYPE}
  Every well-typed \STLCRecSub term $\typ{\Delta;\Gamma}{M : A}$ is translated
  to a well-typed \STLCRecPre term $\typ{\transl{\Delta};\transl{\Gamma}}{\transl{M}
  : \transl{A}}$.
\end{restatable}

Similarly to \Cref{sec:encode-stlcvarsub-stlcvarrow}, we annotate type
applications introduced by the translation of upcast with @, and write
$\vreducesto$ for the associated reduction rule and $\vreducestocl$
for its compatible closure.
\begin{reductions}
  \vlab{PreLam} & (\Lambda \theta.M)\,\hk{P} &\vreducesto& M[P/\theta]\\
\end{reductions}
We also re-use the \upcastnewlab{Nested} reduction rule defined in
\Cref{sec:encode-stlcvarsub-stlcvarrow}.
There is a one-to-one correspondence between $\beta$-reductions
(modulo $\treducestocl$), and a correspondence between one upcast
reduction and a sequence of $\vreducestocl$ reductions.

\begin{restatable}[Operational Correspondence]{theorem}{STLCRecSubSTLCRecPre}
  \label{thm:oc-recsub-recpre}
  The translation $\transl{-}$ from \STLCRecSub to \STLCRecPre has the following properties:
  \begin{itemize}[align=spaceleft]
  \item[\textsc{\emph{Simulation}}]
  If $M \breducestocl N$, then $\transl M \treducestocl^\ast\breducestocl \transl
  N$; if $M \ureducestocl N$, then $\transl M \vreducestocl^\ast \transl N$.
  \item[\textsc{\emph{Reflection}}]
  If $\transl M \treducestocl^\ast\breducestocl \transl N$, then
  $M \breducestocl N$;
  if $\transl M \vreducestocl N'$, then there exists $N$ such that
  $N'\vreducestocl^\ast \transl N$ and $M \uunewreducestocl N$.
  \end{itemize}
\end{restatable}

Unlike \Cref{thm:oc-varsub-varrow}, one step of reduction in
\STLCRecSub might be mapped to multiple steps of reduction in
\STLCRecPre because we might need to reduce the type application of
multiple presence types in the translation results of projection and
upcast.
Reflection is again basically the reverse of simulation, requiring at
least one step of reduction in \STLCRecPre.
The proofs of type preservation and operational correspondence can be found in \Cref{app:proof-stlcrecsub-stlcrecpre}.

\subsection{Swapping Row and Presence Polymorphism}
\label{sec:swapping-row-and-pre}

In \Cref{sec:encode-stlcvarsub-stlcvarrow} and
\Cref{sec:encode-stlcrecsub-stlcrecpre}, we encode simple subtyping for variants
using row polymorphism, and simple subtyping for records using presence
polymorphism. These encodings enjoy the property that they only introduce new
type abstractions and applications.
A natural question is whether we can swap the polymorphism used by the encodings
meanwhile preserve the type-only property.
As we have seen in \Cref{sec:example-swapping}, an intuitive attempt to encode
simple record subtyping with row polymorphism failed. Specifically, we have the
problematic translation
\begin{examples}
& \phantom{=}\ \transl{\var{getName}\ (\var{alice}\Upcast\Record{\Name:\String})} \\
& =\transl{\var{getName}}\ (\Age:\Int)\ \transl{\var{alice}\Upcast\Record{\Name:\String}} \\
& =\var{getName}_\xmark\ (\Age:\Int)\ \var{alice}
\end{examples}
First, the type information $\Age:\Int$ is not accessible to a compositional
type-only translation of the function application here.
Moreover, the type preservation property is also broken:
$\transl{\var{alice}\Upcast\Record{\Name:\String}}$ should have type
$\transl{\Record{\Name:\String}}$, but here it is just translated to
$\var{alice}$ itself, which has an extra label $\Age$ in its record type.
We give a general non-existence theorem.

\begin{restatable}{theorem}{globalRecSubRecRow}
  \label{thm:swapping}
  There exists no global type-only encoding of $\STLCRecSub$ in $\STLCRecRow$,
  and no global type-only encoding of $\STLCVarSub$ in $\STLCVarPre$.
\end{restatable}

The extensions for $\STLCRecRow$ and $\STLCVarPre$ are straightforward
and can be found in \Cref{app:more-calculi}.
The proofs of this theorem can be found in
\Cref{app:non-existence-swapping}.
We will give further non-existence results in
\Cref{sec:full-subtyping}.
The core idea underlying the proofs of this kind of non-existence
result is to construct counterexamples and use proof by contradiction.
One important observation is that in our case a type-only translation
ensures that terms are invariant under the translation modulo type
abstraction and type application.
As a consequence, we may characterise the general form of any such
translation by accounting for the possibility of adding type
abstractions and type applications in every possible position.
Then we can obtain a contradiction by considering the general form of
type-only translations of carefully selected terms.

To give an example, let us consider the proof of \Cref{thm:swapping}.
Consider $\Record{}$ and $\Record{\ell=y}\Upcast\Record{}$ which have
the same type under environments $\Delta=\alpha_0$ and $\Gamma =
y:\alpha_0$.
Any type-only translation must yield $\transl{\Record{}} = \Lambda
\ol{\alpha} . \Record{}$ and
\[
  \transl{\Record{\ell=y} \Upcast \Record{}}
  = \Lambda \ol{\beta} . \transl{\Record{\ell=y}}\ \ol B
  = \Lambda \ol{\beta} . (\Lambda\ol{\alpha}' .
    \Record{\ell = \transl{y}\ \ol A'})\ \ol B
  = \Lambda \ol{\beta} . (\Lambda\ol{\alpha}' .
    \Record{\ell = (\Lambda\ol{\beta}'.y)\ \ol A'})\ \ol B
\]
which can be simplified to
$\Lambda \ol{\gamma} . \Record{\ell=\Lambda \ol{\delta} . y}$.
Thus, $\transl{\Record{}}$ has type $\forall\ol{\alpha}.\Record{}$,
and $\transl{\Record{\ell=y} \Upcast \Record{}}$ has type $\forall
\ol{\gamma} . \Record{\ell : \forall \ol{\delta}.\alpha_0}$.
By type preservation, they should still have the same type, which
implies $\forall \ol{\alpha} . \Record{} = \forall \ol{\gamma} .
\Record{\ell : \forall \ol{\delta} . \alpha_0}$.
However, this equation obviously does not hold, showing a
contradiction.

The above proof relies on the assumption that translations should
always satisfy the type preservation theorem.
Sometimes this assumption can be too strong.
In order to show the robustness of our theorem, we provide three
proofs of \Cref{thm:swapping} in \Cref{app:non-existence-swapping},
where only one of them relies on type preservation.
The second proof uses the compositionality and a similar argument to
the $\var{getName}_\xmark$ example in \Cref{sec:example-swapping},
while the third proof does not rely on either of them.

In \Cref{sec:prenex-polymorphism}, we will show that it is possible to
simulate record subtyping with rank-1 row polymorphism and type
inference, at the cost of a weaker type preservation property and some
extra conditions on the source language.

\section{Full Subtyping as Polymorphism}
\label{sec:full-subtyping}
\label{SEC:FULL}

So far we have only considered simple subtyping, which means the subtyping
judgement applies shallowly to a single variant or record constructor
(width subtyping).
Any notion of simple subtyping can be mechanically lifted to
full subtyping by inductively propagating the subtyping relation to the
components of each type.
The direction of the subtyping relation remains the same for covariant
positions, and is reversed for contravariant positions.

In this section, we consider encodings of full subtyping.
We first formalise the calculus \STLCVarRecSubFull\ with full subtyping for
records and variants, and give its standard term-involved translation to
\STLCVarRec (\Cref{sec:encode-full-subtyping}).
Next we give a type-only encoding of strictly covariant record subtyping
(\Cref{sec:encode-co-subtyping}) and a non-existence result for
variants (\Cref{sec:non-stlcvarsubco-stlcvarrowpre}).
Finally, we give a non-existence result for type-only encodings of
full record subtyping as polymorphism (\Cref{sec:non-full-subtyping}).

\subsection{Local Term-Involved Encoding of $\STLCVarRecSubFull$ in $\STLCVarRec$}
\label{sec:encode-full-subtyping}

\begin{figure}[t]
  \flushleft
  \fbox{$A \subtype A'$}
  \begin{mathpar}
    \inferrule*[Lab=\fsublab{Var}]
      { }
      {\alpha \subtype \alpha}

    \inferrule*[Lab=\fsublab{Fun}]
      {A' \subtype A \\ B \subtype B'}
      {A \to B \subtype A' \to B'}

    \inferrule*[Lab=\fsublab{Variant}]
      {\dom{R}\subseteq \dom{R'} \\\\
       [A_i \subtype A_i']_{(\ell_i:A_i) \in R, (\ell_i:A_i') \in R'}}
      {\Variant{R} \subtype \Variant{R'}}

    \inferrule*[Lab=\fsublab{Record}]
      {\dom{R'}\subseteq \dom{R} \\\\
       [A_i \subtype A_i']_{(\ell_i:A_i) \in R, (\ell_i:A_i') \in R'}}
      {\Record{R} \subtype \Record{R'}}
  \end{mathpar}
  \caption{Full subtyping rules of \STLCVarRecSubFull.}
  \label{fig:stlc-var-rec-sub-full}
\end{figure}

We first consider encoding \STLCVarRecSubFull, an extension of
\STLCVarSub and \STLCRecSub with full subtyping, in \STLCVarRec, the
combination of \STLCVar and \STLCRec. %
\Cref{fig:stlc-var-rec-sub-full} shows the standard full subtyping
rules of \STLCVarRecSubFull.
We inductively propagate the subtyping relation to sub-types, and reverse the
subtyping order for function parameters because of contravariance.
The reflexivity and transitivity rules are admissible.

For the dynamic semantics of \STLCVarRecSubFull, one option is to give
concrete upcast rules for each value constructor, similar to \STLCVarSub
and \STLCRecSub.
However, as encoding full subtyping is more intricate than encoding
simple subtyping (especially the encoding in
\Cref{sec:encode-co-subtyping}), upcast reduction rules significantly
complicate the operational correspondence theorems.
To avoid such complications we adopt an \emph{erasure semantics} for
\STLCVarRecSubFull which, following \citet{tapl}, interprets upcasts
as no-ops.
The type erasure function $\erase{-}$ transforms typed terms in
\STLCVarRecSubFull to untyped terms in \STLCVarRec by erasing all
upcasts and type annotations.
It is given by the homomorphic extension of the following equations.
\begin{equations}
  \erase{M \Upcast A} = \erase{M} \quad
  \erase{\lambda x^A . M} = \lambda x.\erase{M} \quad
  \erase{(\ell\, M)^A} = \ell\, \erase{M}
\end{equations}

We show a correspondence between upcasting and erasure
in \Cref{app:erasure-semantics}.
In the following, we always use the erasure semantics for calculi with
full subtyping or strictly covariant subtyping.

The idea of the local term-involved translation from
\STLCVarRecSubFull to \STLCVarRec in
\Cref{sec:example-co-record-subtyping} has been well-studied as the
\emph{coercion semantics} of subtyping \citep{BREAZUTANNEN1991172,
  BREAZUTANNEN90, tapl}, which transforms subtyping relations
$A\subtype B$ into coercion functions $\transl{A\subtype B}$.
Writing translations in the form of coercion functions ensures
compositionality. The translation is standard and shown in
\Cref{app:coercion-functions}.
For instance, the full subtyping relation in
\Cref{sec:example-co-record-subtyping} is translated to
\begin{examples}
  &\phantom{=}\ \transl{\Record{\Name:\String;\Child:\Record{\Name:\String;\Age:\Int}} \subtype
  \Record{\Child:\Record{\Name:\String}}} \\
  &= (\lambda x%
  .\Record{\Child = \transl{\Record{\Name:\String;\Age:\Int}\subtype \Record{\Name:\String}}\
          x.\Child})\\
  &= \lambda x%
  .\Record{\Child = (\lambda x%
          .\Record{\Name=x.\Name})\
          x.\Child}) \\
  &\breducestocl^\ast \lambda x%
  .\Record{\Child =
          \Record{\Name=x.\Child.\Name}}
\end{examples}
We refer the reader to \citet{tapl} and \citet{BREAZUTANNEN90} for the
standard type preservation and operational correspondence theorems and
proofs.

\subsection{Global Type-Only Encoding of \STLCRecSubCo in \STLCRecPre}
\label{sec:encode-co-subtyping}

As a stepping stone towards exploring the possibility of type-only
encodings of full subtyping, we first consider an easier problem: the
encoding of \STLCRecSubCo, a calculus with strictly covariant structural
subtyping for records.
Strictly covariant subtyping lifts simple subtyping through only the
covariant positions of all type constructors.
For \STLCVarRecSubCo, the only change with respect to
\STLCVarRecSubFull is to replace the subtyping rule \fsublab{Fun} with
the following rule which requires the parameter types to be equal:
\begin{mathpar}
  \inferrule*%
    {B \subtype B'}
    {A \to B \subtype A \to B'}
\end{mathpar}

As illustrated by the examples $\var{carol}_\xmark$ and $\var{carol'}$
from \Cref{sec:example-co-record-subtyping}, we can extend the idea of
encoding simple record subtyping as presence polymorphism described in
\Cref{sec:encode-stlcrecsub-stlcrecpre} by hoisting quantifiers to the
top-level, yielding a global but type-only encoding of \STLCRecSubCo
in \STLCRecPre.
The full type and term translations are spelled out in
\Cref{fig:encode-stlcrecsubco-stlcrecpre} together with three
auxiliary functions.
\noindent
\begin{figure}[htbp]
\hspace{-1em}
\begin{minipage}[t]{.45\linewidth}
\raggedright \small
\begin{equations}
  \transl{-} &:& \Types \to \Types\\
  \transl{A \to B}
  &=& \forall \ol{\theta} . \transl{A} \to \transl{B, \ol{\theta}} \\
  \multicolumn{3}{r}{\text{where }
      \ba[t]{@{~}l@{~}l}
      \ol{\theta} = \transpre{\theta}{B} \\
      \ea
  }\\
  \transl{\Record{\ell_i : A_i}_i}
  &=& \forall (\theta_i)_i\, (\ol{\theta}_i)_i.
      \Record{\ell_i^{\theta_i} : \transl{A_i, \ol{\theta}_i}}_i \\
  \multicolumn{3}{r}{\text{where }
      \ba[t]{@{~}l@{~}l}
      \ol{\theta}_i = \transpre{\theta_i}{A_i} \\
      \ea
  }\\[3ex]

  \transl{-} &:& \Deriv \to \Terms\\
  \transl{\lambda x^A . M^B}
  &=& \Lambda \ol{\theta} . \lambda x^{\transl{A}} . \transl{M}\ \ol{\theta} \\
  \multicolumn{3}{r}{\text{where }
      \ba[t]{@{~}l@{~}l}
      \ol{\theta} = \transpre{\theta}{B}
      \ea
  }\\

  \transl{M^A\ N^B}
  &=& \Lambda \ol{\theta} . (\transl{M}\ \ol{\theta})\ \transl{N} \\
  \multicolumn{3}{r}{\text{where }
      \ba[t]{@{~}l@{~}l}
      \ol{\theta} = \transpre{\theta}{A}
      \ea
  }\\

  \transl{\Record{\ell_i = M_i^{A_i}}_i}
  &=& \Lambda (\theta_i)_i\, (\ol{\theta}_i)_i .
      \Record{\ell_i = \transl{M_i}\ \ol{\theta}_i}_i^{\Record{\ell_i^{\theta_i} : \transl{A_i}}_i}\\
  \multicolumn{3}{r}{\text{where }
      \ba[t]{@{~}l@{~}l}
      \ol{\theta}_i = \transpre{\theta_i}{A_i} \\
      \ea
  }\\

  \transl{M^{\Record{\ell_i : A_i}_i}.\ell_j}
  &=& \Lambda \ol{\theta} . (\transl{M}\ (P_i)_i\ (\ol{P}_i)_{i<j}\ \ol{\theta}\ (\ol{P}_i)_{j<i}).\ell_j \\
  \multicolumn{3}{r}{\text{where }
      \ba[t]{@{~}l@{~}l@{\qquad}l}
      P_i = \Absent &, i \neq j &
      \ol{\theta} = \transpre{\theta}{A_j} \\
      P_j = \Present & &
      \ol{P}_i = \transpre{\Absent}{A_i}
      \ea
  }\\

  \transl{M^A \Upcast B}
  &=& \Lambda \ol{\theta} . \transl{M}\ {\ol{P}} \\
  \multicolumn{3}{r}{\text{where }
      \ba[t]{@{~}l@{~}l}
      (\ol{\theta}, \ol{P}) = \transpre{\theta}{A\subtype B}
      \ea
}\\
\end{equations}
\end{minipage}%
\hspace{1em}
\begin{minipage}[t]{.5\linewidth}
\raggedleft \small
\begin{equations}
  \transl{-,-} &:& (\Types,\ol{\Presence}) \to \Types\\
  \transl{A, \ol{P}}
  &=& A'[\ol{P} / \ol{\theta}'] \\
  \multicolumn{3}{r}{\text{where }
      \ba[t]{@{~}l@{~}l}
      \forall \ol{\theta}'.A' = \transl{A}
      \ea
  }\\[3ex]

  \transpre{-}{-} &:& (\Presence, \Type) \to \ol{\Presence}\\
  \transpre{P}{\alpha}
  &=& \cdot
  \\
  \transpre{P}{A \to B}
  &=& \transpre{P}{B}
  \\
  \transpre{P}{\Record{\ell_i : A_i}_i}
  &=& (P_i)_i\ \transpre{P_i}{A_i}_i \\
  \multicolumn{3}{r}{\text{where }
      \ba[t]{@{~}l@{~}l}
      P_i = \theta_i &, P \text{ is a variable }\theta \\
      P_i = \Absent &, P = \Absent \\
      P_i = \Present &, P = \Present
      \ea
  }
  \\[3ex]
  \\

  \transpre{-}{-} : (\Presence, \Type\subtype\Type) \to (\ol{\Presence}, \ol{\Presence}) \span\span\\
  \transpre{\theta}{\alpha\subtype\alpha}
  &=& (\cdot, \cdot)
  \\
  \transpre{\theta}{A \to B \subtype A \to B'}
  &=& \transpre{\theta}{B \subtype B'}
  \\
  \transpre{\theta}{\Record{\ell_i : A_i}_i \subtype \Record{\ell_j' : A_j'}_j}
  &=& ((\theta_j)_j\ (\ol{\theta}_j)_j, (P_i)_i\ (\ol{P}_i)_i) \\
  \multicolumn{3}{r}{\text{where }
      \ba[t]{@{~}l@{~}l@{\qquad}l@{~}l}
      (\ol{\theta}_j, \ol{P}_j') = \transpre{\theta_j}{A_i \subtype A_j'} , \ell_i = \ell_j' \span\span\span\\
      P_i = \Absent &, \ell_i \notin (\ell'_j)_j &
      \ol{P}_i = \transpre{\Absent}{A_i} &, \ell_i \notin (\ell'_j)_j \\
      P_i = \theta_j &, \ell_i = \ell_j' &
      \ol{P}_i = \ol{P}_j' &, \ell_i = \ell_j'
      \ea
  } \\
\end{equations}
\end{minipage}

\caption{A global type-only translation from $\STLCRecSubCo$ to $\STLCRecPre$.}
\label{fig:encode-stlcrecsubco-stlcrecpre}
\end{figure}

As in \Cref{sec:encode-stlcrecsub-stlcrecpre}, we rely on a canonical
order on labels.
The auxiliary function $\transl{A,\ol{P}}$ instantiates a polymorphic
type $A$ with $\ol{P}$, simulating the type application in the term
level.
The auxiliary function $\transpre{\theta}{A}$ takes a presence
variable $\theta$ and a type $A$, and generates
a sequence of presence variables based on $\theta$ that have the same
length as the presence variables bound by $\transl{A}$.
It is used to allocate a fresh presence variable for every label in
records on strictly covariant positions.
We can also use it to generate a sequence of $\Present$ or $\Absent$
for the instantiation of $\transl{A}$ by $\transpre{\Present}{A}$ and
$\transpre{\Absent}{A}$.
The auxiliary function $\transpre{\theta}{A\subtype B}$ takes a
presence variable $\theta$ and a subtyping relation $A\subtype B$, and
returns a pair $(\ol \theta, \ol{P})$.
The sequence of presence variables $\ol{\theta}$ is the same as
$\transpre{\theta}{B}$.
The sequence of presence types are used to instantiate $\transl{A}$ to
get $\transl{B}$ (as illustrated by the term translation
$\transl{M^A\Upcast B} = \Lambda \ol{\theta}.\transl{M}\,{\ol{P}}$ which
has type $\transl{B}$).

The translation on types is straightforward. We not only introduce a
presence variable for every element of record types, but also move the
quantifiers of the types of function bodies and record elements to the
top level, as they are on strictly covariant positions.
While the translation on terms (derivations) may appear complicated,
it mainly focuses on moving type abstractions to the top level by type
application and re-abstraction using the auxiliary functions.
For the projection and upcast cases, it also instantiates the
sub-terms with appropriate presence types.
Notice that for function application $M\ N$, we only need to move the type
abstractions in $\transl{M}$, and for projection $M.\ell_j$, we only need to
move the type abstractions in the payload of $\ell_j$.

Strictly speaking, the type translation is actually not compositional
because of the type application introduced by the term translation.
As a consequence, in the type translation, we need to use the
auxiliary function $\transl{A,\ol{P}}$ which looks into the concrete
structure of $\transl{A}$ instead of using it compositionally.
However, we believe that it is totally fine to slightly compromise the
compositionality of the type translation, which is much less
interesting than the compositionality of the term translation.
Moreover, we can still make the type translation compositional by
extending the type syntax with type operators and type-level type
application of System F$\omega$.

We have the following type preservation theorem. The proof shown in
\Cref{app:proof-recsubco-recpre} follows from induction on typing
derivations of $\STLCRecSubCo$.
\begin{restatable}[Type Preservation]{theorem}{STLCVarRecSubCoSTLCVarRecPreTYPE}
  \label{thm:type-recsubco-recpre}
  Every well-typed \STLCRecSubCo term $\typ{\Delta;\Gamma}{M : A}$ is
  translated to a well-typed \STLCRecPre term
  $\typ{\transl{\Delta};\transl{\Gamma}}{\transl{M} : \transl{A}}$.
\end{restatable}

In order to state an operational correspondence result, we use the
erasure semantics for \STLCRecPre given by the standard type erasure
function defined as the homomorphic extension of the following
equations.
\begin{equations}
  \erase{\Lambda \theta . M} = \erase{M} \quad
  \erase{M\ P} = \erase{M} \quad
  \erase{\lambda x^A . M} = \lambda x.\erase{M} \quad
\end{equations}
Since the terms in \STLCRecSubCo and \STLCRecPre are both erased to
untyped \STLCRec, for the operational correspondence we need only show
that any term in \STLCRecSubCo is still erased to the same term after
translation.

\begin{restatable}[Operational Correspondence]{theorem}{STLCRecSubCoSTLCRecPre}
  \label{thm:oc-recsubco-recpre}
  The translation $\transl{-}$ from \STLCRecSubCo to \STLCRecPre satisfies the
  equation $\erase{M} = \erase{\transl{M}}$ for any well-typed term $M$ in
  \STLCRecSubCo.
\end{restatable}
\begin{proof}
  By straightforward induction on $M$.
\end{proof}

By using erasure semantics, the operational correspondence becomes
concise and obvious for type-only translations, as all constructs
introduced by type-only translations are erased by type erasure
functions.
It is also possible to reformulate \Cref{thm:oc-varsub-varrow} and
\Cref{thm:oc-recsub-recpre} to use erasure semantics, but the current
versions are somewhat more informative and not excessively complex.

\subsection{Non-Existence of Type-Only Encodings of \STLCVarSubCo in \STLCVarRowPre}
\label{sec:non-stlcvarsubco-stlcvarrowpre}

As illustrated by the example
$\var{parseAge}_\xmark\ \var{data}_\xmark$ in
\Cref{sec:example-co-variant-subtyping}, the approach of hoisting
quantifiers to the top-level does not work for variants, because of
case splits.
Formally, we have the following general non-existence theorem showing
that no other approaches exist.
\begin{restatable}{theorem}{globalVarSubCo}
  \label{thm:globalVarSubCo}
  There exists no global type-only encoding of $\STLCVarSubCo$ in
  $\STLCVarRowPre$.
\end{restatable}

The idea of the proof is the same as that of \Cref{thm:swapping} which
we have shown in \Cref{sec:swapping-row-and-pre}: construct the
schemes of type-only translations for certain terms and derive a
contradiction.
The terms we choose here are the nested variant $M =
(\ell\,(\ell\,y)^{\Variant{\ell}})^{\Variant{\ell: \Variant{\ell}}}$
for some free term variable $y$ in the environment together with its
upcast $M_1 = M\Upcast \Variant{\ell:\Variant{\ell;\ell'}}$ and its
case split $M_2 = \Case\ M\ \{\ell\ x \mapsto x\Upcast
\Variant{\ell;\ell'}\}$, similar to the counterexamples we give in
\Cref{sec:example-co-variant-subtyping}.
To obtain a contradiction, we show that we cannot give a uniform
type-only translation of $M$ such that both $M_1$ and $M_2$ can be
translated compositionally.
The details of the proof can be found in
\Cref{app:non-existence-stlcvarsubco}.

As a corollary, there can be no global type-only encoding of
\STLCVarSubFull in \STLCVarRowPre.

One might worry that \Cref{thm:globalVarSubCo} contradicts the duality
between records and variants, especially in light of
\citet{BlumeAC06}'s translation from variants with default cases to
records with record extensions. In their translation, a variant is
translated to a function which takes a record of functions. For
instance, the translation of variant types is:
\begin{equations}
  \transl{\Variant{\ell_i : A_i}_i} = \forall\alpha.\Record{\ell_i:A_i\to\alpha}_i\to\alpha
\end{equations}
In fact, there is no contradiction because a variant in a covariant
position corresponds to a record in a contravariant position, which
means that the encoding of $\STLCRecSubCo$ in
\Cref{sec:encode-co-subtyping} cannot be used. Moreover, the
translation from variants to records is not type-only as it introduces
$\lambda$-abstractions.

\subsection{Non-Existence of Type-Only Encodings of \STLCRecSubFull in \STLCRecRowPre}
\label{sec:non-full-subtyping}

As illustrated by the examples $\var{getName'}_\xmark$ and
$\var{getUnit}_\xmark$ in \Cref{sec:example-full-record-subtyping},
one attempt to simulate full record subtyping by both making record
types presence-polymorphic and adding row variables for records in
contravariant positions fails. In fact no such encoding exists.

\begin{restatable}{theorem}{globalSubFull}
  \label{thm:globalSubFull}
  There exists no global type-only encoding of $\STLCRecSubFull$ in
  $\STLCRecRowPre$.
\end{restatable}
Again, the proof idea is to give general forms of type-only
translations for certain terms and proof by contradiction.
Our choice of terms here are different from the counterexamples in
\Cref{sec:example-full-record-subtyping} this time.
Instead, we first consider two functions $f_1 = \lambda x^{\Record{}}
. x$ and $f_2 = \lambda x^{\Record{}} . \Record{}$ of the same type
$\Record{} \to \Record{}$.
Any type-only translations of these functions must yield terms of the following forms:
\[\begin{array}{l@{}l@{}c@{}l}
\transl{f_1} &= \Lambda \ol{\alpha}_1 . \lambda
x^{A_1} . \Lambda \ol{\beta}_1 . &x&\ \ol B_1 \\
\transl{f_2} &= \Lambda
\ol{\alpha}_2 . \lambda x^{A_2} . \Lambda \ol{\beta}_2 . &\transl{\Record{}}&\ \ol B_2
= \Lambda \ol{\alpha}_2 . \lambda x^{A_2} . \Lambda \ol{\beta}_2 . (\Lambda
\ol{\gamma} . \Record{})\ \ol B_2
\end{array}\]
By type preservation, they should have the same type, which means $x\
\ol B_1$ and $(\Lambda \ol{\gamma} . \Record{})\ \ol B_2$ should also
have the same type.
As a result, the type $A_1$ of $x$ cannot contain any type variables
bound in $\ol{\alpha}_1$ unless they are inside the type of some labels
which are instantiated to absent by the type application $x\ \ol B_1$.
Then, it is problematic when we want to upcast the parameter of $f_1$
to be a wider record, e.g., $f_1 \Upcast (\Record{\ell:\Record{}} \to
\Record{})$.
Intuitively, because $A_1$ cannot be an open record type with the row
variable bound in $\ol{\alpha}_1$, we actually have no way to expand
$A_1$, which leads to a contradiction.
The full proof can be found in \Cref{app:non-existence-full}.

\section{Full Subtyping as Rank-1 Polymorphism}
\label{sec:prenex-polymorphism}
\label{SEC:PRENEX}

In \Cref{sec:swapping-row-and-pre}, we showed that no type-only
encoding of record subtyping as row polymorphism exists.
The main obstacle is a lack of type information for instantiation.
By focusing on rank-1 polymorphism in the target language, we need no
longer concern ourselves with type abstraction and application
explicitly anymore.
Instead we defer to Hindley-Milner type inference~\citep{DamasMilner}
as demonstrated by the examples in \Cref{sec:example-prenex}.
In this section, we formalise the encodings of full subtyping as
rank-1 polymorphism.

Here we focus on the encoding of \STLCRecSubFull\ in \STLCRecPrenex, a
ML-style calculus with records and rank-1 row polymorphism (the same
idea applies to each combination of encoding records or variants as
rank-1 row polymorphism or rank-1 presence polymorphism).
The specification of \STLCRecPrenex\ is given in
\Cref{app:stlcrecprenex}, which uses a standard declarative
Hindley-Milner style type system and extends the term syntax with
let-binding $\Let\,x=M\,\In\,N$ for polymorphism.
We also extend \STLCRecSubFull\ with let-binding syntax and its
standard typing and operational semantics rules.

As demonstrated in \Cref{sec:example-prenex}, we can use the following
(local and type-only) erasure translation to encode
\STLCRecSubFullCond{2}, the fragment of \STLCRecSubFull where types
are restricted to have rank-2 records, in \STLCRecPrenex.

\noindent
  \begin{equations}
    \transl{-} &:& \Deriv \to \Terms\\
    \transl{M \Upcast A} &=& M
  \end{equations}

Since the types of translated terms in \STLCRecPrenex are given by
type inference, we do not need to use a translation on types in the
translation on terms.
Moreover, we implicitly allow type annotations on
$\lambda$-abstractions to be erased as they no longer exist in the
target language.

To formalise the definition of rank-$n$ records defined in
\Cref{sec:example-prenex}, we introduce the predicate $\recrank{n}{A}$
defined as follows for any natural number $n$.

\begin{equations}
  \begin{split}
  \recrank{n}{\alpha} &= \True \\
  \recrank{n}{A \to B} &= \recrank{n-1}{A} \land  \recrank{n}{B} \\
  \recrank{n}{\Record{\ell_i : A_i}_i} &= \land_i \recrank{n}{A_i}
  \end{split}
  \qquad
  \begin{split}
  \recrank{0}{\alpha} &= \True \\
  \recrank{0}{A \to B} &= \recrank{0}{A} \land \recrank{0}{B} \\
  \recrank{0}{\Record{\ell_i : A_i}_i} &= \False
  \end{split}
\end{equations}

We define a type $A$ to have rank-$n$ records, if $\recrank{n}{A}$
holds.
The predicate $\recrank{n}{A}$ basically means no record types can
appear in the left subtrees of $n$ or more arrows.

The operational correspondence of the erasure translation comes
for free.
Note that both \STLCRecSubFullCond{2} and \STLCRecPrenex are type
erased to untyped \STLCRec.
The type erasure function of \STLCRecSubFullCond{2} inherited from
\STLCVarRecSubFull in \Cref{sec:encode-full-subtyping} is identical to
the erasure translation.
The type erasure function $\erase{-}$ of \STLCRecPrenex is simply the
identity function (as there is no type annotation at all).
We have the following theorem.

\begin{restatable}[Operational Correspondence]{theorem}{STLCRecRowPrenexOC}
  The translation $\transl{-}$ from \STLCRecSubFullCond{2} to
  \STLCRecPrenex satisfies the equation $\erase{M} =
  \erase{\transl{M}}$ for any well-typed term $M$ in
  \STLCRecSubFullCond{2}.
\end{restatable}
\begin{proof}
  By definition of $\erase{-}$ and $\transl{-}$.
\end{proof}

Proving type preservation is more challenging.
To avoid the complexity of reasoning about type inference, we state
the type preservation theorem using the declarative type system of
\STLCRecPrenex, which requires us to give translations on types.
We define the translations on types and environments in
\Cref{fig:encode-stlcrecsubfull-stlcrecprenex-type}.
As in \Cref{sec:encode-stlcrecsub-stlcrecpre} and
\Cref{sec:encode-co-subtyping}, we assume a canonical order on labels
and require all rows and records to conform to this order.
The translation on type environments is still the identity
$\transl{\Delta} = \Delta$.
To define the translation on term environments, we need to explicitly
distinguish between variables bound by $\lambda$ and variables bound
by $\Let$. We write $a,b$ for the former, and $x,y$ for the latter.
Because the translation on term environments may introduce fresh free
type variables which are not in the original type environments, we
define $\transl{\Delta;\Gamma}$ as a shortcut for
$(\transl{\Delta},\ftv{\transl{\Gamma}});\transl{\Gamma}$.

The type translation $\transl{A}$ returns a type scheme.
It
                               opens up row types in $A$ that appear
strictly covariantly inside the left-hand-side of strictly covariant
function types, 
binding all of the freshly generated row
variables at the top-level.
It
  applies 
           the auxiliary translation $\translb{A}$
to function parameter types, 
similarly extending all
record types appearing strictly covariantly in $A$ with fresh row
variables, and binding them all at the top-level.

We define four auxiliary functions for the translation.
The functions $\transrowa{\rho}{A}$ and $\transrowb{\rho}{A}$ are used
to generate fresh row variables.
The $\transrowa{\rho}{A}$ takes a row variable $\rho$ and a type $A$,
and generates a sequence of row variables based on $\rho$ with the
same length of row variables bound by $\transla{A}$.
The function $\transrowb{\rho}{A}$ does the same thing for
$\translb{A}$.
The functions $\transla{A,\ol{\rho}}$ and $\translb{A,\ol{\rho}}$
instantiate polymorphic types, simulating term-level type application.
As we discussed in \Cref{sec:encode-co-subtyping}, these functions
actually break the compositionality of the type translation, because
they must inspect the concrete structure of $\transl{A}$.
However, we only use the type translation in the theorem and proof;
the compositionality of the erasure translation itself remains intact.

\begin{figure}[tp]
  \noindent
  \begin{minipage}[t]{.5\textwidth}\small
  \raggedright
  \begin{equations}
    \transla{-} &:& \Types \to \TyScheme\\
    \transla{A \to B}
    &=& \forall \ol{\rho}_1 \ol{\rho}_2 . \translb{A,\ol{\rho}_1} \to \transla{B, \ol{\rho}_2} \\
    \multicolumn{3}{r}{\text{where }
        \ba[t]{@{~}l@{~}l}
        \ol{\rho}_1 = \transrowb{\rho_1}{A},\
        \ol{\rho}_2 = \transrowa{\rho_2}{B} \\
        \ea
    }\\
    \transla{\Record{\ell_i : A_i}_i}
    &=& \forall (\ol{\rho}_i)_i.
        \Record{\ell_i : \transla{A_i, \ol{\rho}_i}}_i \\
    \multicolumn{3}{r}{\text{where }
        \ba[t]{@{~}l@{~}l}
        \ol{\rho}_i = \transrowa{\rho_i}{A_i} \\
        \ea
    }\\[2ex]

    \transla{-,-} &:& (\Types,\ol{\RowVars}) \to \Type\\
    \transla{A, \ol{\rho}}
    &=& A'[\ol{\rho} / \ol{\rho}']
    \ \text{where}\ \forall \ol{\rho}'.A' = \transla{A}
    \\[2ex]

    \transrowa{-}{-} &:& (\RowVars, \Type) \to \ol{\RowVars}\\
    \transrowa{\rho}{\alpha}
    &=& \cdot
    \\
    \transrowa{\rho}{A \to B}
    &=& \transrowb{\rho_1}{A}\, \transrowa{\rho_2}{B}
    \\
    \transrowa{\rho}{\Record{\ell_i : A_i}_i}
    &=& \transrowa{\rho_i}{A_i}_i
    \\[2ex]

    \transl{-} &:& \Envs \to \Envs\\
    \transl{\cdot}
    &=& \cdot \\
    \transl{\Gamma, x:A}
    &=& \transl{\Gamma}, x:\transla{A}\\
    \transl{\Gamma, a:A}
    &=& \transl{\Gamma}, a:\translb{A, \transrowb{\rho_{|\Gamma|}}{A}}\\
  \end{equations}
  \end{minipage}%
  \begin{minipage}[t]{.5\textwidth}\small
  \raggedleft
  \begin{equations}
    \translb{-} &:& \Types \to \TyScheme\\
    \translb{A \to B}
    &=& \forall \ol{\rho} . A \to \translb{B, \ol{\rho}}\\
    \multicolumn{3}{r}{\text{where }
        \ba[t]{@{~}l@{~}l}
        \ol{\rho} = \transrowb{\rho}{B} \\
        \ea
    }\\
    \translb{\Record{\ell_i : A_i}_i}
    &=& \forall \rho\, (\ol{\rho}_i)_i.
        \Record{\ell_i : \translb{A_i, \ol{\rho}_i}; \rho}_i \\
    \multicolumn{3}{r}{\text{where }
        \ba[t]{@{~}l@{~}l}
        \ol{\rho}_i = \transrowb{\rho_i}{A_i} \\
        \ea
    }\\[2ex]

    \translb{-,-} &:& (\Types,\ol{\RowVars}) \to \Type\\
    \translb{A, \ol{\rho}}
    &=& A'[\ol{\rho} / \ol{\rho}']
    \ \text{where}\ \forall \ol{\rho}'.A' = \translb{A}
    \\[2ex]

    \transrowb{-}{-} &:& (\RowVars, \Type) \to \ol{\RowVars}\\
    \transrowb{\rho}{\alpha}
    &=& \cdot
    \\
    \transrowb{\rho}{A \to B}
    &=& \transrowb{\rho}{B}
    \\
    \transrowb{\rho}{\Record{\ell_i : A_i}_i}
    &=& \rho\,\transrowb{\rho_i}{A_i}_i
    \\
  \end{equations}
  \end{minipage}

  \caption{The translations of types and environments from \STLCRecSubFullCond{2} to \STLCRecPrenex.}
  \label{fig:encode-stlcrecsubfull-stlcrecprenex-type}
\end{figure}

After giving the type and environment translation, we aim for a weak
type preservation theorem which allows the translated terms to have
subtypes of the original terms, because the erasure translation
ignores all upcasts. As we have row variables in \STLCRecPrenex, the
types of translated terms may contain extra row variables in strictly
covariant positions. We need to define an auxiliary subtype relation
$\subtyperow$ which only considers row variables. {\small
\begin{mathpar}
  \inferrule
    { }
    {\alpha \subtyperow \alpha}

  \inferrule
    {[A_i \subtyperow A_i']_i}
    {\Record{\ell_i:A_i}_i \subtyperow \Record{\ell_i:A_i'}_i}

  \inferrule
    {[A_i \subtyperow A_i']_i}
    {\Record{(\ell_i:A_i)_i;\rho} \subtyperow \Record{\ell_i:A_i'}_i}

  \inferrule
    {B \subtyperow B'}
    {A\to B \subtyperow A\to B'}

  \inferrule
    {\tau \subtyperow \tau'}
    {\forall \rho^K . \tau \subtyperow \forall \rho^K . \tau'}
\end{mathpar}}
Finally, we have the following weak type preservation theorem.
\begin{restatable}[Weak Type Preservation]{theorem}{STLCRecRowPrenexType}
  \label{thm:type-rec-rowprenex}
  Every well-typed \STLCRecSubFullCond{2} term $\typ{\Delta;\Gamma}{M : A}$ is
  translated to a well-typed \STLCRecPrenex term
  $\typ{\transl{\Delta;\Gamma}}{\transla{M} : \tau}$ for some $A'\subtype A$ and
  $\tau \subtyperow \transla{A'}$.
\end{restatable}

The proof makes use of \STLCRecSubFullACond{2}, an algorithmic variant
of the type system of \STLCRecSubFullCond{2} which combines
\tylab{App} and \tylab{Upcast} into one rule \tylab{AppSub}, and
removes all explicit upcasts in terms.
\begin{mathpar}
  \inferrule*[Lab=\tylab{AppSub}]
          { \typ{\Delta;\Gamma}{M : A \to B} \\ \typ{\Delta;\Gamma}{N : A'} \\
            A' \subtype A }
          {\typ{\Delta;\Gamma}{M\,N : B}}
\end{mathpar}
It is standard that \STLCRecSubFullACond{2} is sound and complete with
respect to \STLCRecSubFullCond{2} \citep{tapl}.
Immediately, we have that $\typ{\Delta;\Gamma}{M:A}$ in
\STLCRecSubFullCond{2} implies $\typ{\Delta;\Gamma}{\widehat{M}:A'}$ in
\STLCRecSubFullACond{2} for some $A'\subtype A$, where $\widehat{M}$
is defined as $M$ with all upcasts erased.
Thus, we only need to prove that $\typ{\Delta;\Gamma}{M:A}$ in
\STLCRecSubFullACond{2} implies
$\typ{\transl{\Delta;\Gamma}}{\transl{M}:\tau}$ for some $\tau
\subtyperow \transl{A}$ in \STLCRecPrenex.
The remaining proof can be done by induction on the typing derivations
in \STLCRecSubFullACond{2}, where the most non-trivial case is the
\tylab{AppSub} rule.
The core idea is to use instantiation in \STLCRecPrenex to simulate
the subtyping relation $A' \subtype A$ in the \tylab{AppSub} rule.
This is possible because the source language \STLCRecSubFullACond{2}
is restricted to have rank-$2$ records, which implies that $A\to B$ is
translated to a polymorphic type where the record types in parameters
are open and can be extended to simulate the subtyping relation.
The full proof can be found in
\Cref{app:proof-stlcrecsubfull-stlcrecprenex}.

So far, we have formalised the erasure translation from
\STLCRecSubFullCond{2} to \STLCRecPrenex.
As shown in \Cref{sec:example-prenex}, we have three other results.
For records, we have another erasure translation from \STLCRecSubFullCond{1},
the fragment of \STLCRecSubFull where types are restricted to have rank-1
records, to \STLCRecPrePrenex with rank-1 presence polymorphism.
Similarly, for variants, we formally define a type $A$ to have
rank-$n$ variants, if the predicate $\varrank{n}{A}$ defined as
follows holds.
\begin{equations}
  \begin{split}
  \varrank{n}{\alpha} &= \True \\
  \varrank{n}{A \to B} &= \varrank{n-1}{A} \land  \varrank{n}{B} \\
  \varrank{n}{\Variant{\ell_i : A_i}_i} &= \land_i \varrank{n}{A_i}
  \end{split}
  \qquad
  \begin{split}
  \varrank{0}{\alpha} &= \True \\
  \varrank{0}{A \to B} &= \varrank{0}{A} \land \varrank{0}{B} \\
  \varrank{0}{\Variant{\ell_i : A_i}_i} &= \False
  \end{split}
\end{equations}
We also have two erasure translations from \STLCVarSubFullCond{1}
to
\STLCVarPrenex
and from \STLCVarSubFullCond{2}
to \STLCVarPrePrenex.
They all use the same idea to let type inference infer
row/presence-polymorphic types for terms involving records/variants,
and use instantiation to automatically simulate subtyping.
We omit the metatheory of these three results as they are similar to
what we have seen for the encoding of \STLCRecSubFullCond{2} in
\STLCRecPrenex.

The requirement of rank-1 polymorphism and Hindley-Milner type
inference for target languages is not mandatory; target languages can
support more advanced type inference for higher-rank polymorphism like
FreezeML~\citep{FreezeML}, as long as no type annotation is needed to
infer rank-1 polymorphic types.
One might hope to also relax the $\nocontratwice{-}$ restriction in
\STLCRecSubFullCond{2} via enabling higher-rank polymorphism.
However, at least the erasure translation do not work anymore.
For instance, consider the functions $\var{id}=\lambda x^{\Record{\ell:\Int}} .
x$ and $\var{const}=\lambda x^{\Record{\ell:\Int}} . \Record{\ell = 1}$ with the
same type $\Record{\ell:\Int} \to \Record{\ell:\Int}$.
Type inference would give $\transl{\var{id}}$ the type $\forall
\rho^{\Row_{\{\ell\}}}
. \Record{\ell:\Int;\rho}\to\Record{\ell:\Int;\rho}$, and
$\transl{\var{const}}$ the type $\forall \rho^{\Row_{\{\ell\}}} .
\Record{\ell:\Int;\rho}\to\Record{\ell:\Int}$.
For a second-order function of type $(\Record{\ell:\Int} \to
\Record{\ell:\Int}) \to A$, we cannot give a type to the parameter of
the function after translation which can be unified with the types of
both $\transl{\var{id}}$ and $\transl{\var{const}}$.
We leave it to future work to explore whether there exist other
translations making use of type inference for higher-rank
polymorphism.

\section{Discussion}
\label{sec:discussion}

We have now explored a range of encodings of structural subtyping for
variants and records as parametric polymorphism under different
conditions.
These encodings and non-existence results capture the extent to which
row and presence polymorphism can simulate structural subtyping and
crystallise longstanding folklore and informal intuitions.
In the remainder of this section we briefly discuss record extensions
and default cases (\Cref{sec:record-extension}), combining subtyping
and polymorphism (\Cref{sec:combine-sub-and-poly}),
related work (\Cref{sec:related-work}) and conclusions and future work
(\Cref{sec:conclusion}).

\subsection{Record Extensions and Default Cases}
\label{sec:record-extension}

Two important extensions to row and presence polymorphism are record
extensions \citep{Remy94}, and its dual, default cases
\citep{BlumeAC06}.
These operations provide extra expressiveness beyond structural
subtyping.
For example, with default cases, we can give a default age $42$ to the
function $\var{getAge}$ in
\Cref{sec:example-simple-variant-subtyping}, and then apply it to
variants with arbitrary constructors.
\begin{examples}
  &\var{getAgeD} :\: \, \forall\rho^{\Row_{\{\Age,\Year\}}}.
    {\Variant{\Age:\Int;\Year:\Int; \rho}}\to \Int \\
  &\var{getAgeD} = \lambda x.
    ~\Case\,x\,\{\Age\, y\mapsto y; \Year\,y\mapsto 2023-y; z\mapsto 42\} \\
  &\var{getAgeD}\ (\Name\, \makestring{Carol}) \breducestocl^\ast 42
\end{examples}
\vspace{-1\baselineskip}

\subsection{Combining Subtyping and Polymorphism}
\label{sec:combine-sub-and-poly}

Though row and presence polymorphism can simulate subtyping well and
support expressive extensions like record extension and default cases,
it can still be beneficial to allow both subtyping and polymorphism
together in the same language.
For example, the OCaml programming language combines row and presence
polymorphism with subtyping. Row and presence variables are hidden in
its core language.
It supports both polymorphic variants and polymorphic objects (a
variation on polymorphic records) as well as explicit upcast for
closed variants and records.
Our results give a rationalisation for why OCaml supports subtyping in
addition to row polymorphism. Row polymorphism simply is not expressive
enough to give a local encoding of unrestricted structural subtyping,
even though OCaml indirectly supports full first-class polymorphism.

Bounded quantification~\citep{CardelliW85,CardelliMMS94} extends
system F with subtyping by introducing subtyping bounds to type
variables.
There is also much work on the type inference for both polymorphism
and subtyping based on collecting, solving, and simplifying
constraints~\citep{TrifonovS96,pottier:inria-00073205,Pottier01}.
Algebraic subtyping \citep{DolanM17, Dolan16} combines subtyping and
parametric polymorphism, offering compact principal types and
decidable subsumption checking.
MLstruct \citep{ParreauxC22} extends algebraic subtyping with
intersection and union types, giving rise to another alternative to
row polymorphism.

\subsection{Related Work}
\label{sec:related-work}

\paragraph{Row types.}
\citet{Wand87} first introduced rows and row polymorphism.
There are many further papers on row types, which take a variety of
approaches, particularly focusing on extensible records.
\citet{Harper90} extended System F with constrained quantification,
where predicates $\rho\,\mathsf{lacks}\,L$ and $\rho\,\mathsf{has}\,L$
are used to indicate the presence and absence of labels in row
variables.
\citet{Gaster96} and \citet{Gaster98} explored a calculus with a
similar $\mathsf{lacks}$ predicate based on qualified types.
\citet{Remy89} introduced the concept of presence types and
polymorphism, and \citet{Remy94} combines row and presence
polymorphism.
\citet{Leijen05} proposed a variation on row polymorphism with support
for scoped labels.
\citet{pottier04ML} considered type inference for row and presence
polymorphism in HM(X).
\citet{MorrisM19} introduced \textsc{Rose}, an algebraic foundation for
row typing via a rather general language with two predicates
representing the containment and combination of rows. It is parametric
over a row theory which enables it to express different styles of row
types (including \citeauthor{Wand87} and \citeauthor{Remy94}'s style
and \citeauthor{Leijen05}'s style).

\paragraph{Row polymorphism vs structural subtyping.}

\citet{Wand87} compared his calculus with row polymorphism (similar to
\stlc{\Variantonly\Recordonly}{\rho 1}) with \citet{Cardelli84}'s calculus
with structural subtyping (similar to \STLCVarRecSubFull) and showed
that they cannot be encoded in each other by examples.
\citet{pottier:inria-00073205} conveyed the intuition that row
polymorphism can lessen the need for subtyping to some extent, but
there are still situations where subtyping are necessary, e.g., the
reuse of $\lambda$-bound variables which cannot be polymorphic given
only rank-1 polymorphism.

\paragraph{Disjoint polymorphism.}
Disjoint intersection types~\citep{OliveiraSA16} generalise record
types.
Record concatenation and restriction~\citep{CardelliM91} are replaced
by a merge operator~\citep{Dunfield14} and a type difference
operator~\citep{XuHO23}, respectively.
Parametric polymorphism of disjoint intersection types is supported
via disjoint polymorphism~\citep{AlpuimOS17} where type variables are
associated with disjointness constraints.
Similarly to our work, \citet{XieOBS19} formally prove that both row
polymorphism and bounded quantification of record types can be encoded
in terms of disjoint polymorphism.

\subsection{Conclusion and Future Work}
\label{sec:conclusion}

We carried out a formal and systematic study of the encoding of
structural subtyping as parametric polymorphism.
To better reveal the relative expressive power of these two type
system features, we introduced the notion of type-only translations to
avoid the influence of non-trivial term reconstruction.
We gave type-only translations from various calculi with subtyping to
calculi with different kinds of polymorphism and proved their
correctness; we also proved a series of non-existence results.
Our results provide a precise characterisation of the long-standing
folklore intuition that row polymorphism can often replace subtyping.
Additionally, they offer insight into the trade-offs between subtyping
and polymorphism in the design of programming languages.

In future, we would like to explore whether it might be possible to
extend our encodings relying on type inference to systems supporting
higher-rank polymorphism, such as FreezeML~\citep{FreezeML}.
We would also like to consider other styles of row typing such as
those based on scoped labels \citep{Leijen05} and \textsc{Rose}
\citep{MorrisM19}.
In addition to variant and record types, row types are also the
foundation for various effect type systems, e.g. for effect
handlers~\citep{HillerstromL16, Leijen17}.
It would be interesting to investigate to what extent our approach can
be applied to effect typing.
Aside from studying the relationship between subtyping and row and
presence polymorphism we would also like to study the ergonomics of
these programming language features in practice, especially their
compatibility with others such as algebraic data types.

\begin{acks}                            %
   This work was supported by the UKRI Future Leaders Fellowship
  ``Effect Handler Oriented Programming'' (reference number
  MR/T043830/1) and ERC Consolidator Grant no. 682315 (Skye).
\end{acks}

\bibliography{\jobname}

\clearpage

\appendix

\section{More Calculi}
\label{app:more-calculi}

In this section, we elaborate on calculi that are not fully detailed
in the body of the paper.

\subsection{A Calculus with Row Polymorphic Records $\STLCRecRow$}
\label{sec:stlcrecrow}

The extensions to the syntax, static semantics, and dynamic semantics of
$\STLCRec$ for a calculus with row polymorphic records are shown in
\Cref{fig:stlcrecrow}. Actually, they are exactly the same as the extensions to
$\STLCVar$ for $\STLCVarRow$ in \Cref{fig:var-calc-row}.

\begin{figure}[h]
  \flushleft
  \textbf{Syntax}

  \vspace{-\the\baselineskip}
  \begin{minipage}{0.6\textwidth}
    \begin{syntax}
      \slab{Types} &\Belong{A}{\Types}     &::=& \ldots \mid \forall \rho^K.A \\
      \slab{Rows}  &\Belong{R}{\Rows} &::=& \ldots \mid \rho\\
    \end{syntax}
  \end{minipage}
  \begin{minipage}{0.4\textwidth}
    \begin{syntax}
      \slab{Terms} &\Belong{M}{\Terms}     &::=& \ldots \mid \Lambda \rho^K . M \mid M\,R\\
      \slab{Type~Environments}&\Belong{\Delta}{\TyEnvs} &::=& \ldots \mid \Delta, \rho : K \\
    \end{syntax}
  \end{minipage}

  \textbf{Static Semantics}\medskip\\
  \def \MathparLineskip {\lineskip=8pt}
  \begin{minipage}[t]{0.45\textwidth}
  \fbox{$\typ{\Delta\phantom{;\!}}{A : K}$}
  \begin{mathpar}
    \inferrule*[Lab=\klab{RowVar}]
      {~}
      {\typ{\Delta, \rho : \Row_\LS}{\rho : \Row_\LS}} \\

    \inferrule*[Lab=\klab{RowAll}]
      {\typ{\Delta,\rho : \Row_\LS}{A : \Type}}
      {\typ{\Delta}{\forall \rho^{\Row_\LS} .  A : \Type}}
  \end{mathpar}
  \end{minipage}
  \begin{minipage}[t]{0.5\textwidth}
  \fbox{$\typ{\Delta;\Gamma}{M : A}$}
  \begin{mathpar}
    \inferrule*[Lab=\tylab{RowLam}]
       {\typ{\Delta, \rho : K;\Gamma}{M : A} \\ \rho \notin \ftv{\Gamma}}
       {\typ{\Delta;\Gamma}{\Lambda \rho^K.M : \forall \rho^K.A}} \\

    \inferrule*[Lab=\tylab{RowApp}]
       {\typ{\Delta;\Gamma}{M : \forall \rho^K.B} \\ \typ{\Delta}{A : K}}
       {\typ{\Delta;\Gamma}{M\,A : B[A/\rho]}}
  \end{mathpar}
  \end{minipage}

  \textbf{Dynamic Semantics}
  \begin{reductions}
    \tlab{RowLam} & (\Lambda \rho^K.M)\,R &\treducesto& M[R/\rho] \\
  \end{reductions}
  \caption{Extensions of \STLCRec with row polymorphism \STLCRecRow}
  \label{fig:stlcrecrow}
\end{figure}

\subsection{A Calculus with Presence Polymorphic Variants \STLCVarPre}
\label{sec:stlcvarpre}

The extensions and modifications to the syntax, static semantics, and dynamic
semantics of $\STLCVar$ for a calculus with presence polymorphic variants
$\STLCVarPre$ are shown in \Cref{fig:stlcvarpre}.

One thing worth noting is that in \tylab{Case}, we do not require all
labels in the type of $M$ to be present, which is dual to the
\tylab{Record} rule in \Cref{fig:rec-calc-pre}. It does not loss any
generality as our equivalence relation between rows only considers
present labels.

\begin{figure}[h]
\flushleft
\noindent\textbf{Syntax}
\begin{minipage}{0.45\textwidth}
  \begin{syntax}
  \slab{Kinds} & \Belong{K}{\Kinds}   &::=& \ldots %
                                      \mid \Presence \\
  \slab{Types} &  \Belong{A}{\Types}     &::=& \ldots %
                                      \mid \forall \theta.A \\
  \slab{Rows}  &\Belong{R}{\Rows} &::=& \ldots \mid \hlmod{\ell^{P} : A; R} \\
  \end{syntax}
\end{minipage}
\begin{minipage}{0.55\textwidth}
  \begin{syntax}
  \slab{Presence} & \Belong{P}{\Presences} &::=& \Absent \mid \Present \mid \theta \\
  \slab{Terms}& \Belong{M}{\Terms} &::=& \ldots
        \mid \Lambda \theta . M \mid M\,P \\
  \slab{Type~Environments}& \Belong{\Delta}{\TyEnvs} &::=& \ldots \mid \Delta, \theta \\
  \end{syntax}
\end{minipage}

\textbf{Static Semantics} \medskip\\
\fbox{$\typ{\Delta\phantom{;\!}}{A : K}$}
\def \MathparLineskip {\lineskip=8pt}
\vspace{-2em}
\begin{mathpar}
 \inferrule*[Lab=\klab{Absent}]
           {~}
           {\typ{\Delta}{\Absent : \Presence}}

 \inferrule*[Lab=\klab{Present}]
           {~}
           {\typ{\Delta}{\Present : \Presence}}

 \inferrule*[Lab=\klab{PreVar}]
           {~}
           {\typ{\Delta, \theta}{\theta : \Presence}}

 \inferrule*[Lab=\klab{PreAll}]
           {\typ{\Delta, \theta}{A : \Type}}
           {\typ{\Delta}{\forall \theta.A : \Type}}

  \hlmod{
  \inferrule*[Lab=\klab{ExtendRow}]
           {
            \typ{\Delta}{P : \Presence} \\\\
            \typ{\Delta}{A : \Type} \\\\
            \typ{\Delta}{R : \Row_{\LS \uplus \{\ell\}}}
          }
           {\typ{\Delta}{\ell^P : A; R : \Row_{\LS}}}
  }
\end{mathpar}
\fbox{$\typ{\Delta;\Gamma}{M : A}$}
\begin{mathpar}
 \inferrule*[Lab=\tylab{PreLam}]
           {\typ{\Delta,\theta;\Gamma}{M : A} \\ \theta \notin \ftv{\Gamma}}
           {\typ{\Delta;\Gamma}{\Lambda \theta.M : \forall \theta.A}}

 \inferrule*[Lab=\tylab{PreApp}]
           {\typ{\Delta;\Gamma}{M : \forall \theta.A}\\
            \Delta\vdash P:\Presence}
           {\typ{\Delta;\Gamma}{M\,P : A[P/\theta]}}

  \hlmod{
  \inferrule*[Lab=\tylab{Inject}]
  {(\ell^\Present : A) \in R \\ \typ{\Delta;\Gamma}{M : A}}
  {\typ{\Delta;\Gamma}{(\ell\,M)^{\Variant{R}} : \Variant{R}}}

  \inferrule*[Lab=\tylab{Case}]
   {\typ{\Delta;\Gamma}{M : \Variant{\ell_i^{P_i}:A_i}_i} \\
    [\typ{\Delta;\Gamma,x_i : A_i}{N_i : B}]_i}
   {\typ{\Delta;\Gamma}{\Case~M~\{\ell_i~x_i \mapsto N_i\}_i : B}}
  }
\end{mathpar}
\textbf{Dynamic Semantics}
\begin{reductions}
  \tbetalab{PreLam} & (\Lambda \theta.M)\,P &\treducesto& M[P/\theta]\\
\end{reductions}

\caption{Extensions and modifications to \STLCVar with presence
polymorphism \STLCVarPre. Highlighted parts replace the old ones in
\STLCVar, rather than extensions.}
\label{fig:stlcvarpre}
\end{figure}

\subsection{A Calculus with Rank-1 Row Polymorphic Records \STLCRecPrenex}
\label{app:stlcrecprenex}

The extensions to the syntax, static semantics, and dynamic semantics for
\STLCRecPrenex, a calculus with records and rank-1 row polymorphism are shown in
\Cref{fig:rec-calc-prenex}.
For the type syntax, we introduce row variables and
type schemes. For the term syntax, we drop the type annotation on $\lambda$
abstractions, and add the $\Let$ syntax for polymorphism. We only give the
declarative typing rules, as the syntax-directed typing rules and type inference
are just standard \citep{DamasMilner}.
Notice that we do not introduce type variables for values in type schemes for
simplicity. The lack of principal types is fine here as we are working with
declarative typing rules. It is easy to regain principal types by adding
value type variables.

\begin{figure}[h]
  \flushleft
  \textbf{Syntax}
  \begin{syntax}
    \slab{Type~Schemes} & \Belong{\tau}{\TyScheme} &::=& A \mid \forall \rho^K . \tau \\
    \slab{Rows}  &\Belong{R}{\Rows} &::=& \ldots \mid \rho \\
    \slab{Terms} & \Belong{M,N}{\Terms} &::=& \ldots \mid \lambda x.M \mid \Let\,x = M \,\In\, N \\
    \slab{Type~Environments} &\Belong{\Delta}{\TyEnvs} &::=& \ldots \mid \Delta, \rho:K \\
    \slab{Term~Environments} & \Belong{\Gamma}{\Envs} &::=& \hlmod{\cdot \mid \Gamma, x : \tau}\\
  \end{syntax}

  \textbf{Static Semantics} \medskip\\
  \qquad\fbox{$\typ{\Delta\phantom{;\!}}{A : K}$}
  \begin{mathpar}
  \inferrule*[Lab=\klab{RowVar}]
  {~}
  {\typ{\Delta, \rho : \Row_\LS}{\rho : \Row_\LS}}

  \inferrule*[Lab=\klab{RowAll}]
  {\typ{\Delta,\rho : \Row_\LS}{A : \Type}}
  {\typ{\Delta}{\forall \rho^{\Row_\LS} .  A : \Type}}
  \end{mathpar}

  \qquad \fbox{$\typ{\Delta;\Gamma}{M : A}$}
  \begin{mathpar}
    \hlmod{
    \inferrule*[Lab=\tylab{Lam}]
      {\typ{\Delta;\Gamma, x : A}{M : B}}
      {\typ{\Delta;\Gamma}{\lambda x. M : A \to B}}
    }

    \inferrule*[Lab=\tylab{Let}]
      { \typ{\Delta;\Gamma}{M : \tau} \\
        \typ{\Delta;\Gamma, x:\tau}{N : A}
      }
      {\typ{\Delta;\Gamma}{\Let\, x = M \,\In\, N : A}}

    \inferrule*[Lab=\tylab{Inst}]
      { \typ{\Delta;\Gamma}{M : \forall \rho^{\Row_\LS} . \tau} \\
        \typ{\Delta}{R : \Row_\LS}
      }
      {\typ{\Delta;\Gamma}{M : \tau[R/\rho]}}

    \inferrule*[Lab=\tylab{Gen}]
      {\typ{\Delta, \rho : \Row_\LS;\Gamma}{M : \tau} \\ \rho \notin \ftv{\Gamma,\Delta}}
      {\typ{\Delta;\Gamma}{M : \forall \rho^{\Row_\LS} . \tau}}
  \end{mathpar}

  \textbf{Dynamic Semantics}
  \begin{reductions}
    \betalab{Let} & \Let\, x=M \,\In\, N &\breducesto& N[M/x]\\
  \end{reductions}
  \caption{Extensions and modifications to $\STLCRec$ for a calculus
   with rank-1 row polymorphism $\STLCRecPrenex$. Highlighted parts
   replace the old ones in \STLCRec, rather than extensions.}
  \label{fig:rec-calc-prenex}
\end{figure}

\section{Proofs of Encodings in Section~\ref{SEC:TRANS}}

In this section, we show the proofs of type preservation and operational
correspondence for all the four translations in \Cref{sec:translations}.

\subsection{Proof of the Encoding of \STLCVarSub in \STLCVar}
\label{app:proof-stlcvarsub-stlcvar}

\begin{lemma}[Translation commutes with substitution]
  \label{lemma:varsubvar-subst}
  If $\typ{\Delta;\Gamma,x:A}{M:B}$ and $\typ{\Delta;\Gamma}{N:A}$,
  then $\transl{M[N/x]} = \transl{M}[\transl{N}/x]$.
\end{lemma}
\begin{proof}
  By straightforward induction on $M$.
  \begin{description}
    \item[$x$] $\transl{x[N/x]} = \transl{N} = \transl{x}[\transl{N}/x]$.
    \item[$y(y\neq x)$] $\transl{y[N/x]} = y = \transl{y}[\transl{N}/x]$
    \item[$M_1\, M_2$] Our goal follows from IH and definition of substitution.
    \item[$(\ell\,M')^A$] Our goal follows from IH and definition of substitution.
    \item[$\Case\,M'\,\{\ell_i\,x_i\mapsto N_i\}_i$] Our goal follows from IH
    and definition of substitution.
    \item[$M'\Upcast A$] By IH and definition of substitution, we have
    $ \transl{(M^{\Variant{\ell_i:A_i}_i} \Upcast \Variant{R}) [N/x]}
    = \transl{M^{\Variant{\ell_i:A_i}_i}[N/x] \Upcast \Variant{R}}
    = \Case~\transl{M[N/x]}~\{\ell_i~x_i \mapsto (\ell_i~x_i)^{\Variant{R}}\}_{i}
    = \Case~\transl{M}[\transl{N}/x]~\{\ell_i~x_i \mapsto (\ell_i~x_i)^{\Variant{R}}\}_{i}
    = (\Case~\transl{M}~\{\ell_i~x_i \mapsto (\ell_i~x_i)^{\Variant{R}}\}_{i})[\transl{N}/x]
    = \transl{M^{\Variant{\ell_i:A_i}_i} \Upcast \Variant{R}}[\transl{N}/x]
    $.
  \end{description}
\end{proof}

\STLCVarSubSTLCVarTYPE*

\begin{proof}
  By straightforward induction on typing derivations.
  \begin{description}
    \item[\tylab{Var}] Our goal follows from $\transl{x} = x$ and \tylab{Var}.
    \item[\tylab{Lam}] Our goal follows from IH and \tylab{Lam}.
    \item[\tylab{App}] Our goal follows from IH and \tylab{App}.
    \item[\tylab{Inject}] Our goal follows from IH and \tylab{Inject}.
    \item[\tylab{Case}] Our goal follows from IH and \tylab{Case}.
    \item[\tylab{Upcast}] The only subtyping relation in \STLCVarSub is for
    variant types. Given $\typ{\Delta;\Gamma}{M^{\Variant{R}} \Upcast
    \Variant{R'} : \Variant{R'}}$, by $\typ{\Delta;\Gamma}{M:\Variant{R}}$ and
    IH we have $\typ{\transl{\Delta};\transl{\Gamma}}{\transl{M}:\Variant{R}}$.
    Then, supposing $R = (\ell_i:A_i)_i$, by definition of translation,
    $\Variant{R}\subtype\Variant{R'}$ and \tylab{Case} we have
    $\typ{\transl{\Delta};\transl{\Gamma}}{\Case~\transl{M}~\{\ell_i~x_i \mapsto
    (\ell_i~x_i)^{\Variant{R'}}\}_{i} :\Variant{R'}}$.
  \end{description}
\end{proof}

\STLCVarSubSTLCVar*

\begin{proof}~\\
  \noindent\textsc{Simulation}: First, we prove the base case that the whole
    term $M$ is reduced, i.e. $M \bureducesto N$ implies $\transl{M} \breducesto
    \transl{N}$. The proof proceeds by case analysis on the reduction relation:
    \begin{description}
      \item[\betalab{Lam}]  We have $(\lambda x^A . M_1)\, M_2 \breducesto
      M_1[M_2 / x]$. Then, \refa{$\transl{(\lambda x^A . M_1)\, M_2} = (\lambda
      x^A . \transl{M_1})\ \transl{M_2} \breducesto \transl{M_1}[\transl{M_2} /
      x] = \transl{M_1[M_2 / x]}$}, where the last equation follows from
      \Cref{lemma:varsubvar-subst}.
      \item[\betalab{Case}] We have $\Case~M'~\{\ell_i\,x_i \mapsto N_i\}_i
      \breducesto N_j[M_j / x_j]$. Similar to the \betalab{Lam} case.
      \item[\upcastlab{Upcast}] We have $(\ell\,M_1)^{\Variant{R}} \Upcast A
      \ureducesto (\ell\,M_1)^A$. Supposing $R = (\ell_i : A_i)_i$, we have
      \refb{$\transl{(\ell\,M_1)^{\Variant{R}} \Upcast A} = \Case\
      (\ell\,\transl{M_1})^{\Variant{R}}\ \{\ell_i\,x_i \mapsto
      (\ell_i\,x_i)^A\}_i \breducesto (\ell\,\transl{M_1})^A =
      \transl{(\ell\,M_1)^A}$}.
    \end{description}

    Then, we prove the full theorem by induction on $M$. We only need to prove
    the case where reduction happens in sub-terms of $M$.

    \begin{description}
      \item[$x$] No reduction.
      \item[$\lambda x^A . M'$] The reduction can only happen in $M'$. Supposing
      $\lambda x^A . M' \bureducestocl \lambda x^A . N'$, by IH on $M'$, we have
      $\transl{M'} \breducestocl \transl{N'}$, which then gives $\transl{\lambda
      x^A.M'} = \lambda x^A.\transl{M'} \breducestocl \lambda x^A.\transl{N'} =
      \transl{\lambda x^A . N'}$.
      \item[$M_1\, M_2$] Similar to the $\lambda x^A . M'$ case as reduction
      can only happen either in $M_1$ or $M_2$.
      \item[$(\ell\,M')^A$] Similar to the $\lambda x^A . M'$ case as reduction
      can only happen in $M'$.
      \item[$\Case~M'~\{\ell_i\,x_i \mapsto N_i\}_i$] Similar to the $\lambda x^A
      . M'$ case as reduction can only happen in $M'$ or one of $(N_i)_i$.
      \item[$M' \Upcast A$] Similar to the $\lambda x^A . M'$ case as reduction can only
      happen in $M'$.
    \end{description}

  \noindent\textsc{Reflection}: First, we prove the base case that the whole
  term $\transl{M}$ is reduced, i.e. $\transl{M} \breducesto \transl{N}$ implies
  $M \bureducesto N$. The proof proceeds by case analysis on the reduction relation:
    \begin{description}
      \item[\betalab{Lam}] By definition of translation, there exists $M_1$ and
      $M_2$ such that $M = (\lambda x^A . M_1)\ M_2$. Our goal follows from
      \refa{} and $M = (\lambda x^A . M_1)\ M_2 \breducesto M_1[M_2 / x]$.
      \item[\betalab{Case}] By definition of translation, the top-level syntax
      construct of $M$ can either be $\Case$ or upcast. Proceed by a case
      analysis:
      \begin{itemize}[topsep=0em]
        \item $M = \Case~(\ell_j\,M_j)^{\Variant{R}}~\{\ell_i\,x_i \mapsto
        N_i\}_i$ where $R = (\ell_i:A_i)_i$. Similar to the \betalab{Lam} case.
        \item $M = (\ell\,M_1)^{\Variant{R}}\Upcast A$ where $R = (\ell_i :
        A_i)_i$. Our goal follows from \refb{} and
        $(\ell\,M_1)^{\Variant{R}}\Upcast A \ureducesto (\ell\,M_1)^A$.
      \end{itemize}
    \end{description}

    Then, we prove the full theorem by induction on $M$. We only need to prove
    the case where reduction happens in sub-terms of $\transl{M}$.

    \begin{description}
      \item[$x$] No reduction.
      \item[$\lambda x^A . M'$] By definition of translation, there exists $N'$
      such that $N = \lambda x^A . N'$ and $\transl{M'} \breducestocl
      \transl{N'}$. By IH, we have $M' \bureducestocl N'$, which then implies
      $\lambda x^A . M' \bureducestocl \lambda x^A . N'$.
      \item[$M_1\, M_2$] Similar to the $\lambda x^A . M'$ case as reduction
      can only happen either in $\transl{M_1}$ or $\transl{M_2}$.
      \item[$(\ell\,M')^A$] Similar to the $\lambda x^A . M'$ case as reduction
      can only happen in $\transl{M'}$.
      \item[$\Case~M'~\{\ell_i\,x_i \mapsto N_i\}_i$] Similar to the $\lambda
      x^A . M'$ case as reduction can only happen in $\transl{M'}$ or one of
      $(\transl{N_i})_i$.
      \item[$M' \Upcast A$] Similar to the $\lambda x^A . M'$ case as reduction can only
      happen in $\transl{M'}$.
    \end{description}
\end{proof}

\subsection{Proof of the Encoding of \STLCVarSub in \STLCVarRow}
\label{app:proof-stlcvarsub-stlcvarrow}

\begin{lemma}[Translation commutes with substitution]
  \label{lemma:varsubvarrow-subst}
  If $\typ{\Delta;\Gamma,x:A}{M:B}$ and $\typ{\Delta;\Gamma}{N:A}$,
  then $\transl{M[N/x]} = \transl{M}[\transl{N}/x]$.
\end{lemma}
\begin{proof}
  By straightforward induction on $M$. Only consider cases that are
  different from the proof of \Cref{lemma:varsubvar-subst}.
  \begin{description}
    \item[$(\ell\,M')^{\Variant{R}}$] By IH and definition of substitution, we have
    \[\ba{ll}
     & \transl{(\ell\,M)^{\Variant{R}} [N/x]} \\
    =& \transl{(\ell\,M[N/x])^{\Variant{R}}} \\
    =& \Lambda \rho^{\Row_{R}}. (\ell\,\transl{M[N/x]})^{\Variant{\transl{R}; \rho}} \\
    =& \Lambda \rho^{\Row_{R}}. (\ell\,\transl{M}[\transl{N}/x])^{\Variant{\transl{R}; \rho}} \\
    =& (\Lambda \rho^{\Row_{R}}. (\ell\,\transl{M})^{\Variant{\transl{R}; \rho}})[\transl{N}/x] \\
    =& \transl{(\ell\,M)^{\Variant{R}}}[\transl{N}/x]
    \ea\]
    \item[$\Case\,M'\,\{\ell_i\,x_i\mapsto N_i\}_i$] By an equational reasoning similar to the
    case of $(\ell\,M')^{\Variant{R}}$.
    \item[$M'\Upcast A$] By an equational reasoning similar to the case of
    $(\ell\,M')^{\Variant{R}}$.
  \end{description}
\end{proof}

\STLCVarSubSTLCVarRowTYPE*

\begin{proof}
  By induction on typing derivations.
  \begin{description}
    \item[\tylab{Var}] Our goal follows from $\transl{x} = x$.
    \item[\tylab{Lam}] Our goal follows from IH and \tylab{Lam}.
    \item[\tylab{App}] Our goal follows from IH and \tylab{App}.
    \item[\tylab{Inject}] By definition we have $(l:A) \in R$ implies
    $(l:\transl{A})\in \transl{R}\rho$ for any $\rho$. Then our goal follows from
    IH, \tylab{Inject} and \tylab{RowLam}.
    \item[\tylab{Case}] Our goal follows from IH and \tylab{Case}.
    \item[\tylab{Upcast}] The only subtyping relation in \STLCVarSub is for
    variant types. Given $\typ{\Delta;\Gamma}{M^{\Variant{R}} \Upcast
    \Variant{R'} : \Variant{R'}}$, by $\typ{\Delta;\Gamma}{M:\Variant{R}}$ and
    IH we have
    $\typ{\transl{\Delta};\transl{\Gamma}}{\transl{M}:\transl{\Variant{R}}}$.
    Then, by definition of translation and \tylab{RowApp} we have
    $\typ{\transl{\Delta};\transl{\Gamma}}{\transl{M^{\Variant{R}} \Upcast
    \Variant{R'}}:\transl{\Variant{R'}}}$.
  \end{description}
\end{proof}

\STLCVarSubSTLCVarRow*

\begin{proof}~\\
  \noindent\textsc{Simulation}: First, we prove the base case where the whole
  term $M$ is reduced, i.e. $M \breducesto N$ implies $\transl{M}
  \treducestocl^\que\breducestocl \transl{N}$, and $M \ureducesto N$ implies
  $\transl{M} \vreducestocl \transl{N}$. The proof proceeds by case analysis on
  the reduction relation:
  \begin{description}
    \item[\betalab{Lam}] We have $(\lambda x^A . M_1)\, M_2 \breducesto M_1[M_2
    / x]$. Then, \refa{$\transl{(\lambda x^A . M_1)\, M_2} = (\lambda x^A .
    \transl{M_1})\ \transl{M_2} \breducesto \transl{M_1}[\transl{M_2} / x] =
    \transl{M_1[M_2 / x]}$}, where the last equation follows from
    \Cref{lemma:varsubvarrow-subst}.
    \item[\betalab{Case}] We have $\Case\ (\ell_j\,M_j)^{\Variant{R}}\ \{\ell_i\
    x_i \mapsto N_i\} \breducesto N_j[M_j/x_j]$. Supposing $R = (\ell_i :
    A_i)_i$, we have \refb{$\transl{\Case\ (\ell_j\,M_j)^{\Variant{R}}\ \{\ell_i\ x_i
    \mapsto N_i\}} = \Case\ (\transl{(\ell_j\,M_j)^{\Variant{R}}}\,{\cdot})\
    \{\ell_i\ x_i \mapsto \transl{N_i}\} \treducestocl \Case\
    ((\ell_j\,\transl{M_j})^{\Variant{\transl{R}}})\ \{\ell_i\ x_i \mapsto
    \transl{N_i}\} \breducesto \transl{N_j}[\transl{M_j} / x_j] =
    \transl{N_j[M_j / x_j]}$}, where the last equation follows from
    \Cref{lemma:varsubvarrow-subst}.
    \item[\upcastlab{Upcast}] We have $(\ell\,M_1)^{\Variant{R}} \Upcast
    \Variant{R'} \ureducesto (\ell\,M_1)^{\Variant{R'}}$. We have
    \refc{$\transl{(\ell\,M_1)^{\Variant{R}} \Upcast \Variant{R'}} =$\\ $\Lambda
    \rho^{\Row_{R'}}.\transl{(\ell\,M_1)^{\Variant{R}}}\hk{(\transl{R' \backslash
    R}; \rho)} \vreducestocl \Lambda \rho^{\Row_{R'}}.
    (\ell\,M_1)^{\Variant{\transl{R'}; \rho}} =
    \transl{(\ell\,M_1)^{\Variant{R'}}}$}.
  \end{description}

  Then, we prove the full theorem by induction on $M$. We only need to prove the
  case where reduction happens in sub-terms of $M$.
  \begin{description}
    \item[$x$] No reduction.
    \item[$\lambda x^A . M'$] The reduction can only happen in $M'$. Supposing
    $\lambda x^A . M' \breducestocl \lambda x^A . N'$, by IH on $M'$, we have
    $\transl{M'} \treducestocl^\que\breducestocl \transl{N'}$, which then gives
    $\transl{\lambda x^A.M'} = \lambda x^A.\transl{M'}
    \treducestocl^\que\breducestocl \lambda x^A.\transl{N'} = \transl{\lambda x^A
    . N'}$.
    The same applies to the second case of the theorem.
    \item[$(\ell\,M')^{\Variant{R}}$] Similar to the $\lambda x^A . M'$ case as
    reduction can only happen in $M'$.
    \item[$M_1\ M_2$] Similar to the $\lambda x^A . M'$ case as reduction can only
    happen either in $M_1$ or $M_2$.
    \item[$\Case~M'~\{\ell_i\,x_i \mapsto N_i\}_i$] Similar to the $\lambda x^A
    . M'$ case as reduction can only happen in $M'$ or one of $(N_i)_i$.
    \item[$M' \Upcast A$] Similar to the $\lambda x^A . M'$ case as reduction can only
    happen in $M'$.
  \end{description}

  \noindent\textsc{Reflection}:
  We proceed by induction on $M$.
  \begin{description}
    \item[$x$] No reduction.
    \item[$\lambda x^A . M'$] We have $\transl{M} = \lambda
    x^{\transl{A}}.\transl{M'}$. The reduction can only happen in $\transl{M'}$.
    By definition of translation, there exists $N'$ such that $N = \lambda x^A .
    N'$ and $\transl{M'} \treducestocl^\que\breducestocl \transl{N'}$. By IH, we have $M'
    \breducestocl N'$, which then implies $M \breducestocl N$.
    The same applies to the second case of the theorem.
    \item[$M_1\ M_2$] We have $\transl{M} = \transl{M_1}\ \transl{M_2}$.
    Proceed by case analysis where the first step of reduction happens.
    \begin{itemize}[topsep=0em]
      \item Reduction happens in either $\transl{M_1}$ or $\transl{M_2}$.
      Similar to the $\lambda x^A . M'$ case.
      \item The application is reduced by \betalab{Lam}. By
      definition of translation, we have $M_1 = \lambda x^A . M'$. By \refa{},
      we have $\transl{M} \breducesto \transl{M'[M_2 / x]}$, which then gives $N
      = M'[M_2 / x]$. Our goal follows from $M \breducesto N$.
    \end{itemize}
    \item[$(\ell\,M')^{\Variant{R}}$] We have $\transl{M} = \Lambda
    \rho^{\Row_{R}}. (\ell\,\transl{M'})^{\Variant{\transl{R}; \rho}}$. Similar
    to the $\lambda x^A . M'$ case as the reduction can only happen in
    $\transl{M'}$.
    \item[$\Case~M'~\{\ell_i\,x_i \mapsto N_i\}_i$] We have $\transl{M} =
    \Case~(\transl{M'}\,{\cdot})~\{\ell_i~x_i \mapsto \transl{N_i}\}_i$.
    Proceed by case analysis where the first step of reduction happens.
    \begin{itemize}[topsep=0em]
      \item Reduction happens in $\transl{M'}$ or one of $\transl{N_i}$.
      Similar to the $\lambda x^A . M'$ case.
      \item The row type application $\transl{M'}\,{\cdot}$ is reduced by
      \tlab{RowLam}. Supposing $\transl{M}\treducestocl N'$, by the
      definition of translation, because $\transl{N}$ must be in the codomain of
      the translation, we can only have $N' \breducesto \transl{N}$ by applying
      \betalab{Case}, which implies $M' = (\ell_j\, M_j)^{\Variant{R}}$. By
      \refb{}, we have $\transl{M} \treducestocl \breducesto \transl{N_j[M_j /
      x_j]}$, which then gives us $N = N_j[M_j/x_j]$. Our goal follows from $M
      \breducesto N$.
    \end{itemize}
    \item[$M'^{\Variant{R}} \Upcast \Variant{R'}$] We have $\transl{M} =
    \Lambda \rho^{\Row_{R'}}.\transl{M'}\, (\transl{R' \backslash R}; \rho)$.
    Proceed by case analysis where the first step of reduction happens.
    \begin{itemize}[topsep=0em]
      \item Reduction happens in $\transl{M'}$. Similar to the $\lambda x^A .
      M'$ case.
      \item The row type application $\transl{M'}\, (\transl{R' \backslash R};
      \rho)$ is reduced by \tbetalab{RowLam}. Because $\transl{M'}$ should be a
      type abstraction, there are only two cases. Proceed by case analysis on
      $M'$.
      \begin{itemize}
        \item $M' = (\ell\,M_1)^{\Variant{R}}$. By \refc{}, we have $\transl{M}
        \breducestocl \transl{(\ell\,M_1)^{\Variant{R'}}}$, which then gives us
        $N = {(\ell\,M_1)^{\Variant{R'}}}$. Our goal follows from $M\breducestocl
        N$.
        \item $M' = M_1^{\Variant{R_1}} \Upcast \Variant{R}$. We have
        \[\ba{ll}
        \transl{M} &=
        \Lambda \rho^{\Row_{R'}}.\transl{M_1^{\Variant{R_1}} \Upcast \Variant{R}}\, (\transl{R' \backslash R}; \rho) \\ &=
        \Lambda \rho^{\Row_{R'}}. (\Lambda \rho^{\Row_{R}}.\transl{M_1}\hk{(\transl{R \backslash R_1}; \rho)}) \hk{(\transl{R' \backslash R}; \rho)} \\ &\vreducestocl
        \Lambda \rho^{\Row_{R'}}. \transl{M_1}\hk{(\transl{R \backslash R_1}; \transl{R' \backslash R}; \rho)} \\ &=
        \Lambda \rho^{\Row_{R'}}. \transl{M_1}\hk{(\transl{R' \backslash R_1}; \rho)} \\ &=
        \transl{M_1^{\Variant{R_1}}\Upcast \Variant{R'}}
        \ea\]
        By the definition of translation, we know that $N = {M_1^{\Variant{R_1}}\Upcast \Variant{R'}}$.
        Our goal follows from $M\unewreducestocl N$.
      \end{itemize}
    \end{itemize}
  \end{description}
\end{proof}

\subsection{Proof of the Encoding \STLCRecSub in \STLCRec}
\label{app:proof-stlcrecsub-stlcrec}

\begin{lemma}[Translation commutes with substitution]
  \label{lemma:recsubrec-subst}
  If $\typ{\Delta;\Gamma,x:A}{M:B}$ and $\typ{\Delta;\Gamma}{N:A}$,
  then $\transl{M[N/x]} = \transl{M}[\transl{N}/x]$.
\end{lemma}
\begin{proof}
  By straightforward induction on $M$.
  \begin{description}
    \item[$x$] $\transl{x[N/x]} = \transl{N} = \transl{x}[\transl{N}/x]$.
    \item[$y(y\neq x)$] $\transl{y[N/x]} = y = \transl{y}[\transl{N}/x]$
    \item[$M_1\, M_2$] Our goal follows from IH and definition of substitution.
    \item[$\Record{\ell_i=M_i}_i$] Our goal follows from IH and definition of substitution.
    \item[$M'.\ell$] Our goal follows from IH and definition of substitution.
    \item[$M'\Upcast A$] By IH and definition of substitution, we have
    \[\ba{ll} &\transl{(M' \Upcast \Record{\ell_i:A_i}_i) [N/x]} \\
    =& \transl{M'[N/x] \Upcast \Record{\ell_i:A_i}_i} \\
    =& \Record{\ell_i = \transl{M'[N/x]}.\ell_i}_i \\
    =& \Record{\ell_i = \transl{M'}[\transl{N}/x].\ell_i}_i \\
    =& (\Record{\ell_i = \transl{M'}.\ell_i}_i)[\transl{N}/x] \\
    =& \transl{M'\Upcast\Record{\ell_i:A_i}_i}[\transl{N}/x]
    \ea\]
  \end{description}
\end{proof}

\STLCRecSubSTLCRecTYPE*
\begin{proof}
  By straightforward induction on typing derivations.
  \begin{description}
    \item[\tylab{Var}] Our goal follows from $\transl{x} = x$ and \tylab{Var}.
    \item[\tylab{Lam}] Our goal follows from IH and \tylab{Lam}.
    \item[\tylab{App}] Our goal follows from IH and \tylab{App}.
    \item[\tylab{Record}] Our goal follows from IH and \tylab{Record}.
    \item[\tylab{Project}] Our goal follows from IH and \tylab{Project}.
    \item[\tylab{Upcast}] The only subtyping relation in \STLCRecSub is for
    record types. Given $\typ{\Delta;\Gamma}{M \Upcast \Record{R'} :
    \Record{R'}}$ and $\typ{\Delta;\Gamma}{M:\Record{R}}$, by IH we have
    $\typ{\transl{\Delta};\transl{\Gamma}}{\transl{M}:\Record{R}}$. Then,
    supposing $M=\Record{\ell_i = M_{\ell_i}}_i$ and $R' = (\ell_j':A_j)_j$, by
    definition of translation, $\Record{R}\subtype\Record{R'}$ and
    \tylab{Record} we have
    $\typ{\transl{\Delta};\transl{\Gamma}}{\Record{\ell_j'=M_{\ell_j'}}_j :
    \Record{R'}}$.
  \end{description}
\end{proof}

\STLCRecSubSTLCRec*

\begin{proof}~\\
  \noindent\textsc{Simulation}:

  First, we prove the base case that the whole term $M$ is reduced, i.e. $M
  \bureducesto N$ implies $\transl{M} \breducesto^\ast \transl{N}$. The proof
  proceeds by case analysis on the reduction relation.
  \begin{description}
    \item[\betalab{Lam}] We have $(\lambda x^A . M_1)\, M_2 \breducesto M_1[M_2
    / x]$. Then, \refa{$\transl{(\lambda x^A . M_1)\, M_2} = (\lambda x^A .
    \transl{M_1})\ \transl{M_2} \breducesto \transl{M_1}[\transl{M_2} / x] =
    \transl{M_1[M_2 / x]}$}, where the last equation follows from
    \Cref{lemma:recsubrec-subst}.
    \item[\betalab{Project}] We have $\Record{(\ell_i=M_i)_i}.\ell_j \breducesto
    M_j$. Our goal follows from \refb{$\transl{\Record{(\ell_i=M_i)_i}.\ell_j} =
    \Record{(\ell_i=\transl{M_i})_i}.\ell_j \breducesto \transl{M_j}$}.
    \item[\upcastlab{Upcast}] We have $\Record{\ell_i=M_{\ell_i}}_i \Upcast
    \Record{\ell_j':A_j}_j \ureducesto \Record{\ell'_j=M_{\ell'_j}}_j$. Our goal
    follows from $\transl{\Record{\ell_i=M_{\ell_i}}_i \Upcast
    \Record{\ell_j':A_j}_j} = \Record{\ell_j' =
    {\transl{\Record{\ell_i=M_{\ell_i}}_i}}.{\ell_j'} }_j = \Record{\ell_j'
    = {\Record{\ell_i=\transl{M_{\ell_i}}}_i}.{\ell_j'} }_j \breducestocl^\ast
    \Record{\ell_j' = \transl{M_{\ell_j'}}}_j$.
  \end{description}

  Then, we prove the full theorem by induction on $M$. We only need to prove the
  case where reduction happens in sub-terms of $M$.
  \begin{description}
    \item[$x$] No reduction.
    \item[$\lambda x^A . M'$] The reduction can only happen in $M'$. Supposing
    $\lambda x^A . M' \bureducestocl \lambda x^A . N'$, by IH on $M'$, we have
    $\transl{M'} \breducestocl^\ast \transl{N'}$, which then gives
    $\transl{\lambda x^A.M'} = \lambda x^A.\transl{M'}
    \breducestocl^\ast \lambda x^A.\transl{N'} = \transl{\lambda x^A
    . N'}$.
    \item[$M_1\ M_2$] Similar to the $\lambda x^A . M'$ case as reduction can only
    happen either in $M_1$ or $M_2$.
    \item[$\Record{\ell_i=M_i}_i$] Similar to the $\lambda x^A
    . M'$ case as reduction can only happen in one of $(M_i)_i$.
    \item[$M'.\ell$] Similar to the $\lambda x^A . M'$ case as reduction can only
    happen in $M'$.
    \item[$M' \Upcast A$] Similar to the $\lambda x^A . M'$ case as reduction can only
    happen in $M'$.
  \end{description}

  \noindent\textsc{Reflection}:
  We proceed by induction on $M$.
  \begin{description}
    \item[$x$] No reduction.
    \item[$\lambda x^A . M'$] We have $\transl{M} = \lambda
    x^{\transl{A}}.\transl{M'}$. The reduction can only happen in $\transl{M'}$.
    Suppose $\transl{M} \breducestocl \lambda x^{\transl{A}}.N_1$. By IH on
    $\transl{M'}$, there exists $N'$ such that $N_1 \breducestocl^\ast
    \transl{N'}$ and $M' \bureducestocl N'$. Our goal follows from setting
    $N$ to $\lambda x^A . N'$.
    \item[$M_1\ M_2$] We have $\transl{M} = \transl{M_1}\ \transl{M_2}$.
    Proceed by case analysis where the reduction happens.
    \begin{itemize}[topsep=0em]
      \item Reduction happens in either $\transl{M_1}$ or $\transl{M_2}$.
      Similar to the $\lambda x^A . M'$ case.
      \item The application is reduced by \betalab{Lam}. By
      definition of translation, we have $M_1 = \lambda x^A . M'$. By \refa{},
      we have $\transl{M} \breducesto \transl{M'[M_2 / x]}$. Our goal follows from
      setting setting $N$ to $M'[M_2 / x]$.
    \end{itemize}
    \item[$\Record{\ell_i=M_i}_i$] We have $\transl{M} =
    \Record{\ell_i=\transl{M_i}}_i$. Similar to the $\lambda x^A . M'$ case as
    the reduction can only happen in one of $\transl{M_i}$.
    \item[$M'.\ell_j$] We have $\transl{M} = \transl{M'}.\ell_j$.
    Proceed by case analysis where the reduction happens.
    \begin{itemize}[topsep=0em]
      \item Reduction happens in $\transl{M'}$. Similar to the $\lambda x^A .
      M'$ case.
      \item The projection is reduced by \betalab{Project}. By definition of
      translation, we have $M' = \Record{\ell_i=M_i}_i$. By \refb{}, we have
      $\transl{M} \breducesto \transl{M_j}$. Our goal follows from
      setting setting $N$ to $M_j$.
    \end{itemize}
    \item[$M'\Upcast\Record{\ell_i:A_i}_i$] We have $\transl{M' \Upcast
    \Record{\ell_i:A_i}_i} = \Record{\ell_i = \transl{M'}.\ell_i}_i$.
    Proceed by case analysis where the reduction happens.
    \begin{itemize}[topsep=0em]
      \item Reduction happens in one of $\transl{M'}$ in the result record.
      Supposing $\transl{M} \breducestocl M_1$, and in $M_1$ one of
      $\transl{M'}$ is reduced to $N_1$. By IH on $\transl{M'}$, there exists
      $N'$ such that $N_1\breducestocl^\ast \transl{N'}$ and $M'\bureducestocl
      N'$. Thus, we can apply the reduction $\transl{M'}\breducestocl N_1
      \breducestocl^\ast \transl{N'}$ to all $\transl{M'}$ in the result record,
      which gives us $\transl{M} \breducestocl M_1 \breducestocl^\ast
      \transl{N'\Upcast\Record{\ell_i:A_i}_i}$. Our goal follows from setting
      $N$ to $N'\Upcast\Record{\ell_i:A_i}_i$ and
      ${M'\Upcast\Record{\ell_i:A_i}_i} \bureducestocl
      {N'\Upcast\Record{\ell_i:A_i}_i}$.
      \item One of $\transl{M'}.\ell_i$ is reduced by \betalab{Project}. By the
      definition of translation, we know that $M' =
      \Record{\ell_j'=M_{\ell_j'}}_j$. Supposing $\transl{M} \breducestocl M_1$,
      we can reduce all projection in $\transl{M}$, which gives us $M_1
      \breducestocl^\ast \Record{\ell_i = \transl{M_{\ell_i}}}_i =
      \transl{\Record{\ell_i = M_{\ell_i}}_i}$. Our goal follows from setting
      $N$ to ${\Record{\ell_i = M_{\ell_i}}_i}$ and
      $M'\Upcast\Record{\ell_i:A_i}_i \ureducesto N$.
    \end{itemize}

  \end{description}
\end{proof}

\subsection{Proof of the Encoding \STLCRecSub in \STLCRecPre}
\label{app:proof-stlcrecsub-stlcrecpre}

\begin{lemma}[Translation commutes with substitution]
  \label{lemma:recsubrecpre-subst}
  If $\typ{\Delta;\Gamma,x:A}{M:B}$ and $\typ{\Delta;\Gamma}{N:A}$,
  then $\transl{M[N/x]} = \transl{M}[\transl{N}/x]$.
\end{lemma}
\begin{proof}
  By straightforward induction on $M$. We only need to consider cases that are
  different from the proof of \Cref{lemma:recsubrec-subst}.
  \begin{description}
    \item[$\Record{\ell_i=M_i}_i$] By IH and definition of substitution, we have
    \[\ba{ll}
     & \transl{\Record{\ell_i=M_i}_i^{\Record{\ell_i:A_i}_i} [N/x]} \\
    =& \transl{\Record{\ell_i=M_i[N/x]}_i^{\Record{\ell_i:A_i}_i}} \\
    =& (\Lambda \theta_i)_i.\Record{\ell_i = \transl{M_i[N/x]}}_i^{
      \Record{\ell_i^{\theta_i} : \transl{A_i}}_i} \\
    =& (\Lambda \theta_i)_i.\Record{\ell_i = \transl{M_i}[\transl{N}/x]}_i^{
      \Record{\ell_i^{\theta_i} : \transl{A_i}}_i} \\
    =& ((\Lambda \theta_i)_i.\Record{\ell_i = \transl{M_i}}_i^{
      \Record{\ell_i^{\theta_i} : \transl{A_i}}_i}) [\transl{N}/x] \\
    =& \transl{\Record{\ell_i=M_i}_i^{\Record{\ell_i:A_i}_i}} [\transl{N}/x]
    \ea\]
    \item[$M'.\ell$] By an equational reasoning similar to the case of
    $\Record{\ell_i=M_i}_i$.
    \item[$M'\Upcast A$] By an equational reasoning similar to the case of
    $\Record{\ell_i=M_i}_i$.
  \end{description}
\end{proof}

\STLCRecSubSTLCRecPreTYPE*

\begin{proof}~
  By induction on typing derivations.
  \begin{description}
    \item[\tylab{Var}] Our goal follows from $\transl{x} = x$.
    \item[\tylab{Lam}] Our goal follows from IH and \tylab{Lam}.
    \item[\tylab{App}] Our goal follows from IH and \tylab{App}.
    \item[\tylab{Record}] Our goal follows from IH, \tylab{Record} and
    \tylab{PreLam}.
    \item[\tylab{Project}] Supposing $M = M'.\ell_j$ and
    $\typ{\Delta;\Gamma}{M':\Record{\ell_i:A_i}_i}$, by definition of
    translation we have $\transl{M'.\ell_j} = (\transl{M'}(P_i)_i).\ell_j$ where
    $P_j = \Present$. IH on $M'$ implies
    $\typ{\transl\Delta;\transl\Gamma}{\transl{M'} : (\forall
    \theta_i)_i.\Record{\ell_i^{\theta_i} : \transl{A_i}}_i}$. Our goal follows
    from \tylab{PreApp} and \tylab{Project}.
    \item[\tylab{Upcast}] The only subtyping relation in \STLCRecSub is for
    record types. Given $\typ{\Delta;\Gamma}{M^{\Record{R}} \Upcast
    \Variant{R'} : \Variant{R'}}$, by $\typ{\Delta;\Gamma}{M:\Record{R}}$ and
    IH we have
    $\typ{\transl{\Delta};\transl{\Gamma}}{\transl{M}:\transl{\Record{R}}}$.
    Then, by definition of translation and \tylab{RowApp} we have
    $\typ{\transl{\Delta};\transl{\Gamma}}{\transl{M^{\Record{R}} \Upcast
    \Record{R'}}:\transl{\Record{R'}}}$.
  \end{description}
\end{proof}

\STLCRecSubSTLCRecPre*

\begin{proof}~\\
  \noindent\textsc{Simulation}: First, we prove the base case that the whole
  term $M$ is reduced, i.e. $M \breducesto N$ implies
  $\transl{M} \treducestocl^\ast\breducesto \transl{N}$, and $M \ureducesto N$
  implies $\transl{M} \vreducestocl^\ast \transl{N}$. The proof proceeds by case
  analysis on the reduction relation:
  \begin{description}
    \item[\betalab{Lam}] We have $(\lambda x^A . M_1)\, M_2 \breducesto M_1[M_2
    / x]$. Then, \refa{$\transl{(\lambda x^A . M_1)\, M_2} = (\lambda x^A .
    \transl{M_1})\ \transl{M_2} \breducesto \transl{M_1}[\transl{M_2} / x] =
    \transl{M_1[M_2 / x]}$}, where the last equation follows from
    \Cref{lemma:recsubrecpre-subst}.
    \item[\betalab{Project}]
    We have $\Record{(\ell_i=M_i)_i}.\ell_j \breducesto M_j$. By definition of
    translation, we have $\transl{\Record{(\ell_i=M_i)_i}.\ell_j} =
    (\transl{\Record{\ell_i=M_i}_i} (P_i)_i).\ell_j = (((\Lambda \theta_i)_i .
    \Record{\ell_i^{\theta_i} = \transl{M_i}}_i) (P_i)_i).\ell_j $, where
    $P_j=\Present$ and $P_i=\Absent (i\neq j)$. Applying \betalab{PreLam}, we
    have \refb{$\transl{\Record{(\ell_i=M_i)_i}.\ell_j} \treducestocl^\ast
    (\Record{\ell_i^{P_i} = \transl{M_i}}_i).\ell_j \breducesto \transl{M_j}$}.
    \item[\upcastlab{Upcast}] We have $\Record{(\ell_i=M_{\ell_i})_i}^{\Record{R}}
    \Upcast \Record{R'} \ureducesto \Record{\ell_j' = M_{\ell_j'}}_j$, where $R
    = (\ell_i : A_{\ell_i})_i$ and $R' = (\ell_j' : A_{\ell_j'})_j$. By
    definition, \refc{$\transl{\Record{(\ell_i=M_{\ell_i})_i}^{\Record{R}} \Upcast \Record{R'}} =
    (\Lambda \theta_j')_j.\transl{\Record{(\ell_i=M_{\ell_i})_i}^{\Record{R}}}\,(\hk{P_i})_i =
    (\Lambda \theta_j')_j. ((\Lambda \theta_i)_i.\Record{\ell_i = \transl{M_{\ell_i}}}_i^{
          \Record{\ell_i^{\theta_i} : A_{\ell_i}}_i}) \,(\hk{P_i})_i \vreducestocl^\ast
    (\Lambda \theta_j')_j. \Record{\ell_i = \transl{M_{\ell_i}}}_i^{
          \Record{\ell_i^{P_i} : A_{\ell_i}}_i}$},
    where $P_i = \Absent$ when $\ell_i \notin (\ell_j')_j$,
    and $P_i = \theta_j'$ when $\ell_i = \ell_j'$.
    By the fact that we ignore absent labels when comparing records in \STLCRecPre,
    we have \refd{$(\Lambda \theta'_j)_j.
    \Record{\ell_i = \transl{M_{\ell_i}}}_i^{\Record{\ell_i^{P_i} : A_{\ell_i}}_i}
    = (\Lambda \theta_j')_j.\Record{\ell_j' = \transl{M_{\ell_j'}}}^{
          \Record{\ell_j'^{\theta_j'} : A_{\ell_j'}}_j}
    = \transl{\Record{\ell_j' = M_{\ell_j'}}_j}$}.
  \end{description}

  Then, we prove the full theorem by induction on $M$. We only need to prove the
  case where reduction happens in sub-terms of $M$.
  \begin{description}
    \item[$x$] No reduction.
    \item[$\lambda x^A . M'$] The reduction can only happen in $M'$. Supposing
    $\lambda x^A . M' \breducestocl \lambda x^A . N'$, by IH on $M'$, we have
    $\transl{M'} \treducestocl^\ast\breducestocl \transl{N'}$, which then gives
    $\transl{\lambda x^A.M'} = \lambda x^A.\transl{M'}
    \treducesto^\ast\breducestocl \lambda x^A.\transl{N'} = \transl{\lambda x^A
    . N'}$.
    The same applies to the second part of the theorem.
    \item[$M_1\ M_2$] Similar to the $\lambda x^A . M'$ case as reduction can only
    happen either in $M_1$ or $M_2$.
    \item[$\Record{\ell_i=M_i}_i$] Similar to the $\lambda x^A
    . M'$ case as reduction can only happen in one of $(M_i)_i$.
    \item[$M'.\ell$] Similar to the $\lambda x^A . M'$ case as reduction can only
    happen in $M'$.
    \item[$M' \Upcast A$] Similar to the $\lambda x^A . M'$ case as reduction can only
    happen in $M'$.
  \end{description}

  \noindent\textsc{Reflection}: We proceed by induction on $M$.
  \begin{description}
    \item[$x$] No reduction.
    \item[$\lambda x^A . M'$] We have $\transl{M} = \lambda
    x^{\transl{A}}.\transl{M'}$. The reduction can only happen in $\transl{M'}$.
    Suppose $\transl{M} \treducestocl^\ast\breducestocl \lambda x^{\transl{A}}.\transl{N'}$.
    By IH on $\transl{M'}$, $M' \treducestocl^\ast\breducestocl N'$. Our goal follows from
    $\lambda x^A.M' \treducestocl^\ast\breducestocl \lambda x^A.N'$.
    Suppose $\transl{M} \vreducestocl \lambda x^{\transl{A}}.N_1$. By IH on
    $\transl{M'}$, there exists $N'$ such that $N_1 \vreducestocl^\ast
    \transl{N'}$ and $M' \uunewreducestocl N'$. Our goal follows from setting
    $N$ to $\lambda x^A . N'$.

    \item[$M_1\ M_2$] We have $\transl{M} = \transl{M_1}\ \transl{M_2}$.
    Proceed by case analysis where the reduction happens.
    \begin{itemize}[topsep=0em]
      \item Reduction happens in either $\transl{M_1}$ or $\transl{M_2}$.
      Similar to the $\lambda x^A . M'$ case.
      \item The application is reduced by \betalab{Lam}. By
      definition of translation, we have $M_1 = \lambda x^A . M'$. By \refa{},
      we have $\transl{M} \breducesto \transl{M'[M_2 / x]}$. Our goal follows from
      setting setting $N$ to $M'[M_2 / x]$.
    \end{itemize}

    \item[$\Record{\ell_i=M_i}_i$] We have $\transl{M} = (\Lambda
    \theta_i)_i.\Record{\ell_i = \transl{M_i}}_i^{\Record{\ell_i^{\theta_i} :
    \transl{A_i}}_i}$. Similar to the $\lambda x^A . M'$ case as the reduction
    can only happen in one of $\transl{M_i}$.

    \item[$M'.\ell_j$] We have $\transl{M} = (\transl{M'}\,({P_i})_i).\ell_j$,
    where $P_i = \Absent$ for $i\neq j$ and $P_j = \Present$.
    Proceed by case analysis where the $\beta$-reduction happens.
    \begin{itemize}[topsep=0em]
      \item Reduction happens in $\transl{M'}$. Similar to the $\lambda x^A .
      M'$ case.
      \item The projection is reduced by \betalab{Project^\star}. Supposing $\transl{M}\treducestocl^\ast\breducestocl \transl{N}$,
      because $\transl{N}$ is in the codomain of the translation,
      the $\treducestocl^\ast$ can only be the type applications of $(P_i)_i$
      and $M' = \Record{\ell_i=M_i}_i$.
      By \refb{}, we have $\transl{M'.\ell_j} \treducestocl^\ast
      \breducesto \transl{M_j}$. Our goal follows from
      $M'.\ell_j \breducesto M_j$.
    \end{itemize}

    \item[${M'^{\Record{\ell_i : A_i}_i} \Upcast \Record{\ell'_j : A'_j}_j}$] We have
    $\transl{M} = (\Lambda \theta_j)_j.\transl{M'}\,(\hk{P_i})_i$, where $P_i = \Absent$
    for $\ell_i \notin (\ell_j')_j$, and $P_i = \theta_j$ for $\ell_i = \ell_j'$.
    Proceed by case analysis where the reduction happens.
    \begin{itemize}[topsep=0em]
      \item Reduction happens in $\transl{M'}$. Similar to the $\lambda x^A . M'$ case.
      \item The presence type application $\transl{M'}\,\hk{P_1}$ is reduced by
      \vlab{PreLam}. Because the top-level constructor of $\transl{M'}$ should be
      type abstraction, there are two cases. Proceed by case analysis on $M'$.
      \begin{itemize}[topsep=0em]
        \item $M' = \Record{\ell_i=M_{\ell_i}}_i$. We can reduce all presence
        type application of $P_i$. By \refc{} and \refd{}, we have $\transl{M}
        \vreducestocl^\ast \transl{\Record{\ell_j' = M_{\ell_j'}}_j}$. Our goal
        follows from setting $N$ to ${\Record{\ell_j' = M_{\ell_j'}}_j}$ and
        $M \ureducesto N$.
        \item $M' = M_1^{\Record{\ell_k'' : B_k}_k} \Upcast \Record{\ell_i : A_i}_i$.
        We can reduce all presence type application of $P_i$. We have $\transl{M} =
        (\Lambda \theta_j)_j.\transl{M_1 \Upcast \Record{\ell_i : A_i}_i}\,(\hk{P_i})_i =$\\
        $(\Lambda \theta_j)_j.((\Lambda \theta_i)_i.\transl{M_1}\,(\hk{P_k'})_k)\,(\hk{P_i})_i
        \vreducestocl^\ast
        (\Lambda \theta_j)_j.\transl{M_1}\,(\hk{Q_k})_k$, where $P_k' = \Absent$ for $\ell_k''
        \notin (\ell_i)_i$, and $P_k' = \theta_i$ for $\ell_k'' = \ell_i$.
        Thus, we have $Q_k = \Absent$ for $\ell_k''
        \notin (\ell_j')_j$, and $Q_k = \theta_j$ for $\ell_k'' = \ell_j'$, which implies
        $\transl{M_1 \Upcast \Record{\ell_j':A_j'}_j} =
        (\Lambda \theta_j)_j.\transl{M_1}\,(\hk{Q_k'})_k$. Our goal follows from setting
        $N$ to ${M_1 \Upcast \Record{\ell_j':A_j'}_j}$ and $M \unewreducesto N$.
      \end{itemize}
    \end{itemize}
  \end{description}
\end{proof}

\section{Encodings, Proofs and Definitions in Section~\ref{SEC:FULL}}

In this section, we provide the encodings, proofs, and definitions
missing from \Cref{sec:full-subtyping}.

\subsection{Local Term-Involved Encoding of \STLCVarRecSubFull in \STLCVarRec}
\label{app:coercion-functions}

The local term-involved encoding of \STLCVarRecSubFull in \STLCVarRec
\citep{tapl,BREAZUTANNEN1991172} is formalised as follows.

\begin{equations}
  \transl{-} &:& \Deriv \to \Terms\\
  \transl{M^A \Upcast B} &=& \transl{A \subtype B}\ \transl{M} \\[2ex]
  \transl{-} &:& \mathrm{Subtyping} \to \Terms\\
  \transl{\alpha \subtype \alpha}
  &=& \lambda x^\alpha . x\\
  \transl{A\to B \subtype A'\to B'}
  &=& \lambda f^{A\to B}. \lambda x^{A'}. \transl{B\subtype B'}\ (f\ (\transl{A'\subtype A}\ x))\\
  \Bigl\llbracket
    \inferrule
      {\dom{R}\subseteq \dom{R'} \quad
       [A_i \subtype A_i']_{(\ell_i:A_i) \in R, (\ell_i:A_i') \in R'}}
      {\Variant{R} \subtype \Variant{R'}}
  \Bigr\rrbracket
  &=& \lambda x^{\Variant{R}}. \Case~x~\{\ell_i\ y\mapsto
    (\ell_i\ (\transl{A_i\le A_i'}\ y))^{\Variant{R'}}\}\\
  \Bigl\llbracket
    \inferrule
      {\dom{R'}\subseteq \dom{R} \quad
       [A_i \subtype A_i']_{(\ell_i:A_i) \in R, (\ell_i:A_i') \in R'}}
      {\Record{R} \subtype \Record{R'}}
  \Bigr\rrbracket
  &=& \lambda x^{\Record{R}} . \Record{\ell_i = \transl{A_i \le A_i'}\ x.\ell_i}
\end{equations}

\subsection{Dynamic Semantics of \STLCVarRecSubFull}
\label{app:erasure-semantics}

In addition to the erasure semantics, the other style of dynamic
semantics of \STLCVarRecSubFull is given by extending the operational
semantics rules with the following four upcast rules.
\begin{reductions}
  \upcastlab{Var} & M \Upcast \alpha &\ureducesto& M \\
  \upcastlab{Lam} & (\lambda x^A . M) \Upcast A' \to B' &\ureducesto&
    \lambda y^{A'} . (M[(y \Upcast A) / x] \Upcast B') \\
  \upcastlab{Variant} & (\ell_j\,M)^A \Upcast \Variant{\ell_i : A_i}_i &\ureducesto&
    (\ell_j\,(M\Upcast A_j))^{\Variant{\ell_i : A_i}_i} \\
  \upcastlab{Record} & \Record{\ell_i=M_{\ell_i}}_i \Upcast \Record{\ell_j':A_j}_j &\ureducesto&
    \Record{\ell'_j = M_{\ell'_j} \Upcast A_j}_j
\end{reductions}

We show that there is a correspondence between these two styles of
dynamic semantics of \STLCVarRecSubFull.
We first give a preorder $M \equivsub N$ on terms of the untyped
\STLCVarRec which allows records in $M$ to contain more elements than
those in $N$, because the erasure semantics does not truly perform
upcasts. The full definition is shown in \Cref{fig:equivsub}.

\begin{figure}[h]
  \begin{mathpar}
    \inferrule
    { \{\ell_j'\}_j \subseteq \{\ell_i\}_i \\
      [M_{i} \equivsub N_{j}]_{\ell_i = \ell_j'}}
    {\Record{\ell_i = M_{i}}_i \equivsub \Record{\ell_j' = N_{j}}_j}

    \inferrule
    {}
    {x \equivsub x}

    \inferrule
    {M \equivsub M'}
    {\lambda x.M \equivsub \lambda x.M'}

    \inferrule
    {M \equivsub M'\\ N \equivsub N'}
    {M\ N \equivsub M'\ N'}

    \inferrule
    {M \equivsub M'}
    {\ell\,M \equivsub \ell\,M'}

    \inferrule
    {M\equivsub M'\\ [N_i\equivsub N_i']_i}
    {\Case\,M\,\{\ell_i\,x_i\mapsto N_i\}_i \equivsub \Case\,M'\,\{\ell_i\,x_i\mapsto N_i'\}_i}

    \inferrule
    {M \equivsub M'}
    {M.\ell \equivsub M'.\ell}
  \end{mathpar}
  \caption{The preorder $\equivsub$ of untyped \STLCVarRec.}
  \label{fig:equivsub}
  \end{figure}

The correspondence is given by the following theorem.
\begin{restatable}[Operational Correspondence]{theorem}{STLCVarRecSubFullOC}
  \label{thm:oc-erasure}
  Given a well-typed term $M$ in \STLCVarRecSubFull and a term $M'$ in untyped
  \STLCVarRec with $M' \equivsub \erase{M}$, we have:
  \begin{itemize}[align=left, topsep=0em]
    \item[\textsc{\emph{Simulation}}] If $M \breducestocl N$, then
    there exists $N'$ such that $N' \equivsub \erase{N}$ and
    $M' \breducestocl N'$;
    if $M \ureducestocl N$, then $M' \equivsub \erase{N}$.
    \item[\textsc{\emph{Reflection}}] If $M' \breducestocl N'$, then
    there exists $N$ such that $N' \equivsub \erase{N}$ and
    $M \ureducestocl^\ast\breducestocl N$.
  \end{itemize}
\end{restatable}

To prove it, we need two lemmas.

\begin{lemma}[Erasure commutes with substitution]
  \label{lemma:erasure-subst}
  If $\typ{\Delta;\Gamma,x:A}{M:B}$ and $\typ{\Delta;\Gamma}{N:A}$, then for
  $M'\equivsub \erase{M}$ and $N'\equivsub \erase{N}$, we have $M'[N'/x]
  \equivsub \erase{M[N/x]}$.
\end{lemma}
\begin{proof}
  By straightforward induction on $M$.
\end{proof}

\begin{lemma}[Upcasts shrink terms]
  \label{lemma:upcast-shrink}
  For any $M\Upcast A \ureducestocl N$ in \STLCVarRecSubFull, we have $\erase{M}
  \equivsub \erase{N}$.
\end{lemma}
\begin{proof}
  By definition of $\erase{-}$ and $\ureducestocl$.
\end{proof}

Then, we give the proof of \Cref{thm:oc-erasure}.

\begin{proof}~\\
  \noindent\textsc{Simulation}:
  We proceed by induction on $M$.
  \begin{description}
    \item[$x$] No reduction.
    \item[$\lambda x^A . M_1$]
    Supposing $M' = \lambda x . M_1'$, by $M'\equivsub\erase{M}$ we have $M_1'
    \equivsub \erase{M_1}$. The reduction must happen in $M_1$. Our goal follows
    from the IH on $M_1$.
    \item[$M_1\ M_2$]
    Supposing $M' = M_1'\ M_2'$, by $M'\equivsub\erase{M}$ we have $M_1'
    \equivsub \erase{M_1}$ and $M_2' \equivsub \erase{M_2}$. We proceed by case
    analysis where the reduction happens.
    \begin{itemize}[topsep=0em]
      \item The reduction happens in either $M_1$ or $M_2$. Our goal follows
      from the IH.
      \item The reduction reduces the top-level function application. Supposing
      $M_1 = \lambda x^A . M_3$ and $M_1'=\lambda x.M_3'$ with
      $M_3'\equivsub\erase{M_3}$, we have $(\lambda x^A . M_3)\ M_2
      \breducestocl M_3[M_2/x]$ and $(\lambda x^A . M_3')\ M_2' \breducestocl
      M_3'[M_2'/x]$. Our goal follows from \Cref{lemma:erasure-subst}.
    \end{itemize}
    \item[$N.\ell_k$]
    Supposing $M' = N'.\ell_k$, by $M'\equivsub\erase{M}$ we have $N'\equivsub \erase{N'}$.
    We proceed by case analysis where the reduction happens.
    \begin{itemize}[topsep=0em]
      \item The reduction happens in $N$. Our goal follows from the IH on $N$.
      \item The reduction reduces the top-level projection. Supposing $N =
      \Record{\ell_i=M_i}_i$ and $N' = \Record{\ell_j'=M_j'}_j$ with
      $\{\ell_j'\}_j \subseteq \{\ell_i\}_i$ and $(M_j'\equivsub \erase{M_i})_{\ell_i =
      \ell_j'}$, we have $N.\ell_k \breducestocl M_k$ and $N'.\ell_k \breducestocl M_n'$
      where $\ell_k = \ell_n'$.
      Our goal follows from $M_n' \equivsub \erase{M_k}$.
    \end{itemize}
    \item[$\Record{\ell_i=M_i}_i$] The reduction must happen in one of the
    $M_i$. Our goal follows from the IH.
    \item[$M_1\Upcast A$] For the $\beta$-reduction, it must happen in $M_1$. Our goal
    follows from the IH.
    For the upcast reduction, by $M' \equivsub \erase{M}$ we have $M' \equivsub \erase{M_1}$.
    By \Cref{lemma:upcast-shrink}, we have $M' \equivsub \erase{M_1} \equivsub \erase{N}$.
  \end{description}

  \noindent\textsc{Reflection}: We proceed by induction on $M'$.
  \begin{description}
    \item[$x$] No reduction.
    \item[$\lambda x.M_1'$]
    By $M' \equivsub \erase{M}$, we know that there exists $\lambda x^A . M_1$
    such that $M\ureducestocl^\ast \lambda x^A.M_1$. By
    \Cref{lemma:upcast-shrink}, $\erase{M} \equivsub \erase{\lambda x^A.M_1}$.
    Then, by $M' \equivsub \erase{M}$ and transitivity, we have $M_1'
    \equivsub\erase{M_1}$. The $\beta$-reduction must happen in $M_1'$. Our goal
    follows from the IH on $M_1'$.
    \item[$M_1'\ M_2'$]
    By $M' \equivsub \erase{M}$, we know that there exists $M_1\ M_2$ such that
    $M\ureducestocl^\ast M_1\ M_2$. By \Cref{lemma:upcast-shrink} and $M'
    \equivsub \erase{M}$, we have $M_1' \equivsub \erase{M_1}$ and $M_2'
    \equivsub \erase{M_2}$. We proceed by case analysis where the reduction
    happens.
    \begin{itemize}[topsep=0em]
      \item The reduction happens in either $M_1'$ or $M_2'$. Our goal follows
      from the IH.
      \item The reduction reduces the top-level function application. Supposing
      $M_1'=\lambda x.M_3'$, by $M_1' \equivsub \erase{M_1}$, we know that there
      exists $\lambda x^A . M_3$ such that $M_1 \ureducestocl^\ast \lambda x^A .
      M_3$. Thus, $M_1\ M_2 \ureducestocl^\ast\breducestocl M_3[M_2/x]$ and
      $M_1'\ M_2' \breducestocl M_3'[M_2'/x]$. By \Cref{lemma:upcast-shrink}, we
      have $M_1' \equivsub \erase{M_1} \equivsub \erase{\lambda x^A . M_3}$,
      which implies $M_3' \equivsub \erase{M_3}$. Our goal follows from
      \Cref{lemma:erasure-subst}.
    \end{itemize}
    \item[$N'.\ell_k$]
    By $M' \equivsub \erase{M}$, we know that there exists $N.\ell_k$ such that
    $M\ureducestocl^\ast N.\ell_k$. By \Cref{lemma:upcast-shrink} and
    $M'\equivsub \erase{M}$, we have $N' \equivsub \erase{N}$. We proceed by
    case analysis where the reduction happens.
    \begin{itemize}[topsep=0em]
      \item The reduction happens in $N'$. Our goal follows from the IH on $N$.
      \item The reduction reduces the top-level projection. Supposing $N' =
      \Record{\ell_j'=M_j'}_j$, by $N' \equivsub \erase{N}$, we know that there
      exists $\Record{\ell_i=M_i}_i$ such that $N\ureducestocl^\ast
      \Record{\ell_i=M_i}_i$. Thus, $N.\ell_k \ureducestocl^\ast\breducestocl
      M_k$ and $N'.\ell_k \breducestocl M_n'$ where $\ell_n' = \ell_k$. By
      \Cref{lemma:upcast-shrink}, we have $\erase{N} \equivsub
      \erase{\Record{\ell_i=M_i}_i}$. We can further conclude that $M_n'
      \equivsub \erase{M_k}$ from $N' \equivsub \erase{N}$.
    \end{itemize}
  \end{description}
\end{proof}

\subsection{Proof of the Encoding of \STLCRecSubCo in \STLCRecPre}
\label{app:proof-recsubco-recpre}

\begin{lemma}[Upcast Translation]
  \label{lemma:upcast-translation}
  If $A\subtype B$, then $\forall\ol\theta . \transl{A,\ol P} = \transl{B}$ for
  $(\ol\theta,\ol P) = \transpre{\theta}{A\subtype B}$.
\end{lemma}
\begin{proof}
  By a straightforward induction on the definition of $\transpre{\theta}{A\subtype B}$.
\end{proof}

\STLCVarRecSubCoSTLCVarRecPreTYPE*

\begin{proof}~
  By induction on typing derivations.
  \begin{description}
    \item[\tylab{Var}] Our goal follows from $\transl{x} = x$.
    \item[\tylab{Lam}] By the IH on $\typ{\Delta;\Gamma,x:A}{M:B}$, we have
    $$\typ{\Delta;\transl{\Gamma},x:\transl{A}}{\transl{M}:\transl{B}}$$
    Let $\ol\theta = \transpre{\theta}{B}$.
    By \tylab{PreApp} and context weakening, we have
    $$\typ{\Delta,\ol\theta;\transl{\Gamma},x:\transl{A}}{\transl{M}\,\ol\theta:\transl{B,\ol\theta}}$$
    Notice that we always assume variable names in the same context are unique,
    so we do not need to worry that $\ol\theta$ conflicts with $\Delta$.
    Then, by \tylab{Lam}, we have
    $$\typ{\Delta,\ol\theta;\transl{\Gamma}}{\lambda x^{\transl{A}}.\transl{M}\,\ol\theta:
      \transl{A} \to \transl{B,\ol\theta}}$$
    Finally, by \tylab{PreLam}, we have
    $$\typ{\Delta;\transl{\Gamma}}{\Lambda\ol\theta.\lambda x^{\transl{A}}.\transl{M}\,\ol\theta:
      \forall\ol\theta.\transl{A} \to \transl{B,\ol\theta}}$$
    Our goal follows from $\transl{A\to B} = \forall\ol\theta.\transl{A} \to \transl{B,\ol\theta}$.
    \item[\tylab{App}] Similar to the \tylab{Lam} case.
    Our goal follows from IH, \tylab{App}, \tylab{PreApp} and \tylab{PreLam}.
    \item[\tylab{Record}] Similar to the \tylab{Lam} case.
    Our goal follows from IH, \tylab{Record}, \tylab{PreApp} and \tylab{PreLam}.
    \item[\tylab{Project}] Given the derivation of
    $\typ{\Delta;\Gamma}{M.\ell_j}{A_j}$, by the IH on
    $\typ{\Delta;\Gamma}{M:\Record{\ell_i:A_i}_i}$, we have
    $$\typ{\Delta;\transl{\Gamma}}{\transl{M}:\transl{\Record{\ell_i:A_i}_i}}$$
    Let $ P_i = \Absent (i \neq j), P_j = \Present, \ol{\theta} =
    \transpre{\theta}{A_j}, \ol P_i = \transpre{\Absent}{A_i}$.
    By \tylab{PreApp} and context weakening, we have
    $$
    \typ{\Delta,\ol\theta;\transl{\Gamma}}{
      \transl{M}\ (P_i)_i\ (\ol P_i)_{i<j}\ \ol\theta\ (\ol P_i)_{j<i}
      :\Record{R}}
    $$
    where $\ell_j:\transl{A_j,\ol\theta} \in R$ by the definition of translations
    and the canonical order.
    Then, by \tylab{Proj}, we have
    $$
    \typ{\Delta,\ol\theta;\transl{\Gamma}}{
      (\transl{M}\ (P_i)_i\ (\ol P_i)_{i<j}\ \ol\theta\ (\ol P_i)_{j<i}).\ell_j
      :\transl{A_j,\ol\theta}}
    $$
    Finally, by \tylab{PreLam}, we have
    $$
    \typ{\Delta;\transl{\Gamma}}{
      (\transl{M}\ (P_i)_i\ (\ol P_i)_{i<j}\ \ol\theta\ (\ol P_i)_{j<i}).\ell_j
      :\forall \ol\theta.\transl{A_j,\ol\theta}}
    $$
    Our goal follows from $\transl{A_j} = \forall\ol\theta.\transl{A_j,\ol\theta}$
    where $\ol\theta = \transpre{\theta}{A_j}$.
    \item[\tylab{Upcast}] Given the derivation of $\typ{\Delta;\Gamma}{M\Upcast B:B}$,
    by the IH on $\typ{\Delta;\Gamma}{M:A}$, we have
    $$\typ{\Delta;\transl{\Gamma}}{\transl{M} : \transl{A}}$$
    Let $(\ol\theta, \ol P) = \transpre{\theta}{A\subtype B}$.
    By \tylab{PreApp} and context weakening, we have
    $$\typ{\Delta, \ol\theta;\transl{\Gamma}}{\transl{M}\ \ol P : \transl{A, \ol P}}$$
    Then, by \tylab{PreLam}, we have
    $$\typ{\Delta;\transl{\Gamma}}{\Lambda \ol\theta.\transl{M}\ \ol P : \forall \ol\theta.\transl{A, \ol P}}$$
    By \Cref{lemma:upcast-translation}, we have $\transl{B} = \forall \ol\theta.\transl{A, \ol P}$.
  \end{description}
\end{proof}

\section{The Proof in Section~\ref{SEC:PRENEX}}

In this section, we spell out the proofs that are missing from
\Cref{sec:prenex-polymorphism}.

\subsection{Proof of encoding \STLCRecSubFullCond{2} using \STLCRecPrenex}
\label{app:proof-stlcrecsubfull-stlcrecprenex}

\STLCRecRowPrenexType*

\begin{proof}

  As shown in \Cref{sec:prenex-polymorphism}, we only need to prove
  that $\typ{\Delta;\Gamma}{M:A}$ in \STLCRecSubFullA implies
  $\typ{\transl{\Delta;\Gamma}}{\transl{M}:\tau}$ for some $\tau
  \subtyperow \transl{A}$ in \STLCRecPrenex. We proceed by induction
  on the typing derivations in \STLCRecSubFullA.

  \begin{description}[leftmargin=5em]
    \item[\tylab{Var}] Our goal follows directly from the definition of
    translations.
    \item[\tylab{Lam}] Given the derivation of $\typ{\Delta;\Gamma}{\lambda a^A.M
    : A \to B}$, by the IH on $\typ{\Delta;\Gamma,a:A}{M:B}$, we have
    $$
    \typ{ \Delta, \ftv{\transl{\Gamma}}, \transrowb{\rho_{|\Gamma|}}{A};
          \Gamma, a:\translb{A, \transrowb{\rho_{|\Gamma|}}{A}}}
    {\transl{M} : \tau_B}
    $$
    for some $\tau_B \subtyperow \transl{B}$. Supposing $\tau_B = \forall \ol\rho_B .
    B'$, by \tylab{Inst} and environment weakening, we have
    \footnote{We always assume type variables in type environments have
    different names, and we omit kinds when they are easy to reconstruct from
    the context.}
    $$
    \typ{ \Delta, \ftv{\transl{\Gamma}}, \transrowb{\rho_{|\Gamma|}}{A}, \ol\rho_B;
          \Gamma, a:\translb{A, \transrowb{\rho_{|\Gamma|}}{A}}}
    {\transl{M} : B'}
    $$
    Then, by \tylab{Lam}, we have
    $$
    \typ{ \Delta, \ftv{\transl{\Gamma}}, \transrowb{\rho_{|\Gamma|}}{A}, \ol\rho_B;
          \Gamma}
    {\lambda a.\transl{M} : \translb{A, \transrowb{\rho_{|\Gamma|}}{A}} \to B'}
    $$
    Finally, by \tylab{Gen}, we have
    $$
    \typ{ \Delta, \ftv{\transl{\Gamma}};
          \Gamma}
    {\lambda a.\transl{M} :
      \forall \transrowb{\rho_{|\Gamma|}}{A}\, \ol\rho_B .
      \translb{A, \transrowb{\rho_{|\Gamma|}}{A}} \to B'}
    $$
    By definition, we have $\transla{A\to B} = \forall \ol\rho_1 \ol\rho_2 .
    \translb{A,\ol\rho_1} \to \transla{B, \ol\rho_2} $, where $\ol\rho_1 =
    \transrowb{\rho_1}{A},\ \ol\rho_2 = \transrowa{\rho_2}{B}$. It is easy to
    check that $\forall \transrowb{\rho_{|\Gamma|}}{A}\, \ol\rho_B . \translb{A,
    \transrowb{\rho_{|\Gamma|}}{A}} \to B' \subtyperow \transl{A\to B}$ under
    $\alpha$-renaming.

    \item[\tylab{AppSub}] Given the derivation of $\typ{\Delta;\Gamma}{M\,N:B}$, by
    the IH on $\typ{\Delta;\Gamma}{M:A\to B}$, we have
    $$
    \typ{\transl{\Delta;\Gamma}}{\transl{M} : \tau_1}
    $$
    for some $\tau_1 \subtyperow \transla{A\to B}$.
    By the IH on $\typ{\Delta;\Gamma}{B:A_2}$, we have
    $$
    \typ{\transl{\Delta;\Gamma}}{\transl{N} : \tau_2}
    $$
    for some $\tau_2 \subtyperow \transla{A_2}$.
    We have $\nocontratwice{A\to B}$, which implies
    $\nocontraonce{A}$. Then, $A_2 \subtype A$ gives us $\nocontraonce{A_2}$,
    which further implies that $\transla{A_2} = A_2$ and $\tau_2$ is not
    polymorphic.
    Thus, we have $\tau_2 \subtyperow \transla{A_2} = A_2 \subtype A$.
    Notice that given $A \subtyperow \_ \subtype B$ with $\nocontraonce{B}$, we
    can always construct $\ol R$ with $\translb{B,\ol R} = A$, by
    $\transrowsub{A\subtyperow\subtype B}$ defined as follows.
    \begin{equations}
      \transrowsub{-} &:& (\Type\subtyperow\subtype\Type) \to (\ol\Rows) \\
      \transrowsub{\alpha\subtyperow\subtype\alpha}
      &=& (\cdot, \cdot) \\
      \transrowsub{A\to B\subtyperow\subtype A\to B'}
      &=& \transrowsub{B\subtyperow\subtype B'}\\
      \transrowsub{\Record{(\ell_i:A_i)_i}\subtyperow\subtype \Record{(\ell_j':A_j')}}
      &=& (\ell_k:A_k)_{k\in \{\ell_i\}_i \backslash \{\ell_j'\}_j}\
          \transrowsub{A_i\subtyperow\subtype A_j'}_{\ell_i = \ell_j'}\\
      \transrowsub{\Record{(\ell_i:A_i)_i;\rho}\subtyperow\subtype \Record{(\ell_j':A_j')}}
      &=& ((\ell_k:A_k)_{k\in \{\ell_i\}_i \backslash \{\ell_j'\}_j};\rho)\
          \transrowsub{A_i\subtyperow\subtype A_j'}_{\ell_i = \ell_j'}\\
    \end{equations}
    Let $\ol R = \transrowsub{\tau_2 \subtyperow\subtype A}$. We have
    $\translb{A,\ol R} = \tau_2$.
    Suppose $\tau_1 = \forall\ol\rho.A'\to B'$. By definition, we have
    $\transla{A\to B} = \forall \ol\rho_1 \ol\rho_2 . \translb{A,\ol\rho_1} \to
    \transla{B, \ol\rho_2} $, where $\ol\rho_1 = \transrowb{\rho_1}{A},\
    \ol\rho_2 = \transrowa{\rho_2}{B}$. By $\tau_1\subtyperow\transla{A\to B}$,
    we have $A' = \translb{A,\ol\rho_1}$, $B'\subtyperow \transla{B, \ol\rho_2}$
    and $\ol\rho = \ol\rho_1 \ol\rho_2$ after $\alpha$-renaming.
    By \tylab{Inst} and environment weakening, we have
    $$
    \typ{\Delta,\ftv{\transl{\Gamma}},\ol\rho_2; \transl{\Gamma}}
    {\transl{M} : \translb{A,\ol R} \to B'}
    $$
    Notice that $\translb{A,\ol R} = \tau_2$. We can then apply \tylab{App} and
    environment weakening, which gives us
    $$
    \typ{\Delta,\ftv{\transl{\Gamma}},\ol\rho_2; \transl{\Gamma}}
    {\transl{M}\,\transl{N} : B'}
    $$
    Finally, by \tylab{Gen}, we have
    $$
    \typ{\Delta,\ftv{\transl{\Gamma}}; \transl{\Gamma}}
    {\transl{M}\,\transl{N} : \forall \ol\rho_2. B'}
    $$
    The condition $\forall \ol\rho_2. B' \subtyperow \transla{B}$ holds
    obviously.

    \item[\tylab{Record}] Our goal follows from the IH and a sequence of
    applications of \tylab{Inst}, \tylab{Record}, and \tylab{Gen} similar to the
    previous cases.
    \item[\tylab{Project}] Our goal follows from the IH and a sequence of
    applications of \tylab{Inst}, \tylab{Project}, and \tylab{Gen} similar to
    the previous cases.
    \item[\tylab{Let}]
    Given the derivation of $\typ{\Delta;\Gamma}{\Let\, x=M\,\In\, N}$, by the IH
    on $\typ{\Delta;\Gamma}{M:A}$, we have
    $$
    \typ{\Delta,\ftv{\transl{\Gamma}};\transl{\Gamma}}{\transl{M} : \tau_1}
    $$
    for some $\tau_1 \subtyperow \transla{A}$.
    By the IH on $\typ{\Delta;\Gamma,x:A}{N:B}$, we have
    $$
    \typ{\Delta,\ftv{\transl{\Gamma}};\transl{\Gamma},x:\transla{A}}{\transl{N} : \tau_2}
    $$
    for some $\tau_2 \subtyperow \transla{B}$.
    By another straightforward induction on the typing derivations, we can show
    that $\typ{\Delta;\Gamma,x:\tau_1}{M:\tau_2}$ implies
    $\typ{\Delta;\Gamma,x:\tau_1'}{M:\tau_2'}$ for $\tau_1'\subtyperow \tau_1$
    and $\tau_2'\subtyperow \tau_2$.
    Thus, we have
    $$
    \typ{\Delta,\ftv{\transl{\Gamma}};\transl{\Gamma},x:\tau_1}{\transl{N} : \tau_2'}
    $$
    for some $\tau_2' \subtyperow \tau_2 \subtyperow \transla{B}$.
    Then, by \tylab{Let}, we have
    $$
    \typ{\Delta,\ftv{\transl{\Gamma}};\transl{\Gamma}}{\Let\,x=\transl{M}\,\In\,\transl{N} : \tau_2'}
    $$
    with $\tau_2' \subtyperow \transla{B}$.
  \end{description}
\end{proof}

\section{Proofs of Non-existence Results}

In this section, we give the proofs of the non-existence results of
\Cref{sec:translations} and \Cref{sec:full-subtyping}.

\subsection{Non-Existence of Type-Only Encodings of \STLCRecSub in \STLCRecRow
            and \STLCVarSub in \STLCVarPre}
\label{app:non-existence-swapping}
\globalRecSubRecRow*

\begin{proof}
  We provide three proofs of this theorem, the first one is based on
  the type preservation property, the second one is based on the
  compositionality of translations, and the third one carefully avoids
  using the type preservation and compositionality. The point of
  multiple proofs is to show that the non-existence of the encoding of
  \STLCRecSub in \STLCRecRow is still true even if we relax the
  condition of type preservation and compositionality, which emphasises
  the necessity of the restrictions in \Cref{sec:prenex-polymorphism}.

\noindent\textsc{Proof 1:}

  We assume that $\Delta = \alpha_0$ and $\Gamma = y : \alpha_0$ when environments are omitted.

  Consider $\Record{}$ and $\Record{\ell = y} \Upcast \Record{}$. By the fact that
  $\transl{-}$ is type-only, we have $\transl{\Record{}} = \Lambda \ol\alpha .
  \Record{}$ and $\transl{\Record{\ell=y} \Upcast \Record{}} = \Lambda \ol\beta .
  \transl{\Record{\ell=y}}\ \ol B = \Lambda \ol\beta . (\Lambda \ol\gamma .
  \Record{\ell=\Lambda \ol\gamma' . y})\ \ol B$. Thus, $\transl{\Record{\ell=y}
  \Upcast \Record{}}$ has type $\forall \ol\alpha' . \Record{\ell : \forall
  \ol\gamma'.\alpha_0}$ for some $\ol\alpha'$.

  By type preservation, the translated results should have the same
  type, which implies $\forall \ol\alpha . \Record{} = \forall
  \ol\alpha' . \Record{\ell : \forall \ol\gamma' . \alpha_0}$. Thus,
  we have the equation $\Record{} = \Record{\ell : \forall \ol\gamma'
    . \alpha_0}$, which leads to a contradiction as the right-hand
  side has an extra label $\ell$ and we do not have presence types to
  remove labels.

  Similarly, we can prove the theorem for variants by considering
  $(\ell_1\,y)^{\Variant{\ell_1:\alpha_0; \ell_2:\alpha_0}}$ and
  $(\ell_1\,y)^{\Variant{\ell_1:\alpha_0}}\Upcast \Variant{\ell_1:\alpha_0;
  \ell_2:\alpha_0}$. The key point is that $\ell_2$ is arbitrarily chosen, so
  for the translation of $(\ell_1\,y)^{\Variant{\ell_1:\alpha_0}}$ we cannot
  guarantee that $\ell_2$ appears in its type, and presence polymorphism does not
  give us the ability to add new labels to row types.

\noindent\textsc{Proof 2:}

  We assume that $\Delta = \alpha_0$ and $\Gamma = y : \alpha_0$ when environments are omitted.

  Consider the function application $M\ N$ where $M = \lambda
  x^{\Record{}}.\Record{}$ and $N = \Record{\ell=y}\Upcast\Record{}$. By the
  type-only property, we have
  $$
  \transl{\lambda x^{\Record{}}.\Record{}} =
    \Lambda \ol\alpha_1. \lambda x^{A_1}. \Lambda \ol\beta_1. \Record{}\ B_1
  $$
  for some $\ol\alpha_1, \ol\beta_1, A_1$ and $B_1$.
  By \textsc{Proof 1}, we have
  $$
  \transl{\Record{\ell=y}\Upcast\Record{}} =
    \Lambda \ol\alpha_2. \Record{\ell=\Lambda \ol\beta_2. y}
  $$
  for some $\ol\alpha_2$ and $\ol\beta_2$.
  Then, by the type-only property, we have
  $$
  \transl{(\lambda x^{\Record{}}.\Record{})\ (\Record{\ell=y}\Upcast\Record{})} =
  \Lambda \ol\alpha. (\transl{\lambda x^{\Record{}}.\Record{}}\ \ol A)\
                     (\Lambda \ol\beta. \transl{\Record{\ell=y}\Upcast\Record{}}\ \ol B)\ \ol C
  $$
  for some $\ol\alpha, \ol\beta, \ol A, \ol B$ and $\ol C$.
  As we only have row polymorphism, the type application of $\ol B$ cannot
  remove the label $\ell$ from the type of $\transl{N}$. Since $\ell$ is
  arbitrarily chosen, it can neither be already in the type of $\transl{M}$. By
  definition, a compositional translation can only use the type information of
  $M$ and $N$, which contains nothing about the label $\ell$. Thus, the label
  $\ell$ can neither be in $\ol A$, which further implies that the $\transl{M\ N}$
  is not well-typed as the \tylab{App} must fail. Contradiction.

\noindent\textsc{Proof 3:}

  Consider three functions $f_1 = \lambda x^{\Record{}} . x$, $f_2 = \lambda
  x^{\Record{}} . \Record{}$, and $g = \lambda
  f^{\Record{}\to\Record{}}.\Record{}$. By the type-only property, we have
  \[\ba{lll}
  \transl{f_1} &= \Lambda \ol\alpha_1 . \lambda x^{A_1} . \Lambda \ol\beta_1 . x\ \ol B_1
  &: \forall\ol\alpha_1 . A_1 \to \forall\ol\beta_1 . A_1'\\
  \transl{f_2} &=
  \Lambda \ol\alpha_2 . \lambda x^{A_2} . \Lambda \ol\beta_2 . \Record{}
  &: \forall\ol\alpha_2 . A_2 \to \forall\ol\beta_2 . \Record{} \\
  \transl{g} &= \Lambda \ol\alpha_3 . \lambda f^{A_3} . \Lambda \ol\beta_3 . \Record{}
  &: \forall\ol\alpha_3 . A_3 \to \forall\ol\beta_3 . \Record{}
  \ea\]
  where $A_1' = A_1''[\ol B_1/\ol\alpha_1']$ and $A_1 = \forall\ol\alpha_1'.A_1''$.

  If there is some variable $\alpha_1' \in \ol\alpha_1$ appears in $A_1$, then
  it must also appear in $A_1'$ as we have no way to remove it by the
  substitution $[\ol B_1/\ol\alpha']$. Thus, $A_3$ should be of shape $\forall
  \ol\alpha . A \to \forall\ol\beta. A'$ where $A'$ contains some variable
  $\alpha'\in\ol\alpha$. However, this contradicts with the fact that $g$ can be
  applied to $f_2$, because the type $\Record{}$ in the type of $\transl{f_2}$
  cannot contain any variable in $\ol\alpha_2$. Hence, we can conclude that
  $A_1$ cannot contain any variable in $\ol\alpha_1$, which will lead to
  contradiction when we consider the translation of $f_1\
  (\Record{\ell=1}\Upcast\Record{})$ because we can neither add the label $\ell$
  in the type $A_1$, nor remove it in the type of
  $\transl{\Record{\ell=1}\Upcast\Record{}}$.

\end{proof}

\subsection{Non-Existence of Type-Only Encodings of \STLCVarSubCo in \STLCVarRowPre}
\label{app:non-existence-stlcvarsubco}
\globalVarSubCo*

\begin{proof}
We assume that $\Delta = \alpha_0$ and $\Gamma = y : \alpha_0$ when environments are omitted.
For simplicity, we omit the type of labels in variant types if it is $\alpha_0$.

By the fact that $\transl{-}$ is type-only, we have:
\begin{itemize}
\item
  $(\ell\ y)^{\Variant{\ell}}$ is translated to
  $\Lambda \ol\alpha. (\ell\ (\Lambda \ol\beta. y))^{\Variant{R}}$
  where $(\ell:\forall \ol\beta . \alpha_0) \in R$.
  By type preservation, we have $\transl{\Variant{\ell}} = \forall \ol\alpha . \Variant{R}$.
\item
  $(\ell\ y)^{\Variant{\ell}} \Upcast \Variant{\ell; \ell'}$ is translated to
  $\Lambda \ol\tau. \transl{(\ell\ y)^{\Variant{\ell}}}\ \ol T
  =\Lambda \ol\tau. (\Lambda \ol\alpha. (\ell\ (\Lambda \ol\beta. y))^{\Variant{R}})\ \ol T$
  where $(\ell:\forall\beta . \alpha_0) \in R$.
  By type preservation, we have $\transl{\Variant{\ell;\ell'}}
    = \refa{\forall \ol\tau\ \ol\alpha_2' . \Variant{R}[\ol T/\ol\alpha_1']}$
    where $\ol\alpha = \ol\alpha_1'\ \ol\alpha_2'$.
\item
  $(\ell'\ y)^{\Variant{\ell; \ell'}}$ is translated to
  $\Lambda \ol\alpha''. (\ell'\ (\Lambda \ol\beta''. y))^{\Variant{R''}}$
  where $\ell' \in R''$.
  By symmetry, we also have $\ell \in R''$.
  By type preservation, we have $\transl{\Variant{\ell; \ell'}}
    = \refb{\forall \ol\alpha'' . \Variant{R''}}$.
\end{itemize}

By the fact that $\refa{} = \refb{}$ and $\ell'$ can be an arbitrary label, we
can conclude that $R$ has a row variable $\rho_R$ bound in $\ol\alpha_1'$ which is
instantiated to the $\ell'$ label in $R'$ by the substitution $[\ol A/\ol
\alpha_1']$. Thus, we have \refc{$R = (\ell:\forall\ol\beta .
\alpha_0);\dots;\rho_R$ where $\rho_R \in \ol\alpha$}.

Then, consider a nested variant $M =
(\ell\,(\ell\,y)^{\Variant{\ell}})^{\Variant{\ell: \Variant{\ell}}}$.
Because $\transl{-}$ is type-only, we have
$$
  \transl{M} = \Lambda \ol\alpha' . (\ell\
  (\Lambda \ol\beta'.  (\Lambda \ol\alpha. (\ell\ (\Lambda \ol\beta. y))^{\Variant{R}})
  \ \ol A))^{\Variant{R'}}
$$
By \refc{}, $\transl{M}$ has type $\forall \ol\alpha'. \Variant{R'}
= \forall\ol\alpha'. \Variant{(\ell: \forall \ol\beta'\ \ol\alpha_2. \Variant{R}[\ol A/\ol\alpha_1]);
                              \dots}$,
where $\ol\alpha = \ol\alpha_1\ \ol\alpha_2$ and $\rho_R \in \ol\alpha$.

We proceed by showing the contradiction that $\rho_R$ can neither be in
$\alpha_1$ nor $\alpha_2$.

\begin{itemize}
  \item $\rho_R \in \ol\alpha_2$.
  Consider $M' = (\ell\,(\ell\,y)^{\Variant{\ell;\ell'}})^{\Variant{\ell:
  \Variant{\ell;\ell'}}}$ of type $\Variant{\ell:\Variant{\ell;\ell'}}$. By an
  analysis similar to $M$, it is easy to show that $\transl{M'}$ has type $
  \forall\ol\mu. \Variant{(\ell: \forall \ol\nu. \Variant{R_1});\dots}$
  where $\ell \in R_1$ and $\ell' \in R_1$.

  Then, consider $M\Upcast \Variant{\ell:\Variant{\ell;\ell'}}$ of the same type
  $\Variant{\ell:\Variant{\ell;\ell'}}$ as $M'$ which is translated to $\Lambda
  \ol\gamma. \transl{M}\ \ol B$. By type preservation, the translation of $M'$ and
  $M\Upcast \Variant{\ell:\Variant{\ell;\ell'}}$ should have the same type,
  which means $R$ should contain label $\ell'$ after the type application of
  $B$. However, because $\rho_R \in \ol\alpha_2$, we cannot instantiate $\rho_R$
  to contain $\ell'$. Besides, because $\ell'$ is arbitrarily chosen, it cannot
  already exist in $R$. Hence, $\rho_R \not\in \ol\alpha_2$.

  \item $\rho_R \in \ol\alpha_1$.
  Consider $\Case\ M\ \{\ell\ x \mapsto x\Upcast \Variant{\ell;\ell'}\}$ of type
  $\Variant{\ell;\ell'}$. By the type-only condition, it is translated to
  \refd{$\Lambda \ol\gamma. \Case\ (\transl{M}\ \ol C) \{\ell\ x\mapsto \Lambda
  \ol\delta. x\ \ol D\}$}. By \refb{} we have $\transl{\Variant{\ell;\ell'}} =
  \forall \ol\alpha''. \Variant{R''}$ where $\ell \in R''$ and $\ell' \in R''$.
  However, for \refd{}, by the fact that $\rho_R \in \ol\alpha_1$ and
  $\ol\alpha_1$ are substituted by $\ol A$, the new row variable of the inner
  variant of $M$ can only be bound in $\ol\alpha'$. Thus, in the case clause of
  $\ell$, we cannot extend the variant type to contain $\ell'$ by type
  application of $\ol D$.
  Besides, because $\ell'$ is arbitrarily chosen and the translation
  is compositional, it can neither be already in the variant type or
  be introduced by the type application of $\ol C$.
  Hence, $\rho_R \not\in \ol\alpha_1$.
\end{itemize}

Finally, by contradiction, the translation $\transl{-}$ does not exist.

\end{proof}

\subsection{Non-Existence of Type-Only Encodings of Full Subtyping}
\label{app:non-existence-full}

\globalSubFull*

\begin{proof}

Consider two functions $f_1 = \lambda x^{\Record{}} . x$ and $f_2 = \lambda
x^{\Record{}} . \Record{}$ of the same type $\Record{} \to \Record{}$. By the
type-only property, we have
\begin{align*}
\transl{f_1} &= \Lambda \ol\alpha_1 . \lambda
x^{A_1} . \Lambda \ol\beta_1 . x\ \ol B_1 \\
\transl{f_2} &= \Lambda
\ol\alpha_2 . \lambda x^{A_2} . \Lambda \ol\beta_2 . \transl{\Record{}}\ \ol B_2
= \Lambda \ol\alpha_2 . \lambda x^{A_2} . \Lambda \ol\beta_2 . (\Lambda
\ol{\gamma} . \Record{})\ \ol B_2
\end{align*}
By type preservation, they have the same type, which implies $x\ \ol
B_1$ and $(\Lambda \ol{\gamma} . \Record{})\ \ol B_2$ have the same
type.
We can further conclude that $A_1$ must be able to be instantiated to
the empty record type $\Record{}$.
Thus, the only way to have type variables bound by
$\Lambda\ol\alpha_1$ in $A_1$ is to put them in the types of labels
which are instantiated to be absent by the type application $x\ \ol
B_1$.

Then, consider another two functions $g_1 = f_1 \Upcast
(\Record{\ell:\Record{}} \to \Record{})$ and $g_2 = \lambda
x^{\Record{\ell:\Record{}}} . (x.\ell)$ of the same type
$\Record{\ell:\Record{}} \to \Record{}$. By the type-only property, we
have
\begin{align*}
\transl{g_1} &= \Lambda \ol\alpha . \transl{f_1}\ \ol A = \Lambda \ol\alpha .
  (\Lambda \ol\alpha_1 . \lambda x^{A_1} . \Lambda \ol\beta_1 . x\ \ol B_1)\ \ol A \\
\transl{g_2} &= \Lambda \ol\alpha' . \lambda x^{A'}.\Lambda \ol\beta' . \transl{x.\ell}\ \ol B'
= \Lambda \ol\alpha' . \lambda x^{A'}.\Lambda \ol\beta' . (\Lambda \ol\gamma'. (x\ \ol C).\ell\ \ol D)\ \ol B'
\end{align*}
By type preservation, $\transl{g_1}$ and $\transl{g_2}$ have the same
type. The $(x\ \ol C).l$ in $\transl{g_2}$ implies that $x$ has a
polymorphic record type with label $\ell$.
Because $\ell$ is arbitrarily chosen, the only way to introduce $\ell$
in the parameter type of $\transl{g_1}$ is by the type application of
$\ol A$.
However, we also have that type variables in $\ol \alpha_1$ can only
appear in the types of labels in $A_1$, which means we cannot
instantiate $A_1$ to be a polymorphic record type with the label
$\ell$ by the type application of $\ol A$.
Contradiction.

\end{proof}

\end{document}